\documentclass[a4paper,abstracton,12pt]{article}

\usepackage[hmargin=3.5cm,vmargin=3cm]{geometry}

\usepackage[komastyle,ilines]{scrpage2}
\setheadtopline{0pt} 
\setheadsepline{1pt}
\ohead{Pierre Degond \& Laurent Navoret}
\ihead{A multi-layer model for self-propelled interacting disks}
 \ifoot{}
 \ofoot{}

\usepackage{amsfonts} 
\usepackage{amsmath,amsthm,amssymb} 

\usepackage{cite} 

\usepackage{graphicx} 

\usepackage{hyperref}

\def\R{\mathbb{R}}

\def\S{\mathbb{S}}

\def\Z{\mathbb{Z}}

\def\C{\mathbb{C}}

\def\sgn{\text{sgn}}

\def\eps{\varepsilon}

\newtheoremstyle{theorem}
  {}
  {}
  {\sffamily}
  {}
  {\bfseries}
  {.}
  { }
  {}

\theoremstyle{theorem}
\newtheorem{definition}{Definition}
\newtheorem{lem}{Lemma}
\newtheorem{theorem}[lem]{Theorem} 
\newtheorem{proposition}[lem]{Proposition}  

\theoremstyle{remark}

\title{A multi-layer model for self-propelled disks\\  interacting through alignment and volume exclusion}
\author{Pierre Degond\footnote{ Department of Mathematics, Imperial College London,
London SW7 2AZ, United Kingdom,
email: pdegond@imperial.ac.uk}, Laurent Navoret\footnote{Institut de Recherche Math\'ematique Avanc\'ee, UMR 7501, Universit\'e de Strasbourg et CNRS, 7 rue Ren\'e Descartes, 67000 Strasbourg, France, email: laurent.navoret@math.unistra.fr}\ \footnote{INRIA Nancy - Grand Est / IRMA - TONUS}} 
\date{}

\begin{document}

\maketitle

\begin{abstract}
We present an individual-based model describing disk-like self-propelled particles moving inside parallel planes. The disk directions of motion follow alignment rules inside each layer. Additionally, the disks are subject to interactions with those of the neighboring layers arising from volume exclusion constraints. These interactions affect the disk inclinations with respect to the plane of motion. We formally derive a macroscopic model composed of planar Self-Organized Hydrodynamic (SOH) models describing the transport of mass and evolution of mean direction of motion of the disks in each plane, supplemented with transport equations for the mean disk inclination. These planar models are coupled due to the interactions with the neighboring planes.  Numerical comparisons between the individual-based and macroscopic models are carried out. These models could be applicable, for instance, to describe sperm-cell collective dynamics. 
\end{abstract}



\section{Introduction}

Collective motion in systems of self-propelled particles is the subject of a vast literature. How collective motion emerges from the underlying local interactions between the agents is still poorly understood. The interactions are either of cognitive nature (such as in birds, mammals \cite{2012_Vicsek}) and/or are mediated by a surrounding fluid (such as in swimming bacteria, sperm cells, etc. \cite{2011_Koch}). There are several competing strategies to model collective dynamics. One strategy relies on individual-based models (e.g. \cite{Aoki_BullJapSocSciFish92,Chate_etal_PRE08,Chuang_etal_PhysicaD07,Couzin_etal_JTB02,Cucker_Smale_IEEETransAutCont07,Henkes_etal_PRE11,Mogilner_etal_JMB03,Motsch_Tadmor_JSP11,Szabo_PRL06}) that describe how the position and velocity of each individual evolves in the course of time. Another strategy relies on continuum models (e.g. \cite{Baskaran_Marchetti_PRL10,Bertin_etal_JPA09,Ratushnaya_etal_PhysicaA07,Toner_etal_AnnPhys05}) describing the system by locally averaged quantities such as the mean density, mean velocity, etc. An intermediate category of models consist of kinetic models\cite{Fetecau_M3AS11} which describe the individual motions in a probabilistic way. These different types of models can be connected one to each other as kinetic models can be seen as resulting from an infinite particle number limit of individual-based models, while macroscopic models follow from a hydrodynamic or diffusive rescaling of the kinetic models and subsequently passing to the limit of a large rescaling factor. We refer to Ref.~\cite{2008_ContinuumLimit_DM} for an illustration of this methodology in the case of the Vicsek model (see also below). 

In this paper, we are concerned with finding a suitable modeling framework for the three dimensional motion of spermatozoa in the seminal plasma. As there are evidence that sperm-cell motion in the most common experiments is mostly planar\cite{Plouaboue_2015}, we propose a multi-layer model where the motion of sperm-cells is planar and sperm-cells may interact with other sperm-cells of the same layer or of the two neighboring layers. Sperm-cell concentration in raw sperm is incredibly high (as large as $5 \, 10^9$ cm$^{-3}$) and the use of individual-based models to reproduce real sperm-cell experiments is intractable. It is therefore necessary to derive a macroscopic model describing the collective motion of the cells within each layer and their interactions with the neighboring layers.

Spermatozoa can be assimilated to two-dimensional discs. Indeed, as regards the occupied volume, flagella can be neglected and spermatozoa reduced to their head. Although the heads resemble flat ellipsoids, we simply model them as infinitely thin flat discs. The acting flagellum produces almost constant propulsion. Thus, it is a good approximation to suppose that all sperm-cells move with the same constant speed and that only the velocity direction is subject to changes. Finally, each disk possesses some inclination with respect to its plane of motion, measured by an inclination angle. Therefore, the position, velocity and attitude of each disk can be described by the position of its center of mass, its velocity direction and its inclination angle. 

Interactions between sperm-cells are mostly hydrodynamic interactions (due to the perturbation of the fluid velocity induced by the motion of the cells) and volume exclusion (or steric) interactions (due to the impossibility that two sperm-cells overlap). Modeling hydrodynamic and steric interactions within dense suspensions of active particles is a difficult subject. However, it has been shown in \cite{2006_Peruani} that for self-propelled elongated particles, such interactions simply result in local alignment of the particles with their neighbors. Therefore, we assume that all these interactions can be lumped into a local alignment interaction with the close neighbors. 

As already mentioned, we assume that particle motion is planar and takes place in parallel two-dimensional layers. Each particle belongs to one layer for all times without the possibility to change its layer. For the reasons outlined above, spermatozoa interact inside these 2D layers by alignment of both their velocity and inclination with those of their close neighbors. Specifically, we consider the time-continuous version of the Vicsek microscopic model as proposed in Ref. \cite{2008_ContinuumLimit_DM}, where each particle tends to align with the mean direction of its neighbors up to a small Brownian perturbation. The original model was proposed by Vicsek et al.\cite{1995_Vicsek} and several variants have been proposed in \cite{2015_Degond,2013_DegondLiu,2012_Frouvelle,2013_Navoret}. The inclination alignment dynamics follows a similar rule with the exception that the interaction is nematic (i.e. two inclination angles differing by a multiple of $\pi$ lead to the same disk attitude). The combined alignment dynamics in the velocity-inclination variables thus differs from the 3D Vicsek dynamics (specifically, in the 3D Vicsek, velocity belongs to a 2D-sphere whereas here the pair (velocity, inclination) belongs to a 2D torus).

The different layers also interact via the volume-exclusion constraint. Indeed, due the inclination of the disks, the spermatozoa of one layer exert a friction on the spermatozoa of the nearby layers. This interaction results in increasing or decreasing the inclination of the spermatozoa in these layers. A given layer thus acts on the neighboring ones in a similar way as the wind does on plant canopies \cite{2004_Doare_PlantInteract}. The involved mechanical forces depend on the geometric configuration of the discs: it thus depends on their respective inclinations and also on their respective velocities. It results in alignment of velocities and  repulsion of inclination angles. Layers are thus coupled and this coupling depends on the so-called overlap function that quantifies the distance between layers.

We then consider a mean-field kinetic version of the model. This equation provides the time evolution of the distribution function in phase space (position, velocity and inclination) and takes the form of a Fokker-Planck equation. Here, it is formally derived in the limit of an infinite number of interacting particles. Although the mathematical validity of the mean-field limit has been established for the Vicsek model\cite{2012_Bolley}, it is still open for the present model and our result so far is only formal. We perform a spatio-temporal hydrodynamic rescaling of the mean-field kinetic model, considering that the intra-layer interaction scales are much smaller than those of the inter-layer interactions and that the latter occur at the same scales as the macroscopic evolution of the system. There results a singularly perturbed Fokker-Planck equation, involving a small parameter $\varepsilon$ measuring the ratio of the small (microscopic) scale to the large (macroscopic) one. 

The macroscopic description of the system is found by letting $\varepsilon$ to zero in the singularly perturbed mean-field kinetic model. We first need to find the equilibria associated to the Fokker-Planck operator. We show that these are given by products of von Mises distributions in the velocity and inclination angles respectively. von Mises distributions are the natural analog of Gaussian distributions for probabilities on the sphere. We then need to integrate the equation against the collisional invariants.  However, as noticed in Ref. \cite{2008_ContinuumLimit_DM}, only mass is a collisional invariant and we are thus lacking two more collisional invariants to obtain the dynamics on the velocity and inclination. Following Ref. \cite{2008_ContinuumLimit_DM}, we have to introduce the ``generalized collisional invariants'' of this operator: these are collisional invariants valid only on functions with prescribed mean velocity and mean inclination. We then are able to derive the macroscopic model.

The obtained model consists of a continuity equation for the density $\rho$ and two evolution equations for the mean velocity angle $\bar\varphi \in [0,2\pi]$  and mean inclination angle $\bar\theta \in [0,\pi]$. All these quantities are indexed by  $h$ corresponding to the $h^{\text{th}}$ layer. The model is written: 
\begin{align}
&\partial_t \rho_h + c c_1 \, \nabla_x \cdot \left( \rho_h V(\bar \varphi_h) \right)= 0, 
\label{eq:mass_intro} \\
&\rho_h \big( \partial_t \bar \varphi_h +  c c_2 (V(\bar \varphi_h) \cdot\nabla_x) \bar \varphi_h \big) + \frac{c}{\kappa_1} V(\bar \varphi_h)^\bot \cdot\nabla_x \rho_h\nonumber\\
&\qquad= \frac{\nu\beta' }{c_3}\,  V(\bar \varphi_h)^\perp\cdot   \sum_{k,\, k-h = \pm 1}  \langle gM_2M_2\rangle(\bar\theta,\bar\theta_k)\, c_1 \rho_k V(\bar\varphi_k),
\label{eq:phi_intro} \\
&\rho_h( \partial_t \bar \theta_h +  c c_1 (V(\bar \varphi_h) \cdot\nabla_x) \bar \theta_h )\nonumber\\
&\quad= \frac{ \mu'}{c_4}\, \, (c_1  \rho  V(\bar \varphi_h))^\perp\cdot \sum_{k,\, k-h = \pm 1} \sgn(k-h) \langle gM_2M_2\partial_\theta I_2\rangle(\bar\theta,\bar\theta_k)\,  c_1 \rho_k V(\bar \varphi_k),
\label{eq:theta_intro}
\end{align}
where $V(\bar\varphi) = (\cos\bar\varphi,\sin\bar\varphi)^T$ denotes the velocity vector. The constants will be defined further. The left-hand sides of \eqref{eq:mass_intro}-\eqref{eq:phi_intro} form the SOH (Self-Organized Hydrodynamics) model describing the Vicsek dynamics at the macroscopic level \cite{2008_ContinuumLimit_DM}. They respectively account for the conservation of mass and convection of the mean velocity angle. 
The left-hand side of \eqref{eq:theta_intro} describes the advection of the inclination with the same advection velocity as for the mass. Finally, the right-hand sides of \eqref{eq:phi_intro}-\eqref{eq:theta_intro} describe the inter-layer interactions. The obtained model is thus an hyperbolic system (like the SOH model) with source terms that couple the velocity and inclination dynamics.

We remark that, once the velocities of all the different layers are co-linear, the source terms vanish. The model thus simplifies into a superposition of standard SOH models in each layer. However, before reaching an equilibrium, the interplay between inclination and velocity crucially determines which equilibrium velocities and inclinations will be attained. 

To validate the macroscopic model, numerical simulations are performed and compared with those of the particle model. With this aim, we adapt the numerical relaxation method of \cite{2011_MotschNavoret} designed for the Vicsek model. The numerical simulations show that the macroscopic model captures the velocity alignment between the different layers quite well. However, some differences in the inclination dynamics appear. Indeed, some transient ``meta-stable'' configurations arise during the course of time. They are more rapidly left away by the microscopic dynamics than by the macroscopic ones, probably because of the stochastic fluctuations associated with the microscopic dynamics. This effect due to the finiteness of the particle number in the microscopic dynamics could probably be reproduced by including a stochastic term in the macroscopic model such as the one proposed in Ref. \cite{Toner_etal_AnnPhys05}. 

The outline of this article is as follows. In Section \ref{sec:micro_model}, we introduce the microscopic model. In Section \ref{sec:meanfieldmodel}, we present the mean-field limit and the hydrodynamic scaling. In Section \ref{sec:macro_model}, we provide the derivation of the macroscopic model. In Section \ref{section:num_experiment}, we compare the microscopic and macroscopic dynamics on several test-cases. A discussion of the results is provided in Section \ref{sec:conclu}. \ref{appendix:nematic} provides complements on nematic alignment modeling, \ref{annex:coeff_macro} gives some properties of the coefficients of the macroscopic model and \ref{annex:num} describes the numerical schemes used for the simulations.

\section{Microscopic model}
\label{sec:micro_model}

Spermatozoa are represented by discs of radius $R$ moving in different layers. In first approximation, the layers, indexed by $\Z$, are copies of the $\R^2$ plane: layer $h$ denotes the plane $\left\{(x,y,z) \in \R^3,\ z = hd\right\}$ where $d > 0$ denotes the inter-distance between the layers. With no interactions between layers, spermatozoa are supposed to be orthogonal to the layer and follow the Vicsek dynamics. However, if the layer inter-distance is lower than the disc radius $R$, spermatozoa of layer $h$ will exert a force on the spermatozoa of the neighboring layers $h-1$ and $h+1$: they may force them to incline (with respect to the layer plane).

\begin{figure}[h]
\centering
\includegraphics[width=\textwidth]{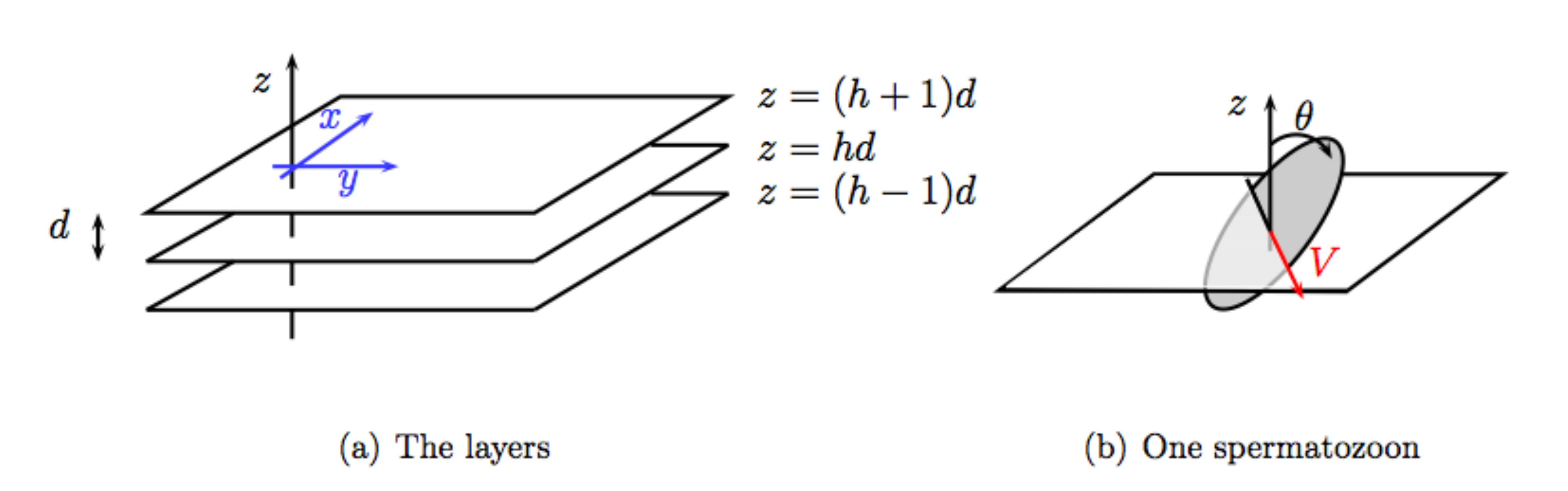}
\caption{Spermatozoa in layers}
\end{figure}

We consider $N$ discs in $\R^3$ labeled by $k \in \left\{1,\ldots,N\right\}$: each disc  is contained into one layer and thus disc $k$ has a permanent altitude $h_k \in \Z$. The two-dimensional movement into the layer is described by the position of its center of mass $X_k(t) \in \R^2$ and the velocity orientation of its center of mass $V_k(t) \in \S^1$: we indeed consider that all particles move at the same speed $c > 0$. The velocity of the particle is thus given by $cV_k$. We introduce the angle $\varphi_k(t)$ of $V_k(t)$ with respect to a reference axis, so that $V_k(t) = V(\varphi_k(t))$ with $V(\varphi) = (\cos \varphi, \sin \varphi)$.  

Concerning the configuration of the disc in space,  we suppose that the disc moves in one direction contained in its plane: the orientation $V_k(t) \in \S^1$ belongs to the disc plane. The angle of this plane with respect to the $z$-axis is denoted $\theta_k(t) \in \R/[0,\pi]$. An angle $\theta_k = 0$ or $\pi$ means that the disc is perpendicular to the plane while an angle $\theta_k = \pm \pi/2$ means that the whole disc lies in the layer plane.  
\subsection{Dynamics for the centers of masses.} The centers of masses follow a Vicsek-like dynamics  as introduced in Ref. \cite{2008_ContinuumLimit_DM}. The dynamics of the positions and the velocities are given by the following equations:
\begin{align}
&\frac{dX_k}{dt}(t) = cV(\varphi_k(t)),\\
&d\varphi_k(t)  = -\nu \sin(\varphi_k(t) - \bar \varphi_k^{\text{tot}}(t)) dt + \sqrt{2D} \, dB_t^{\varphi,k} \quad \mbox{ modulo } 2 \pi.\label{eq:phi}
\end{align}
Two dynamics are in competition: alignment and diffusion. Each particle of a given layer $h_k$ tends to align with a direction $V(\bar \varphi_k^{\text{tot}}(t))$ with an intensity $\nu$ supposed constant. The direction $V(\bar \varphi_k^{\text{tot}}(t))$ is defined as a weighted mean direction of the neighbors particle $k$ in layers $h_k$, $h_k\pm1$  within the disc of radius $R_1$:
\begin{align}
&V(\bar \varphi_k^{\text{tot}}(t)) = \frac{J^{\varphi}_{k}(t)+ \beta J^{\varphi,\text{w},\text{nb}}_{k}(t)}{|J^{\varphi}_{k}(t)+ \beta J^{\varphi,\text{w},\text{nb}}_k(t)|}, \qquad J^{\varphi}_k(t) = \sum_{j, h_j = h_k, |X_j(t) - X_k(t)| \leqslant R_1} V(\varphi_j(t)),\\
&J^{\varphi,\text{w},\text{nb}}_k(t) = \sum_{j, h_j = h_k\pm 1, |X_j(t) - X_k(t)| \leqslant R_1} g(\theta_j(t),\theta_k(t))\,V(\varphi_j(t)),
\end{align}
where $J^{\varphi}_k(t)$ denotes the contribution of neighbors belonging to the same layer $h_k$, $J^{\varphi,\text{w},\text{nb}}_k(t)$ denotes the contribution of neighbors belonging to layers $h_k\pm1$ and $\beta$ quantifies their relative involvement. Superscripts ``nb" means neighboring layers and ``w'' means weighted. Indeed,  due to steric constraints within layers, directions of neighbors of layers $h_k\pm1$ are weighted according to their inclination. Supposing distance $h$ between layers is smaller than twice the particle radius, $h < 2R$, we define the overlap function $g$ by:
\begin{equation}
g(\theta_j,\theta_k) = \frac{1}{2R-h}\big(R\big(|\cos\theta_j| + |\cos\theta_k|\big) - h\big)_+.
\label{eq:weight}
\end{equation}
This function quantifies the overlapping area of two discs in the $\hat z$ direction.
The particle direction is also submitted to a Brownian motion $B_t^{\varphi,k}$ with diffusion coefficient $D > 0$. We here neglect congestion forces and we also neglect alignment between discs of different layers. 

\subsection{Dynamics of the disk orientations.}  The disc angle dynamics follows the following torque balance equation in the over-damped 
regime\footnote{The torque balance equation reads: $I \frac{d^2\theta_k}{dt^2} +C\frac{d\theta_k}{dt} + \tilde K\sin(2(\theta_k - \bar \theta_k)) = T_k$, where  $I$ is the moment of inertia of the disc with respect to its longitudinal axis (parallel to $V_k$) and $C$ is the dissipation coefficient. In the over-damped regime, $I$ is negligible and we recover \eqref{eq:thetadyn} with $K = \tilde K /C$.}:
\begin{equation}
d\theta_k(t) = (-K\sin(2(\theta_k(t) - \bar \theta_k(t))) + T_k(t))dt + \sqrt{2\delta}\, dB_t^{\theta,k} \quad \mbox{ modulo } \pi, 
\label{eq:thetadyn}
\end{equation}
where the alignment dynamics inside each layer is in competition with steric forces between layers. The factor $2$ inside the sine function takes into account that the inclination interaction between disks is a nematic one, i.e. orientations $\theta_k = \bar \theta_k$ or 
$\theta_k = - \bar \theta_k$ are equivalent. In this way, Eq. (\ref{eq:thetadyn}) preserves the fact that $\theta_k$ is defined modulo $\pi$. $K$ is the rotational stiffness and $\delta > 0$ is a diffusion coefficient.

The inclination of each particle in layer $h_k$ tends to align with the mean inclination $\bar \theta_k$ of neighboring particles \textit{of the same layer} in the disc of radius $R_2$. The mean inclination is defined through the following average, which corresponds to a nematic alignment mechanism (see \ref{appendix:nematic}):
\begin{equation}
\label{eq:nematic_alignment}
e^{2 i \bar \theta_k(t)} = \frac{J^\theta_k (t)}{|J^\theta_k(t)|}, \qquad J^\theta_k(t) = \sum_{j, \, h_j = h_k, \, |X_j(t) - X_k(t)| \leqslant R_2} e^{2 i \theta_j(t)},
\end{equation}
where $\theta_k(t)$ is defined modulo $\pi$. 

Particles on neighboring layers $h_k \pm 1$ exert a steric force on particles of layer $h_k$. $T_k(t)$ is sum of the weighted torques $T_{kj}(t)$ (with respect to the longitudinal axis) of the forces exerted by the discs:
\begin{align}
&T_k(t) = \sum_{j, \, h_j = h_k \pm 1, \,  |X_j(t) - X_k(t)| \leqslant R_3} g(\theta_j(t),\theta_k(t))\,T_{kj}(t),\\
&T_{kj}(t) = \mu R \, \sgn(h_j - h_j) \, (V_k(t) \times V_j(t))\cdot \hat z,
\end{align}
where $\hat z$ is the unit vector in the vertical direction and $g(\theta_j,\theta_k)$ is the weight defined in \eqref{eq:weight}. Indeed, the force $F_{kj}$ exerted by disc $j$ on disc $k$ is supposed to be the projection of the velocity direction $V_j$ on the orthogonal plane to $V_k$: 
\begin{equation*}
F_{kj} (t)= \mu \, P_{V_k^\perp(t)} V_j(t),
\end{equation*}
where $\mu$ is a mobility coefficient. Then the torque $T_{kj}$ of the force $F_{kj}$ with respect to the longitudinal axis of the disc is given by:
\begin{equation*}
T_{kj}(t) = [ (X_k(t) - X_j(t))\times F_{kj}(t)]\cdot V_k(t)
\end{equation*}
and using the approximation $(X_k(t) - X_j(t)) \approx R\, \sgn(h_k-h_j)\hat{z}$, we obtain:
\begin{eqnarray*}
T_{kj} (t) &\approx& \mu R\, \sgn(h_k-h_j) \, [\hat z \times  P_{V_k^\perp(t)} V_j(t) ] \cdot V_k(t) \\
&\approx& \mu R\, \sgn(h_k-h_j) \, [P_{V_k^\perp(t)} V_j(t) \times V_k(t) ] \cdot  \hat z \\
&\approx& \mu R\, \sgn(h_k-h_j)\,  [V_j(t) \times V_k(t) ] \cdot  \hat z 
\end{eqnarray*}

\section{Mean-field kinetic model and rescaling}
\label{sec:meanfieldmodel}
\subsection{Mean-field kinetic model} 

We introduce the distribution function in phase space: $f(x,\varphi, \theta,h,t)$, where $f$ is $2 \pi$-periodic in $\varphi$ and $\pi$-periodic in $\theta$. 
Note that $h$ is still a discrete parameter numbering the layers. The distribution function satisfies the mean-field kinetic model:
\begin{align*}
&\partial_t f + c\nabla_x\cdot(V(\varphi)f) =  
-\partial_\theta \big( (-K\sin(2(\theta - \bar \theta_f)) + T_f) f \big)+ \delta\partial_\theta^2 f
\\
& \hspace{6cm} - \partial_\varphi \big(- \nu \sin(\varphi - \bar \varphi_f^{\text{tot}}) f \big) +  D\partial_\varphi^2 f,
\end{align*}
where $\bar \theta_f = \bar \theta_f (x,h,t)$ and $\bar \varphi_f^{\text{tot}} = \bar \varphi_f^{\text{tot}} (x,h,t)$ are defined by:
\begin{align*}
&e^{ 2 i \bar \theta_f (x,h,t)} = \frac{J^\theta_{R_2,f}(x,h,t)}{|J^\theta_{R_2,f}(x,h,t)|},\\ 
&J^\theta_{R_2,f}(x,h,t) = \int_{\begin{subarray}{l}  \theta \in [0,\pi], \varphi \in [0,2\pi],\\y \in \R^2, |y-x| \leqslant R_2\end{subarray}} e^{2 i \theta} f(y,\varphi,\theta,h,t) \, dy \, d \varphi \, d\theta,
\end{align*}
and
\begin{align*}
&V(\bar \varphi_{f}^{\text{tot}}(x,\theta,h,t)) = \frac{J^\varphi_{R_1,f}(x,h,t)+\beta J^{\varphi,\text{w},\text{nb}}_{R_1,f}(x,\theta,h,t)}{|J^\varphi_{R_1,f}(x,h,t)+\beta J^{\varphi,\text{w},\text{nb}}_{R_1,f}(x,\theta,h,t)|},\\ 
&J^{\varphi}_{R_1,f}(x,h,t) = \int_{\begin{subarray}{l}  \theta \in [0,\pi], \varphi \in [0,2\pi],\\ y \in \R^2, |y-x| \leqslant R_1\end{subarray}} V(\varphi) f(y,\varphi,\theta,h,t) \, dy \, d \varphi \, d\theta,\\
&J^{\varphi,\text{w},\text{nb}}_{R_1,f}(x,\theta,h,t) = \sum_{k,\, k-h = \pm 1}J^{\varphi,\text{w}}_{R_1,f}(x,\theta,k,t),\\
&J^{\varphi,\text{w}}_{R_1,f}(x,\theta,h,t) =\int_{\begin{subarray}{l}  \theta \in [0,\pi], \varphi \in [0,2\pi],\\ y \in \R^2, |y-x| \leqslant R_1\end{subarray}} g(\theta,\theta')V(\varphi) f(y,\varphi,\theta',h,t) \, dy \, d \varphi \, d\theta'.
\end{align*}
The torque $T_{f} = T_f(x,h,t)$ is given by:
\begin{align*}
&T_{f}(x,\varphi,\theta, h,t) = \left[N_{f}(x,\theta,h,t) \times V(\varphi)\right] \cdot \hat z,\\
&N_{f}(x,\theta,h,t) = \mu \, R \, \sum_{k,\, k-h = \pm 1} \sgn(k-h) \, J^{\varphi,\text{w}}_{R_3,f}(x,\theta,k,t),
\end{align*}
and depends on neighboring layers $h-1$ and $h+1$. 

\subsection{Hydrodynamic rescaling} 

We then perform a hydrodynamic rescaling to look at the large time and space scale dynamics. The hydrodynamic rescaling consists in introducing macroscopic variables in space and time: $x' = \eps x$, $t' = \eps t$, with $\eps \ll 1$. After dropping the tildes, the kinetic distribution $f^\eps(x',v,w,h,t') = f(x,v,w,h,t)$ satisfies the following equation:
\begin{align}
&\varepsilon \big( \partial_t f^\varepsilon + c\nabla_x\cdot(V(\varphi)f^\varepsilon) \big) =  -\partial_\theta \big( (-K\sin(2(\theta - \bar \theta^\varepsilon_{f^\varepsilon} )) + T^\varepsilon_{f^\varepsilon}) f^\varepsilon \big)+ \delta\partial_\theta^2 f^\varepsilon \nonumber \\
& \hspace{5cm} 
+ \partial_\varphi \big( \nu \sin(\varphi - \bar \varphi^{\text{tot}\, \varepsilon}_{f^\varepsilon}) f^\varepsilon \big) +  D\partial_\varphi^2 f^\varepsilon,
\label{Eq:kin_eps}
\end{align}
where the rescaled $\bar \theta_f^\varepsilon (x,h,t)$, $\bar \varphi^{\text{tot}\,\varepsilon}_f (x,h,t)$, $T^\varepsilon_{f}(x,h,t)$ are defined by:
\begin{align}
&e^{2 i \bar \theta^\varepsilon_f (x,h,t)} = \frac{J^\theta_{\varepsilon R_2,f}(x,h,t)}{|J^\theta_{\varepsilon R_2,f}(x,h,t)|}, \label{Eq:Fmoment_eps_0}
\\
&V(\bar \varphi^{\text{tot}\,\varepsilon}_{f}(x,\theta,h,t)) = \frac{J^{\varphi}_{\varepsilon R_1,f}(x,h,t)+\beta J^{\varphi,\text{nb}}_{\varepsilon R_1,f}(x,\theta,h,t)}{|J^{\varphi}_{\varepsilon R_1,f}(x,h,t)+\beta J^{\varphi,\text{nb}}_{\varepsilon R_1,f}(x,\theta,h,t)|}, \label{Eq:Fmoment_eps_00}\\
&T^\varepsilon_{f}(x,\theta,\varphi, h,t) = \left[N^\varepsilon_{f}(x,\theta,h,t) \times V(\varphi)\right] \cdot \hat z, \label{Eq:Fmoment_eps_000}
\\
& N^\varepsilon_{f}(x,\theta,h,t) = \frac{1}{\varepsilon^2} \, \mu \, R \,   \sum_{k,\, k-h = \pm 1} \sgn(k-h) \, J^{\varphi,\text{w}}_{\varepsilon R_3,f}(x,\theta,k,t).
\label{Eq:Fmoment_eps}
\end{align}

We can easily show the following expansion:
\begin{align}
&J_{\varepsilon R,f} = \varepsilon^2  \big(\pi R^2 \, j_f + O(\eps^2) \big),\nonumber\\
\intertext{with}
&j^{\varphi}_f(x,h,t) = \int_{\theta \in [0,\pi], \, \varphi \in [0,2\pi]} V(\varphi) f(x,\varphi,\theta ,h,t) \, d\varphi \, d\theta,\\
&j^{\varphi,\text{w}}_{f}(x,\theta,k,t) = \int_{\theta \in [0,\pi], \, \varphi \in [0,2\pi]} g(\theta,\theta')V(\varphi) f(x,\varphi,\theta' ,k,t) \, d\varphi \, d\theta',\label{eq:flux-nb}\\
&j^\theta_f(x,h,t) = \int_{\theta \in [0,\pi], \, \varphi \in [0,2\pi]} e^{2 i \theta}  f(x,\varphi,\theta ,h,t) \, d\varphi \, d\theta,
\end{align}
where the last quantities are the localized mean inclination angle and momentum. Therefore, system (\ref{Eq:kin_eps}-\ref{Eq:Fmoment_eps}) becomes, dropping ${\mathcal O}(\varepsilon^2)$ terms:
\begin{align}
&\varepsilon \big( \partial_t f^\varepsilon + c\nabla_x\cdot(V(\varphi)f^\varepsilon)  +  \partial_\theta ( T_{f^\varepsilon} f^\varepsilon) \big)=  K \, \partial_\theta \big( \sin(2(\theta - \bar \theta_{f^\varepsilon}))  f^\varepsilon \big)+ \delta\partial_\theta^2 f^\varepsilon \nonumber \\
& \hspace{5cm} 
+ \nu \,  \partial_\varphi \big( \sin(\varphi - \bar \varphi_{f^\varepsilon}^{\text{tot}}) f^\varepsilon \big) +  D\partial_\varphi^2 f^\varepsilon,
\label{Eq:f_eps}
\end{align}
where 
\begin{align}
&e^{2 i \bar \theta_f (x,h,t)} = \frac{j^\theta_f(x,h,t)}{|j^\theta_f(x,h,t)|}, \label{Eq:Fmoment_eps_0_res}
\\
&V(\bar \varphi_f^{\text{tot}}(x,\theta,h,t)) = \frac{j^{\varphi}_f(x,h,t)+\beta\, j^{\varphi,\text{w},\text{nb}}_{f}(x,\theta,h,t)}{|j^{\varphi}_f(x,h,t)+\beta\, j^{\varphi,\text{w},\text{nb}}_{f}(x,\theta,h,t)|},\nonumber \\
&j^{\varphi,\text{w},\text{nb}}_{f}(x,\theta,h,t) = \sum_{k,\, k-h = \pm 1} j^{\varphi,\text{w}}_f(x,\theta,k,t),\\
&T_f(x,\varphi,\theta, h,t) = \left[N_f(x,\theta,h,t)\times V(\varphi)\right] \cdot \hat z, \label{Eq:Fmoment_eps_000_res}
\\
& N_f(x,\theta,h,t) =  \frac{1}{\varepsilon} \, \mu  R \, \pi R_3^2 \,  \sum_{k,\, k-h = \pm 1} \sgn(k-h) \, j^{\varphi,\text{w}}_f(x,\theta,k,t).\nonumber
\label{Eq:Fmoment_eps_res}
\end{align}

We suppose also that the interaction between the layers happens on large time scale. Therefore, we write: $\frac{1}{\varepsilon} \mu  R \,  \pi R_3^2 = \mu' = {\mathcal O}(1)$. Inserting this ansatz, we obtain:
\begin{equation}
 N_f(x,\theta,h,t) = \mu' \,  \sum_{k,\, k-h = \pm 1} \sgn(k-h) \, j^{\varphi,\text{w}}_f(x,\theta,k,t).
\label{Eq:N_eps}
\end{equation}
Moreover, we suppose that $\beta = \varepsilon \beta'$ with $\beta' = O(1)$. We thus have the following expansion\footnote{For any vectors $a,b \in \R^2$ and $\eps > 0$, we have :
\begin{align*}
\frac{a+\eps b}{|a+\eps b|} &=\frac{a+\eps b}{|a|+\eps \frac{a}{|a|}\cdot b + O(\eps^2)} =  (\frac{a}{|a|}+\eps \frac{b}{|a|})(1-\eps\ \frac{a}{|a|}\cdot \frac{b}{|a|} + O(\eps^2))\\
&= \frac{a}{|a|} + \eps \left(\text{Id} - \frac{a\otimes a}{|a|^2}\right)\frac{b}{|a|} + O(\eps^2).
\end{align*}
}:
\begin{align}
&V(\bar \varphi_f^{\text{tot}}(x,\theta,h,t)) = V(\bar \varphi_f(x,h,t))\nonumber\\ 
&\hspace{2cm}+ \varepsilon\beta'\,\frac{|j^{\varphi,\text{w},\text{nb}}_{f}(x,\theta,h,t)|}{|j^{\varphi}_{f}(x,\theta,h,t)|}\,P_{V(\bar \varphi_f(x,h,t))^\perp}V(\bar \varphi_f^{\text{w},\text{nb}}(x,\theta,h,t))+O(\eps^2),\nonumber\\
\intertext{where $\bar \varphi_f$ and $\bar \varphi_f^{\text{nb}}$ are defined by:}
&V(\bar \varphi_f(x,h,t))=\frac{j^{\varphi}_f(x,h,t)}{|j^{\varphi}_f(x,h,t)|},\quad V(\bar \varphi_f^{\text{w},\text{nb}}(x,\theta,h,t))=\frac{j^{\varphi,\text{w},\text{nb}}_f(x,\theta,h,t)}{|j^{\varphi,\text{w},\text{nb}}_f(x,\theta,h,t)|},
\end{align}
and $P_{V(\bar \varphi_f(x,h,t))^\perp}X$ denotes the projection of $X$ onto the orthogonal plane to $V(\bar \varphi_f(x,h,t))$. Taking the cross product of the previous expansion with $V(\varphi)$, we easily obtain:
\begin{equation*}
\sin(\varphi-\bar \varphi_{f^\eps}^{\text{tot}}) = \sin(\varphi-\bar \varphi_{f^\eps}) - \varepsilon \beta'\, \frac{|j^{\varphi,\text{w},\text{nb}}_{f^\varepsilon}|}{|j^{\varphi}_{f^\varepsilon}|}\sin(\bar \varphi^{,\text{w},\text{nb}}_{f^\varepsilon} - \bar \varphi_{f^\varepsilon})\cos(\varphi - \bar \varphi_{f^\varepsilon}) + O(\eps^2).
\end{equation*}
Therefore, regrouping the $O(\eps)$ terms together, equation \eqref{Eq:f_eps} becomes:
\begin{align}
&\varepsilon \big( \partial_t f^\varepsilon + c\nabla_x\cdot(V(\varphi)f^\varepsilon)  +  \partial_\theta ( T_{f^\varepsilon} f^\varepsilon)  + \partial_\varphi \big(S_{f^\varepsilon} f^\varepsilon)\big) \nonumber \\
& \hspace{0.5cm} 
=  K \, \partial_\theta \big( \sin(2(\theta - \bar \theta_{f^\varepsilon}))  f^\varepsilon \big)+ \delta\partial_\theta^2 f^\varepsilon + \nu \,  \partial_\varphi \big( \sin(\varphi - \bar \varphi_{f^\varepsilon}) f^\varepsilon \big) +  D\partial_\varphi^2 f^\varepsilon,
\label{Eq:f_eps_2}
\end{align}
with
\begin{equation}
S_{f^\varepsilon} = \nu  \beta'\, \frac{|j^{\varphi,\text{w},\text{nb}}_{f^\varepsilon}|}{|j^{\varphi}_{f^\varepsilon}|}\sin(\bar \varphi^{\text{w},\text{nb}}_{f^\varepsilon} - \bar \varphi_{f^\varepsilon})\cos(\varphi - \bar \varphi_{f^\varepsilon}).
\label{Eq:S_eps}
\end{equation}
System \eqref{Eq:Fmoment_eps_0_res}-\eqref{Eq:S_eps} is the starting point for the derivation of the macroscopic model.

 \section{Macroscopic model}
 \label{sec:macro_model}
\subsection{Equilibria}
 
We want to take the limit $\eps \rightarrow 0$ in system \eqref{Eq:Fmoment_eps_0_res}-\eqref{Eq:S_eps}. Therefore, assuming that the distribution function $f^{\eps}$ converge to a limit denoted by $f^0$ as $\eps \rightarrow 0$, this limit satisfies the equilibrium condition $Q(f^0) = 0$ where 
$$ Q(f) = K\, \partial_\theta \big( \sin(2(\theta - \bar \theta_f)) f \big)+ \delta \, \partial_\theta^2 f + \nu\, \partial_\varphi \big(  \sin(\varphi - \bar \varphi_f) f \big) +  D \, \partial_\varphi^2 f . $$ 

We define the von Mises-Fisher (VMF) distribution with  periodicity $2 \pi/n$, $n \in {\mathbb N}\setminus \{0\}$, concentration parameter $\kappa >0$ and direction $\psi_0 \in {\mathbb R}/(2 \pi n {\mathbb Z})$ by:
\begin{equation}
 M_{n,\kappa, \psi_0} (\psi) = \frac{1}{Z_{n,\kappa}} \exp \big(\kappa \cos(n(\psi - \psi_0)) \big), \qquad \psi \in {\mathbb R}, 
\label{eq:def_VMF}
\end{equation}
with 
$$ Z_{n,\kappa} = \int_0^{\frac{2\pi}{n}} \exp \big(\kappa \cos(n(\psi - \psi_0)) \big) \, d \psi. $$
We introduce
\begin{align}
& {\mathcal M}_{\bar \varphi, \bar \theta}(\varphi, \theta) = M_{2,\kappa_2, \bar \theta} (\theta) M_{1,\kappa_1, \bar \varphi} (\varphi),\label{eq:def_VMF_2}\\
\intertext{with }
&\kappa_1 = \frac{\nu}{D},\quad \kappa_2 = \frac{K}{\delta}.\nonumber
\end{align}

We show the following:

\begin{proposition}
We have
\begin{align}
Q(f) &= \delta \, \partial_\theta \Big( {\mathcal M}_{\bar \varphi_f, \bar \theta_f}(\varphi, \theta) \, \partial_\theta \Big(\frac{f(\varphi, \theta)}{{\mathcal M}_{\bar \varphi_f, \bar \theta_f}(\varphi, \theta)} \Big) \Big) \nonumber \\
& \hspace{4cm} +
D \, \partial_\varphi \Big( {\mathcal M}_{\bar \varphi_f, \bar \theta_f}(\varphi, \theta) \, \partial_\varphi \Big(\frac{f(\varphi, \theta)}{{\mathcal M}_{\bar \varphi_f, \bar \theta_f}(\varphi, \theta)} \Big) \Big). 
\label{eq:eq3}
\end{align}
\label{prop:VMF}
\end{proposition}

\begin{proof} We have
\begin{align}
Q(f)&
= \partial_\theta \big( K \sin(2(\theta - \bar \theta_f)) f + \delta \, \partial_\theta f \big) 
+ \partial_\varphi \big( \nu \sin(\varphi- \bar \varphi_f) f + D \, \partial_\varphi f \big) 
\nonumber\\
&= \delta \, \partial_\theta \big( -\frac{K}{2\delta} \partial_\theta (\cos(2(\theta - \bar \theta_f))) f + \partial_\theta f \big) \nonumber \\
& \hspace{4cm} + D \, \partial_\varphi \big( -\frac{\nu}{D} \partial\varphi (\cos(\varphi- \bar \varphi_f)) f + \partial_\varphi f \big). 
\label{eq:eq1}
\end{align}
In view of (\ref{eq:def_VMF}), Eq. (\ref{eq:eq1}) can be written:
\begin{align}
0&= \delta \, \partial_\theta \Big( M_{2,\frac{K}{\delta}, \bar \theta_f} (\theta) \, \partial_\theta \Big(\frac{f(\varphi, \theta)}{M_{2,\frac{K}{\delta}, \bar \theta_f} (\theta)} \Big) \Big) +
D \, \partial_\varphi \Big( M_{1,\frac{\nu}{D}, \bar \varphi_f} (\varphi) \, \partial_\varphi \Big(\frac{f(\varphi, \theta)}{M_{1,\frac{\nu}{D}, \bar \varphi_f} (\varphi)} \Big) \Big). 
\label{eq:eq2}
\end{align}
and from (\ref{eq:def_VMF_2}), we deduce (\ref{eq:eq3}). \end{proof}

\noindent
We define 
$$ \rho_f (x,h,t) = \int_{\theta \in [0,\pi], \, \varphi \in [0,2\pi]} f(x,\varphi,\theta,h,t) \, d\varphi \, d\theta. $$
We have the: 

\begin{proposition}
For any distribution $\rho  {\mathcal M}_{\bar \varphi, \bar \theta}$, we have 
\begin{equation}\rho_{[\rho  {\mathcal M}_{\bar \varphi, \bar \theta} ] } = \rho,\quad j^\varphi_{[\rho  {\mathcal M}_{\bar \varphi, \bar \theta} ] }  =  c_1\big(1,  \frac{\nu}{D} \big) \, \rho \, V(\bar \varphi),\quad
j^\theta_{[\rho  {\mathcal M}_{\bar \varphi, \bar \theta} ] } = 
 c_1\big(2,  \frac{K}{\delta} \big) \, \rho \, e^{2 i \bar \theta},\label{eq:flux_VMF}
 \end{equation}
with 
\begin{equation}
c_1(n,\kappa) = \frac{1}{ Z_{n,\kappa}} \int_{0}^{2 \pi/n} \cos (nu) \, e^{\kappa \cos (nu)} \, du. 
\label{eq:def_c1}
\end{equation}
In particular, we have: 
$$
\bar \theta_{[\rho  {\mathcal M}_{\bar \varphi, \bar \theta} ] } = \bar \theta \mbox{ modulo } \pi, \quad 
\bar \varphi_{[\rho  {\mathcal M}_{\bar \varphi, \bar \theta} ] } = \bar \varphi \mbox{ modulo } 2 \pi. 
$$
\label{prop_moments}
\end{proposition}

This proposition shows that for the distribution $\rho  {\mathcal M}_{\bar \varphi, \bar \theta}$, $\rho$ is its local density and $\bar \varphi$, $\bar \theta$ are its mean direction of motion and mean inclination respectively. 

\begin{proof} By easy computations, we have
\begin{eqnarray*}
&& \hspace{-1cm} j^\varphi_{[\rho  {\mathcal M}_{\bar \varphi, \bar \theta} ] } = \int_{\theta \in [0,\pi], \, \varphi \in [0,2\pi]} V(\varphi) \, \rho {\mathcal M}_{\bar \varphi, \bar \theta} (\varphi, \theta) \, \, d\varphi \, d\theta \nonumber\\
&& \hspace{1cm} =  \rho \Big( \int_{\theta \in [0,\pi]} M_{2,\frac{K}{\delta}, \bar \theta} (\theta) \, d\theta  \Big) \,  
\Big( \int_{\varphi \in [0,2\pi]} V(\varphi) \, M_{1,\frac{\nu}{D}, \bar \varphi} (\varphi) \, d\varphi  \Big) \nonumber\\
&& \hspace{1cm} = \rho  \,  V(\bar \varphi) \, \int_{\varphi \in [0,2\pi]} M_{1,\frac{\nu}{D}, \bar \varphi} (\varphi) \, \cos \varphi \, d\varphi \nonumber\\
&& \hspace{1cm} = c_1\big(1,  \frac{\nu}{D} \big) \, \rho \, V(\bar \varphi), 
\end{eqnarray*}
We note that $c_1(n,\kappa)$ is such that $0 \leq c_1(n,\kappa) \leq 1$. Therefore, 
$$ V(\bar \varphi_{[\rho  {\mathcal M}_{\bar \varphi, \bar \theta} ] } ) = \frac{ j^\varphi_{[\rho  {\mathcal M}_{\bar \varphi, \bar \theta} ] }}{|j^\varphi_{[\rho  {\mathcal M}_{\bar \varphi, \bar \theta} ] }|} = V(\bar \varphi), 
$$
showing that $\bar \varphi_{[\rho  {\mathcal M}_{\bar \varphi, \bar \theta} ] } = \bar \varphi$ modulo $2 \pi$. 
Similarly, we have
\begin{eqnarray*}
&& \hspace{-1cm} j^\theta_{[\rho  {\mathcal M}_{\bar \varphi, \bar \theta} ] } = 
 c_1\big(2,  \frac{K}{\delta} \big) \, \rho \, e^{2 i \bar \theta}, 
\end{eqnarray*}
and so $ e^{2 i \bar \theta_{[\rho  {\mathcal M}_{\bar \varphi, \bar \theta} ] } } = 
e^{2 i \bar \theta}$, showing that $\bar \theta_{[\rho  {\mathcal M}_{\bar \varphi, \bar \theta} ] } = \bar \theta$ modulo $ \pi$. \end{proof}

\medskip
\begin{proposition} The set ${\mathcal E}$ of equilibria, i.e. ${\mathcal E} = \{f(\varphi, \theta) \, \, | \, \, f \geq 0, \mbox{ and } Q(f) = 0 \}$  is given by :
\begin{equation*}
{\mathcal E} = \left\{\rho  {\mathcal M}_{\bar \varphi, \bar \theta} \quad | \quad  \bar \varphi \in [0,2 \pi], \, \bar \theta \in [0,\pi], \, \rho \in \R^+ \right\}.
\end{equation*}
\end{proposition}

\medskip
\begin{proof} From (\ref{eq:eq3}), we deduce that 
\begin{eqnarray}
&&\hspace{-1cm}
\int_{\theta \in [0,\pi], \, \varphi \in [0,2\pi]} Q(f) \, \frac{f}{{\mathcal M}_{\bar \varphi_f, \bar \theta_f}} \, d\varphi \, d\theta \nonumber\\
&&\hspace{-0.5cm}
= - \int_{\theta \in [0,\pi], \, \varphi \in [0,2\pi]}  \Big( d \, \Big| \partial_\theta \big(  \frac{f}{{\mathcal M}_{\bar \varphi_f, \bar \theta_f}} \big) \Big|^2  
+ D \, \Big| \partial_\varphi \big(  \frac{f}{{\mathcal M}_{\bar \varphi_f, \bar \theta_f}} \big) \Big|^2 \Big)\, {\mathcal M}_{\bar \varphi_f, \bar \theta_f} \, d\varphi \, d\theta. 
\label{eq:eq:equi3}
\end{eqnarray}
Now, let $f$ be an equilibrium, i.e. $Q(f) = 0$. Then,  
\begin{eqnarray}
\hspace{-0.2cm}
\int_{\theta \in [0,\pi], \, \varphi \in [0,2\pi]}  \Big( d \, \Big| \partial_\theta \big(  \frac{f}{{\mathcal M}_{\bar \varphi_f, \bar \theta_f}} \big) \Big|^2  
+ D \, \Big| \partial_\varphi \big(  \frac{f}{{\mathcal M}_{\bar \varphi_f, \bar \theta_f}} \big) \Big|^2 \Big)\, {\mathcal M}_{\bar \varphi_f, \bar \theta_f} \, d\varphi \, d\theta = 0. 
\label{eq:eq:equi4}
\end{eqnarray}
This shows that there exists $\rho \in {\mathbb R}_+$, such that $f = \rho {\mathcal M}_{\bar \varphi_f, \bar \theta_f}$, which implies that there exist $\rho \in {\mathbb R}_+$,  $\bar \varphi \in {\mathbb R}$ (modulo $2 \pi$) and $\bar \theta \in {\mathbb R}$ (modulo $\pi$) such that $ f = \rho  {\mathcal M}_{\bar \varphi, \bar \theta}$. 

\noindent
Reciprocally, we show that $ f = \rho  {\mathcal M}_{\bar \varphi, \bar \theta}$ are such that $Q(f) = 0$. Indeed, this follows from (\ref{eq:eq:equi3}) and (\ref{eq:eq:equi4}) if we show that for $ f = \rho  {\mathcal M}_{\bar \varphi, \bar \theta}$, $\bar \theta_f = \bar \theta$ and $\bar \varphi_f = \bar \varphi$. But this follows in turn from Proposition \ref{prop_moments}. 
\end{proof}

According to this proposition, there exists $\rho(x,h,t) \in \R^+$ and $\bar \varphi(x,h,t) \in [0,2 \pi]$, $\bar \theta(x,h,t) \in [0,\pi]$ such that:
\begin{equation}
\lim_{\varepsilon \to 0} f^\varepsilon = f^0, \quad \mbox{with} \quad 
f^0(x,\varphi,\theta,h,t) = \rho(x,h,t){\mathcal M}_{\bar \varphi(x,h,t), \bar \theta(x,h,t)}(\varphi, \theta).
\label{Eq:local_eq}
\end{equation}
Now, the goal is to find equations for $(\rho, \bar \varphi, \bar \theta)$ as functions of $(x,h,t)$. 

\subsection{Collisional invariants}

Eq. (\ref{Eq:f_eps_2}) can be written
\begin{equation} 
\partial_t f^\varepsilon + c\nabla_x\cdot(V(\varphi)f^\varepsilon)  +  \partial_\theta ( T_{f^\varepsilon} f^\varepsilon) +  \partial_\varphi (S_{f^\varepsilon} f^\varepsilon) \big) = \frac{1}{\varepsilon} Q(f^\varepsilon). 
\label{eq:kineps}
\end{equation}
In the limit $\varepsilon \to 0$, the right-hand side of (\ref{eq:kineps}) is singular. The goal of a Collision Invariant is to cancel this singular term through integration in $(\varphi, \theta)$ against a suitable test functions. For this purpose, we define: 

\begin{definition} A Collision Invariant (CI) is a function ${\mathcal I}(\varphi, \theta)$ such that for all function $f(\varphi, \theta)$, we have
$$ \int_{\theta \in [0,\pi], \, \varphi \in [0,2\pi]}  Q(f) \, {\mathcal I} \, d\varphi \, d\theta = 0. $$
We denote by ${\mathcal C}$ the space of CI. It is a vector space. 
\end{definition}

Here, clearly, ${\mathcal C}$ contains the constant functions, since, as seen from (\ref{eq:eq3}), we have  $ \int Q(f) \, d\varphi \, d\theta = 0$. However, no other CI appears obviously from (\ref{eq:eq3}). The use of the constant functions as CI already gives the mass conservation equation. Indeed, integrating (\ref{eq:kineps}) with respect to $(\varphi, \theta)$, we get
\begin{equation}
 \partial_t \rho_{f^\varepsilon} + \nabla_x \cdot j^\varphi_{f^\varepsilon} = 0. 
\label{eq:masscons_eps}
\end{equation}
In this equation, the $1/\varepsilon$ singularity has disappeared and the limit $\varepsilon \to 0$ of Eq. (\ref{eq:masscons_eps}) leads to an equation for $f^0$. However, as seen from (\ref{Eq:local_eq}), $f^0$ depends on three scalar quantities: $\rho$, $\bar \theta$ and $\bar \varphi$ but (\ref{eq:masscons_eps}) is only one single scalar equation. Therefore, it is not sufficient to determine the dynamics of $f^0$. For this reason, we look for a weaker invariant concept, that of Generalized Collision Invariant, as defined in the next section.

\subsection{Generalized Collisional Invariants.} For a given pair $(\bar \varphi, \bar \theta) \in {\mathbb R}/2 \pi {\mathbb Z} \, \times \, {\mathbb R}/\pi {\mathbb Z}$, we define the operator ${\mathcal Q}(\bar \varphi, \bar \theta; f)$ by
\begin{align}
{\mathcal Q}(\bar \varphi, \bar \theta; f) &= \delta \, \partial_\theta \Big( {\mathcal M}_{\bar \varphi, \bar \theta}(\varphi, \theta) \, \partial_\theta \Big(\frac{f(\varphi, \theta)}{{\mathcal M}_{\bar \varphi, \bar \theta}(\varphi, \theta)} \Big) \Big)  \nonumber \\
& \hspace{4cm}
+ D \, \partial_\varphi \Big( {\mathcal M}_{\bar \varphi, \bar \theta}(\varphi, \theta) \, \partial_\varphi \Big(\frac{f(\varphi, \theta)}{{\mathcal M}_{\bar \varphi, \bar \theta}(\varphi, \theta)} \Big) \Big). 
\label{eq:eq36}
\end{align}
We note that $f \mapsto {\mathcal Q}(\bar \varphi, \bar \theta; f)$  is a linear operator and that 
\begin{equation}
Q(f) = {\mathcal Q}(\bar \varphi_f, \bar \theta_f; f). 
\label{eq:QQQ}
\end{equation}
We define the Generalized Collisional Invariants by the following

\begin{definition} Let $(\bar \varphi, \bar \theta) \in {\mathbb R}/2 \pi {\mathbb Z} \, \times {\mathbb R}/\pi {\mathbb Z}$ be given. 
A Generalized Collision Invariant (GCI) associated with $(\bar \varphi, \bar \theta)$ is a function ${\mathcal I}(\varphi, \theta)$ such that 
\begin{eqnarray}
&&\hspace{0cm}
 \int_{\theta \in [0,\pi], \, \varphi \in [0,2\pi]}   {\mathcal Q}(\bar \varphi, \bar \theta; f) \, {\mathcal I} \, d\varphi \, d\theta = 0, \nonumber \\
&&\hspace{2.5cm}
\,  \forall f  \mbox{ such that } 
\bar \theta_f = \bar \theta \mbox{ mod.} \, \frac{\pi}{2} \mbox{ and } \bar \varphi_f = \bar \varphi \mbox{ mod.} \, \pi. 
\label{eq:defgci}
\end{eqnarray}
We denote by ${\mathcal G}$ the space of GCI associated with $(\bar \varphi, \bar \theta)$. It is a vector space. 
\end{definition}

Referring to (\ref{eq:def_VMF}), for the simplicity of notation, we define 
\begin{equation}
M_{1,\bar \varphi} = M_{1,\kappa_1 ,\bar \varphi}, \qquad M_{2,\bar \theta} = M_{2,\kappa_2 ,\bar \theta}.
\label{eq:defM1M2}
\end{equation}
We introduce the two following functions: \\
(i) The function $I_1(\varphi)$ is a $2 \pi$-periodic solution of the problem
\begin{eqnarray}
&&\hspace{-1cm}  \partial_\varphi (  M_{1,0} \, \partial_\varphi I_1 )  = \sin \varphi M_{1,0} , \qquad 
\int_{\varphi \in [0,2\pi]} I_1(\varphi) \, d \varphi = 0 
\label{eq:diff_I1}
\end{eqnarray}
(ii) The function $I_2(\theta)$ is a $\pi$-periodic solution of the problem
\begin{eqnarray}
&&\hspace{-1cm}  \partial_\theta (  M_{2,0}\,  \partial_\theta I_2 )  = \sin(2 \theta ) M_{2,0} , \qquad 
 \int_{\theta \in [0,\pi]} I_2(\theta) \, d \theta = 0 . 
\label{eq:diff_I2}
\end{eqnarray}
Now we have

\begin{theorem}
The solutions $I_1$ and $I_2$ are unique. Moreover, the space ${\mathcal G}$ of GCI associated to $(\bar \varphi, \bar \theta)$ is three dimensional and spanned by 
$$ {\mathcal I}_0(\varphi, \theta) = 1, \quad {\mathcal I}_1(\varphi, \theta) = I_1(\varphi - \bar \varphi), \quad {\mathcal I}_2(\varphi, \theta) = I_2(\theta - \bar \theta). $$
\label{theorem:exist_GCI}
\end{theorem}

\begin{proof}
Introducing the $L^2$-adjoint $ {\mathcal Q}^*(\bar \varphi, \bar \theta; f)$ of $ {\mathcal Q}(\bar \varphi, \bar \theta; f)$, we can write:
$$  \int_{\theta \in [0,\pi], \, \varphi \in [0,2\pi]}   {\mathcal Q}(\bar \varphi, \bar \theta; f) \, {\mathcal I} \, d\varphi \, d\theta  = 
\int_{\theta \in [0,\pi], \, \varphi \in [0,2\pi]}  f \,  {\mathcal Q}^*(\bar \varphi, \bar \theta;  {\mathcal I} ) \,d\varphi \, d\theta. $$
The constraints $\bar \theta_f = \bar \theta$ modulo $\frac{\pi}{2}$ and $\bar \varphi_f = \bar \varphi$ modulo $\pi$ can be equivalently written:
\begin{eqnarray*}
 \int_{\theta \in [0,\pi], \, \varphi \in [0,2\pi]}  f \, \sin (2(\theta - \bar \theta)) \, d \theta \, d \varphi =0, \quad
\int_{\theta \in [0,\pi], \, \varphi \in [0,2\pi]}  f \, \sin (\varphi - \bar \varphi) \, d \theta \, d \varphi = 0. 
\end{eqnarray*}
Since these are linear constraints, by a classical duality argument, (\ref{eq:defgci}) is equivalent to saying that there exist $(\beta^\varphi, \beta^\theta) \in {\mathbb R}^2$ such that 
\begin{eqnarray*}
&&\hspace{-1cm} 
\int_{\theta \in [0,\pi], \, \varphi \in [0,2\pi]}  f \,  {\mathcal Q}^*(\bar \varphi, \bar \theta;  {\mathcal I} ) \,d\varphi \, d\theta 
\\
&&\hspace{1cm} = \int_{\theta \in [0,\pi], \, \varphi \in [0,2\pi]}  f \, \big(\beta^\theta \, \sin (2(\theta - \bar \theta))  + \beta^\varphi \, \sin (\varphi - \bar \varphi) \big) \, d \theta \, d \varphi. 
\end{eqnarray*}
for all functions $f$ without constraints. This implies that ${\mathcal I}$ satisfies: 
\begin{eqnarray*}
&&\hspace{-1cm} 
\exists (\beta^\varphi, \beta^\theta) \in {\mathbb R}^2 \quad \mbox{ such that } \quad  {\mathcal Q}^*(\bar \varphi, \bar \theta;  {\mathcal I} ) = \beta^\theta \, \sin (2(\theta - \bar \theta))  + \beta^\varphi \, \sin (\varphi - \bar \varphi) . 
\end{eqnarray*}
or, using the explicit expression of ${\mathcal Q}^*$: 
\begin{eqnarray}
&&\hspace{-1cm} 
\delta \, \partial_\theta \big( {\mathcal M}_{\bar \varphi, \bar \theta} \, \partial_\theta {\mathcal I} \big) + D \, \partial_\varphi \big( {\mathcal M}_{\bar \varphi, \bar \theta} \, \partial_\varphi {\mathcal I} \big) \nonumber \\
&&\hspace{3cm} 
= \big( \beta^\theta \, \sin (2(\theta - \bar \theta))  + \beta^\varphi \, \sin (\varphi - \bar \varphi) \big) \, {\mathcal M}_{\bar \varphi, \bar \theta}  . 
\label{eq:rel_GCI}
\end{eqnarray}
Multiplying by a test function ${\mathcal J}$, integrating with respect to $(\varphi, \theta)$, using Green's formula and the $2 \pi$ periodicity in $\varphi$ (resp. the $\pi$ periodicity in $\theta$), we obtain the following variational formulation:
\begin{equation}
\int_{\theta \in [0,\pi], \, \varphi \in [0,2\pi]} {\mathcal M}_{\bar \varphi, \bar \theta} \, \big( \delta \, \partial_\theta {\mathcal I} \, \partial_\theta {\mathcal J} + D \, \partial_\varphi {\mathcal I} \, \partial_\varphi {\mathcal J} \big) \, d\varphi \, d\theta = \int_{\theta \in [0,\pi], \, \varphi \in [0,2\pi]} \psi \, {\mathcal J} \, d\varphi \, d\theta, 
\label{eq:varfor}
\end{equation}
with $\psi  = - \big( \beta^\theta \, \sin (2(\theta - \bar \theta))  + \beta^\varphi \, \sin (\varphi - \bar \varphi) \big)  {\mathcal M}_{\bar \varphi, \bar \theta} $. We now introduce the functional spaces $L^2([0,2\pi] \times [0,\pi])$, $H^1([0,2\pi] \times [0,\pi])$
endowed with their classical Hilbert norms and inner products, together with:  
\begin{eqnarray*}
&&\hspace{-1cm}
H_{\mbox{\scriptsize per}}^1([0,2\pi] \times [0,\pi])  = \{  {\mathcal J} \in H^1([0,2\pi] \times [0,\pi]) \, \, | \, \,  {\mathcal J}(0,\theta) = {\mathcal J}(2 \pi,\theta), \\
&&\hspace{3cm}
 \quad {\mathcal J}(\varphi,0) = {\mathcal J}(\varphi, \pi), \, \, a.e. \, \, (\varphi, \theta) \in  [0,2\pi] \times [0,\pi] \}, 
\end{eqnarray*}
 and
$$ \dot H^1([0,2\pi] \times [0,\pi]) = \{  H_{\mbox{\scriptsize per}}^1([0,2\pi] \times [0,\pi])  \, \, | \, \, \int_{\theta \in [0,\pi], \, \varphi \in [0,2\pi]}  {\mathcal J} \, d\varphi \, d\theta = 0 \}. $$
Thanks to a Poincar\'e-Wirtinger inequality (which can be easily proved using the Rellich-Kondrachov compactness theorem), the semi-norm 
$$ |{\mathcal J}|_1^2 = \int_{\theta \in [0,\pi], \, \varphi \in [0,2\pi]} {\mathcal M}_{\bar \varphi, \bar \theta} \, \big( \delta \, \big|\partial_\theta {\mathcal J} \big|^2 + D \, \big| \partial_\varphi {\mathcal J}\big|^2 \big) \, d\varphi \, d\theta, $$
is a norm on $\dot H^1([0,2\pi] \times [0,\pi])$ equivalent to the classical $H^1$ norm. Therefore, thanks to Lax-Milgram theorem in $\dot H^1([0,2\pi] \times [0,\pi])$, there exists a unique  $ {\mathcal I} \in \dot H^1([0,2\pi] \times [0,\pi])$ such that (\ref{eq:varfor}) holds for any $ {\mathcal J} \in \dot H^1([0,2\pi] \times [0,\pi])$. Now, since $\psi$ is the sum of two terms, one being odd in $\theta - \bar \theta $, the other one being odd in  $\varphi - \bar \varphi$, we have 
$$\int_{\theta \in [0,\pi], \, \varphi \in [0,2\pi]} \psi \, d\varphi \, d\theta = 0. $$
Therefore, $ {\mathcal I} $ satisfies (\ref{eq:varfor}) for all ${\mathcal J} \in H_{\mbox{\scriptsize per}}^1([0,2\pi] \times [0,\pi]) $ (and not only for $ {\mathcal J} \in \dot H^1([0,2\pi] \times [0,\pi])$). Furthermore, all solutions in $H_{\mbox{\scriptsize per}}^1([0,2\pi] \times [0,\pi]) $ of (\ref{eq:varfor}) equal the unique solution in $\dot H^1([0,2\pi] \times [0,\pi])$ up to a constant. Indeed, if ${\mathcal I}$ solves  (\ref{eq:rel_GCI}) with $\psi = 0$ in $H_{\mbox{\scriptsize per}}^1([0,2\pi] \times [0,\pi])$, we have $|{\mathcal I}|_1^2 = 0$ and therefore, ${\mathcal I}$ is a constant. 

Now, we solve (\ref{eq:rel_GCI}) for $(\beta^\varphi, \beta^\theta) = (D,0)$ or $(0,\delta)$ and for this purpose, we use the functions $I_1$ and $I_2$ defined at (\ref{eq:diff_I1}) and  (\ref{eq:diff_I2}). We note that $I_{1,\bar \varphi} (\varphi) = I_1(\varphi- \bar \varphi)$ is the unique solution in $\dot H^1([0,2\pi])$ of  the variational formulation 
\begin{eqnarray}
&&\hspace{0cm}
 \int_{\varphi \in [0,2\pi]} M_{1,\bar \varphi} \,  \partial_\varphi I_{1,\bar \varphi}  \, \partial_\varphi J_1 \, d \varphi \nonumber \\
&&\hspace{1.5cm} = - \int_{\varphi \in [0,2\pi]} \sin(\varphi - \bar \varphi) \, M_{1,\bar \varphi} \,  J_1 \, d \varphi, \quad \forall J_1 \in \dot H^1([0,2\pi]). 
\label{eq:varI1}
\end{eqnarray}
and that $I_{2,\bar \theta} (\theta) = I_2(\theta- \bar \theta)$ is the unique solution in $\dot H^1([0,\pi])$ of  the variational formulation 
\begin{eqnarray}
&&\hspace{0cm}
\int_{\theta \in [0,\pi]} M_{2,\bar \theta} \, \partial_\theta I_{2,\bar \theta} \, \partial_\theta J_2 \, d \theta \nonumber \\
&&\hspace{1.5cm}  = - \int_{\theta \in [0,\pi]} \sin(2(\theta - \bar \theta)) \, M_{2,\bar \theta} \,  J_2 \, d \theta, \quad \forall J_2 \in \dot H^1([0,\pi]). 
\label{eq:varI2}
\end{eqnarray}
The existence and uniqueness of solutions to (\ref{eq:varI1}) and (\ref{eq:varI2}) follow from the same kind of arguments as for problem (\ref{eq:varfor}). Now, it is an easy matter to check that both $I_{1,\bar \varphi}$ and $I_{2,\bar \theta}$ are solutions of (\ref{eq:varfor}) with $(\beta^\varphi, \beta^\theta) = (D,0)$ and $(0,\delta)$ respectively. Moreover, they both are in  $\dot H^1([0,2\pi] \times [0,\pi])$ and by the uniqueness of the solution of (\ref{eq:varfor}), they are the unique solution of this problem with these choices of $(\beta^\varphi, \beta^\theta)$. We deduce that the space ${\mathcal G}$ is three-dimensional, spanned by ${\mathcal I}_0(\varphi, \theta) = 1$, ${\mathcal I}_1(\varphi, \theta) = I_{1,\bar \varphi} (\varphi) $ and ${\mathcal I}_2(\varphi, \theta) = I_{2,\bar \theta} (\theta) $, which ends the proof. \end{proof}

Functions $I_1(\varphi)$ and $I_2(\theta)$ equal:
\begin{align}
&I_1(\varphi) = - \frac{\varphi}{\kappa_1} + \frac{\pi}{\kappa_1} \frac{\int_0^\varphi e^{- \kappa_1 \cos u} \, du}{\int_0^\pi e^{- \kappa_1 \cos u} \, du},\label{eq:I1}\\
&I_2(\theta) = - \frac{\theta}{2\kappa_2} + \frac{\pi/2}{2\kappa_2} \frac{\int_0^\theta e^{- \kappa_2 \cos 2 u} \, du}{\int_0^{\pi/2} e^{- \kappa_2 \cos 2 u} \, du}. \label{eq:I2}
\end{align}
Thanks to (\ref{eq:QQQ}) and (\ref{eq:defgci}), they satisfy:
\begin{eqnarray}
&&\hspace{-1cm} 
\int_{\theta \in [0,\pi], \, \varphi \in [0,2\pi]}  Q(f^\varepsilon) \,  I_1(\varphi-\bar \varphi_{f^\varepsilon}) \, d \varphi \, d \theta = 0, 
\label{eq:cancelGCI_I1} \\
&&\hspace{-1cm} 
\int_{\theta \in [0,\pi], \, \varphi \in [0,2\pi]}  Q(f^\varepsilon) \, I_2(\theta-\bar \theta_{f^\varepsilon}) \, d \varphi \, d \theta = 0. 
\label{eq:cancelGCI_I2}
\end{eqnarray}

\subsection{Macroscopic equations} 

We obtain the macroscopic dynamics by integrating system \eqref{Eq:f_eps} against the GCI. The resulting equations are presented in the following proposition:

\begin{theorem} The density $\rho(x,h,t)$, the angle of the mean direction $\bar \varphi(x,h,t)$ and the mean inclination angle $\bar \theta(x,h,t)$ satisfy the following system:
\begin{align}
&\partial_t \rho + c c_1 \, \nabla_x \cdot \left( \rho V(\bar \varphi) \right)= 0, 
\label{eq:mass_final} \\
&\rho \big( \partial_t \bar \varphi +  c c_2 (V(\bar \varphi) \cdot\nabla_x) \bar \varphi \big) + \frac{c}{\kappa_1} V(\bar \varphi)^\bot \cdot\nabla_x \rho\nonumber\\
&\qquad = \frac{\nu\beta' }{c_3}\,  V(\bar \varphi)^\perp\cdot   \sum_{k,\, k-h = \pm 1}  \langle gM_2M_2\rangle(\bar\theta,\bar\theta_k)\, c_1 \rho_k V(\bar\varphi_k),
\label{eq:phi_final} \\
&\rho( \partial_t \bar \theta +  c c_1 (V(\bar \varphi) \cdot\nabla_x) \bar \theta )\nonumber\\
&\quad= \frac{ \mu'}{c_4}\, \, (c_1  \rho  V(\bar \varphi))^\perp\cdot \sum_{k,\, k-h = \pm 1} \sgn(k-h) \langle gM_2M_2\partial_\theta I_2\rangle(\bar\theta,\bar\theta_k)\,  c_1 \rho_k V(\bar \varphi_k),
\label{eq:theta_final}
\intertext{with }
&\langle gM_2M_2\rangle(\bar\theta,\bar\theta_k) = \int_{\theta,\theta' \in [0,\pi]} g(\bar\theta+\theta,\bar\theta_k+\theta') M_2(\theta)M_2(\theta')\, d\theta d\theta',
\label{eq:g_mean1_final}\\
&\langle gM_2M_2\partial_\theta I_2\rangle(\bar\theta,\bar\theta_k) =  2\kappa_2 \int_{\theta,\theta' \in [0,\pi]} g(\bar\theta+\theta,\bar\theta_k+\theta') M_2(\theta)\partial_\theta I_2(\theta)M_2(\theta')\, d\theta d\theta',
\label{eq:g_mean2_final}
\end{align}
where $(\rho_k,\bar \varphi_k,\bar \theta_k)(x,t)$ denotes $(\rho,\bar \varphi,\bar \theta)(x,k,t)$ and $c_1 = c_1(1, \kappa_1)$ (with $c_1(n,\kappa)$ defined at (\ref{eq:def_c1})) and $c_2$, $c_3$, $c_4$ given by:
\begin{align*}
&c_2 = \frac{\int_{\varphi \in [0,2\pi]} \sin \varphi \cos \varphi M_1 I_1\, d\varphi}{\int_{\varphi \in [0,2\pi]} \sin \varphi M_1 I_1\, d\varphi},\\
&c_3 =-\kappa_1^2 \int_{\varphi \in [0,2\pi]} \sin \varphi M_1 I_1\, d\varphi,\quad c_4 =  -4\kappa_2^2\int_{\theta \in [0,\pi]} \sin 2 \theta M_2 I_2\, d\theta,\quad  
\end{align*}
and $M_1 = M_{1,0}$, $M_2 = M_{2,0}$ (see (\ref{eq:defM1M2})). 
\label{prop:macro_system}
\end{theorem}

\noindent Before giving the proof, let us make some comments. The left-hand side of equations \eqref{eq:mass_final}-\eqref{eq:phi_final} is exactly the SOH (Self-Organized Hydrodynamics) model describing the Vicsek dynamics at the macroscopic level (see Ref. \cite{2008_ContinuumLimit_DM}). The right-hand side of equation describes the alignment of $V(\bar\varphi)$ toward a linear combination of the velocities of the neighboring layers. The weights of this linear combination depends on the inclination of the different layers. Equation \eqref{eq:phi_final} describes the advection of the inclination with the same advection velocity as for the mass. The right-hand side of eq. \eqref{eq:phi_final} finally also evaluates weighted alignment terms between layers. 

The weights are given by \eqref{eq:g_mean1_final}-\eqref{eq:g_mean2_final}. They are integral operators, quadratic with respect to the macroscopic equilibria in inclination variable. They are scaled in such a way to be bounded quantities (see  \ref{annex:coeff_macro}).  Weights \eqref{eq:g_mean2_final} involve the second generalized invariant. 

Finally, this macroscopic model depends on several coefficients, $c_1$, $c_2$, $c_3$, $c_4$, that are all positive and bounded by $1$ (see \ref{annex:coeff_macro}). They are all averages of the von Mises equilibria (either in velocity or inclination variables) against the collisional invariants.

\begin{proof} We apply the moment method: first integrate the equation against the collisional invariants and then taking the limit $\eps \rightarrow 0$. 

\textbf{ Mass conservation equation.} Here it is just a matter of passing to the limit $\varepsilon \to 0$ in (\ref{eq:masscons_eps}), using (\ref{Eq:local_eq}), (\ref{eq:flux_VMF}), (\ref{eq:def_c1}). We immediately get (\ref{eq:mass_final}). 

\textbf{ Velocity and inclination equation.} We multiply \eqref{Eq:f_eps_2} by $I_1^\varepsilon := I_1(\varphi- \bar \varphi_{f^\varepsilon})$ or $I_2^\varepsilon := I_2(\theta- \bar \theta_{f^\varepsilon})$, integrate with respect to $(\varphi, \theta)$ and use (\ref{eq:cancelGCI_I1}) or (\ref{eq:cancelGCI_I2}). We get
\begin{eqnarray*} 
&&\hspace{0cm} 
\int_{\theta \in [0,\pi], \, \varphi \in [0,2\pi]} \Big( \partial_t f^\varepsilon + c V(\varphi) \cdot \nabla_x f^\varepsilon + \partial_\theta(T_{f^\varepsilon} f^\varepsilon) + \partial_\varphi (S_{f^\varepsilon} f^\varepsilon)\Big) \, I_k^\varepsilon \, d\varphi \, d \theta \\
&&\hspace{10cm} 
= 0, \quad k=1,2, 
\end{eqnarray*}
In the limit $\eps \rightarrow 0$, we have $(\bar \varphi_{f^{\eps}}, \bar \theta_{f^{\eps}})\rightarrow (\bar \varphi, \bar \theta)$ and consequently $I_k^\varepsilon \to I_k$, where $I_k$ stands for the same quantities with $f^\varepsilon$ replaced by $\rho {\mathcal M}_{\bar \varphi, \bar \theta}$.  Therefore, we get
\begin{eqnarray} 
&&\hspace{-0.5cm} 
\int_{\theta \in [0,\pi], \, \varphi \in [0,2\pi]} \Big( \partial_t (\rho {\mathcal M}_{\bar \varphi, \bar \theta}) + c V(\varphi) \cdot \nabla_x (\rho {\mathcal M}_{\bar \varphi, \bar \theta}) + \partial_\theta(T_{\rho {\mathcal M}_{\bar \varphi, \bar \theta}} \rho {\mathcal M}_{\bar \varphi, \bar \theta}) + \nonumber \\
&&\hspace{3.5cm} 
\partial_\varphi(S_{\rho {\mathcal M}_{\bar \varphi, \bar \theta}} \rho {\mathcal M}_{\bar \varphi, \bar \theta}) \Big) \, I_k \, d\varphi \, d \theta 
= 0, \quad k=1,2, 
\label{eq:intoverIk}
\end{eqnarray} 
Tedious but easy algebra leads to: 
\begin{eqnarray*}
&&\hspace{-1cm} \partial_t (\rho {\mathcal M}_{\bar \varphi, \bar \theta}) + c V(\varphi) \cdot \nabla_x (\rho {\mathcal M}_{\bar \varphi, \bar \theta})  = {\mathcal M}_{\bar \varphi, \bar \theta} \sum_{i,j = 1}^2 {\mathcal T}_{ij}, 
\end{eqnarray*}
with ${\mathcal T}_{ij}$ even with respect to $\varphi - \bar \varphi$ if $i$ is even and even with respect to $\theta - \bar \theta$ if $j$ is even, odd otherwise. This gives:
\begin{eqnarray*}
&&\hspace{-1cm} {\mathcal T}_{11} =   \rho \frac{2 c K}{\delta} \sin(2(\theta - \bar \theta))  \sin(\varphi - \bar \varphi) V(\bar \varphi)^\bot \cdot \nabla_x \bar \theta\\
&&\hspace{-1cm} {\mathcal T}_{12} =  \rho \frac{\nu}{D} \sin(\varphi - \bar \varphi) \partial_t \bar \varphi +c \sin(\varphi - \bar \varphi) V(\bar \varphi)^\bot \cdot \nabla_x \rho \\
&&\hspace{4cm} 
+ \rho \frac{c \nu}{D} \sin(\varphi - \bar \varphi) \cos(\varphi - \bar \varphi) V(\bar \varphi) \cdot \nabla_x \bar \varphi
\\
&&\hspace{-1cm} {\mathcal T}_{21} =  \rho \frac{2 K}{\delta} \sin(2(\theta - \bar \theta)) \partial_t \bar \theta + \rho \frac{2 c K}{\delta} \sin(2(\theta - \bar \theta)) \cos(\varphi - \bar \varphi)  V(\bar \varphi) \cdot \nabla_x \bar \theta\\
&&\hspace{-1cm} {\mathcal T}_{22} = \partial_t \rho + c \cos(\varphi - \bar \varphi) V(\bar \varphi) \cdot \nabla_x \rho + \rho \frac{c \nu}{D} \sin^2(\varphi - \bar \varphi) V(\bar \varphi)^\bot \cdot \nabla_x \bar \varphi\\
\end{eqnarray*}
Since ${\mathcal M}_{\bar \varphi, \bar \theta}$ is even in $\theta - \bar \theta$ and $\varphi - \bar \varphi$ while $I_1$ is odd in $\varphi - \bar \varphi$ and $I_2$ is odd in $\theta - \bar \theta$, we get: 
\begin{eqnarray}
&&\hspace{-1cm} 
\int_{\theta \in [0,\pi], \, \varphi \in [0,2\pi]} \Big( \partial_t (\rho {\mathcal M}_{\bar \varphi, \bar \theta}) + c V(\varphi) \cdot \nabla_x (\rho {\mathcal M}_{\bar \varphi, \bar \theta}) \Big) \, I_1 \, d\varphi \, d \theta  \nonumber\\
&&\hspace{-0.5cm} 
= \int_{\theta \in [0,\pi], \, \varphi \in [0,2\pi]} {\mathcal M}_{\bar \varphi, \bar \theta} \, {\mathcal T}_{12} \, I_1 \, d\varphi \, d \theta \nonumber \\
&&\hspace{-0.5cm} 
= \rho \Big\{ \frac{\nu}{D} \Big(\int_{\varphi \in [0,2\pi]} \sin \varphi M_1 I_1\, d\varphi\Big) \partial_t \bar \varphi \nonumber \\
&&\hspace{2cm} 
+ \frac{c \nu}{D} \Big(\int_{\varphi \in [0,2\pi]} \sin \varphi \cos \varphi M_1 I_1\, d\varphi\Big) V(\bar \varphi) \cdot \nabla_x \bar \varphi \Big\} \nonumber\\
&&\hspace{2cm} 
+ c \Big(\int_{\varphi \in [0,2\pi]} \sin \varphi M_1 I_1\, d\varphi\Big) V(\bar \varphi)^\bot \cdot \nabla_x \rho. 
\label{eq:averageonI1}
\end{eqnarray}
while 
\begin{eqnarray}
&&\hspace{-1cm} 
\int_{\theta \in [0,\pi], \, \varphi \in [0,2\pi]} \Big( \partial_t (\rho {\mathcal M}_{\bar \varphi, \bar \theta}) + c V(\varphi) \cdot \nabla_x (\rho {\mathcal M}_{\bar \varphi, \bar \theta}) \Big) \, I_2 \, d\varphi \, d \theta \nonumber \\
&&\hspace{-0.5cm} 
= \int_{\theta \in [0,\pi], \, \varphi \in [0,2\pi]} {\mathcal M}_{\bar \varphi, \bar \theta} \, {\mathcal T}_{21} \, I_2 \, d\varphi \, d \theta \nonumber \\
&&\hspace{-0.5cm} 
= \rho \Big\{ \frac{2 K }{\delta} \Big(\int_{\theta \in [0,\pi]} \sin 2 \theta M_2 I_2\, d\theta \Big) \partial_t \bar \theta \nonumber \\
&&\hspace{0cm} 
+\frac{2 c K }{\delta} \Big(\int_{\theta \in [0,\pi]} \sin 2 \theta M_2 I_2\, d\theta \Big) \Big(\int_{\varphi \in [0,2\pi]} \cos \varphi M_1\, d\varphi \Big) V(\bar \varphi) \cdot \nabla_x \bar \theta \Big\}. 
\label{eq:averageonI2}
 \end{eqnarray}
We now treat the last two terms of (\ref{eq:intoverIk}). Using integration by parts we have:
\begin{eqnarray*} 
&&\hspace{-1cm} 
\int_{\theta \in [0,\pi], \, \varphi \in [0,2\pi]} \partial_\theta(T_{\rho {\mathcal M}_{\bar \varphi, \bar \theta}} \rho {\mathcal M}_{\bar \varphi, \bar \theta}) \, I_k \, d\varphi \, d \theta \\
&&\hspace{2cm} 
= - \int_{\theta \in [0,\pi], \, \varphi \in [0,2\pi]} T_{\rho {\mathcal M}_{\bar \varphi, \bar \theta}} \rho {\mathcal M}_{\bar \varphi, \bar \theta} \, \partial_\theta I_k \, d\varphi \, d \theta, \\
&&\hspace{-1cm} 
\int_{\theta \in [0,\pi], \, \varphi \in [0,2\pi]} \partial_\varphi(S_{\rho {\mathcal M}_{\bar \varphi, \bar \theta}} \rho {\mathcal M}_{\bar \varphi, \bar \theta}) \, I_k \, d\varphi \, d \theta \\
&&\hspace{2cm} 
= - \int_{\theta \in [0,\pi], \, \varphi \in [0,2\pi]} S_{\rho {\mathcal M}_{\bar \varphi, \bar \theta}} \rho {\mathcal M}_{\bar \varphi, \bar \theta} \, \partial_\varphi I_k \, d\varphi \, d \theta .
\end{eqnarray*} 
Since $I_1$ does not depend on $\theta$, the contribution of the third term of (\ref{eq:intoverIk}) for $k=1$ vanishes. Let us write  \eqref{Eq:S_eps}  as follows:
\begin{equation*}
S_{\rho {\mathcal M}_{\bar \varphi, \bar \theta}} = \nu  \beta'\, \frac{1}{|j^{\varphi}_{\rho {\mathcal M}_{\bar \varphi, \bar \theta}}|}\left(j^{\varphi,\text{w},\text{nb}}_{\rho {\mathcal M}_{\bar \varphi, \bar \theta}}\cdot V(\bar \varphi_{\rho {\mathcal M}_{\bar \varphi, \bar \theta}})^\perp\right)\cos(\varphi - \bar \varphi_{\rho {\mathcal M}_{\bar \varphi, \bar \theta}}).
\end{equation*}
From \eqref{eq:flux-nb} and and using (\ref{eq:flux_VMF}), we get:
\begin{equation}
j^{\varphi,\text{w}}_{\rho {\mathcal M}_{\bar \varphi, \bar \theta}}(x,\theta,h,t) = \left(\int_{\theta \in [0,\pi]} g(\theta,\bar\theta+\theta') M_2 \, d\theta'\right)\, c_1 \rho V(\bar\varphi),
\label{eq:mean-weighted}
\end{equation}
and then:
\begin{equation*}
S_{\rho {\mathcal M}_{\bar \varphi, \bar \theta}} = \nu\beta' \,\cos(\varphi - \bar \varphi)\, V(\bar\varphi)^\perp \cdot \sum_{k,\, k-h = \pm 1} \left(\int_{\theta \in [0,\pi]} g(\theta,\bar\theta_k+\theta') M_2 \, d\theta'\right)\, \frac{c_1 \rho_k V(\bar\varphi_k)}{c_1\rho}.
\end{equation*}
Therefore, after integration by parts, and using that $M_1\partial_\varphi I_1$ is even in $\varphi$, we have:
\begin{align} 
&\int_{\theta \in [0,\pi], \, \varphi \in [0,2\pi]} \partial_\varphi(S_{\rho {\mathcal M}_{\bar \varphi, \bar \theta}} \rho {\mathcal M}_{\bar \varphi, \bar \theta}) \, I_1 \, d\varphi \, d \theta \nonumber  \\
& = -\nu\beta' \rho\big(\int_{\varphi \in [0,2\pi]} \cos \varphi M_1\partial_\varphi I_1\, d\varphi \big) V(\bar \varphi)^\perp\cdot  \sum_{k,\, k-h = \pm 1}  \langle gM_2M_2\rangle(\bar\theta,\bar\theta_k)\,  \frac{c_1\rho_k V(\bar\varphi_k)}{c_1\rho}\nonumber\\
& = \frac{ \nu\beta' }{\kappa_1}  V(\bar \varphi)^\perp\cdot  \sum_{k,\, k-h = \pm 1}  \langle gM_2M_2\rangle(\bar\theta,\bar\theta_k)\, c_1\rho_k  V(\bar\varphi_k).
\label{eq:SonI1}
\end{align} 
with $\langle gM_2M_2\rangle(\bar\theta,\bar\theta_k)$ defined in \eqref{eq:g_mean1_final} and where we use the relation: 
\begin{align*}
\int_{\varphi \in [0,2\pi]} \cos \varphi M_1\partial_\varphi I_1\, d\varphi &= -\int_{\varphi \in [0,2\pi]} \sin \varphi\, \partial_\varphi(M_1\partial_\varphi I_1)\, d\varphi\\
&= - \int_{\varphi \in [0,2\pi]} \sin^2 \varphi\, M_1\, d\varphi = - \frac{c_1}{\kappa_1},
\end{align*}
obtained by integration by part and using \eqref{eq:diff_I1}. Using (\ref{eq:averageonI1}) and \eqref{eq:SonI1}   and dividing by $-c_3 = \frac{\nu}{D} \big(\int_{\varphi \in [0,2\pi]} \sin \varphi M_1 I_1\, d\varphi\big)$ we get (\ref{eq:phi_final}). 

Since $I_2$ does not depend on $\varphi$, the contribution of the last term to (\ref{eq:intoverIk}) for $k=2$ vanishes. From equations (\ref{Eq:Fmoment_eps_000_res})-(\ref{Eq:N_eps}) and (\ref{eq:flux_VMF}), we obtain the expression:
\begin{align*}
&T_{\rho {\mathcal M}_{\bar \varphi, \bar \theta}}(x,\varphi,\theta,h,t)  = -\left[V(\varphi)\times N_{\rho {\mathcal M}_{\bar \varphi, \bar \theta}}(x,\theta,h,t)\right] \cdot \hat z,\\
&N_{\rho {\mathcal M}_{\bar \varphi, \bar \theta}}(x,\theta,h,t) = \mu' \, \sum_{k,\, k-h = \pm 1} \sgn(k-h)\, j^{\varphi,\text{w}}_{\rho {\mathcal M}_{\bar \varphi, \bar \theta}}(x,\theta,k,t).
\end{align*}
Therefore, using \eqref{eq:mean-weighted} and since $V(\varphi) = \cos(\varphi-\bar \varphi)V(\bar \varphi)+ \sin(\varphi-\bar \varphi)V(\bar \varphi)^\perp$ and $M_1$ is even in $\varphi$,  we get after integration by parts:
\begin{align} 
&\int_{\theta \in [0,\pi], \, \varphi \in [0,2\pi]} \partial_\theta(T_{\rho {\mathcal M}_{\bar \varphi, \bar \theta}} \rho {\mathcal M}_{\bar \varphi, \bar \theta}) \, I_2 \, d\varphi \, d \theta \nonumber  \\
& = \rho\big(\int_{\varphi \in [0,2\pi]} \cos \varphi M_1\, d\varphi \big)  V(\bar \varphi)^\perp\cdot  \left[\mu' \, \sum_{k,\, k-h = \pm 1}   \langle gM_2M_2\partial_\theta I_2\rangle(\bar\theta,\bar\theta_k)\, c_1 \rho_k V(\bar\varphi_k)\right].
\label{eq:torqueonI2}
\end{align} 
with $\langle gM_2M_2\partial_\theta I_2\rangle(\bar\theta,\bar\theta_k)$ defined in \eqref{eq:g_mean2_final}.
Using (\ref{eq:averageonI2}) and (\ref{eq:torqueonI2}) into (\ref{eq:intoverIk}) for $k=2$ and dividing by $-c_4 = 2\kappa_2 \big(\int_{\theta \in [0,\pi]} \sin 2 \theta M_2 I_2\, d\theta \big)$, we get  (\ref{eq:theta_final}).
\end{proof}

\subsection{Macroscopic equilibria} 
\label{sec:macro_equilibria}

One simple macroscopic equilibrium consists in layers with the same vector velocity fields (or opposite vector field). In that case, inclinations have no impact on the dynamics and are simply transported by the velocity flow. In particular, we have:

\begin{proposition}  For any $\bar\varphi \in [0,2\pi]$, $(i_k)_{k\in \Z} \in \left\{0,1\right\}^\Z$  and $(\bar\theta_k)_{k\in \Z} \in [0,\pi]^\Z$, the homogeneous functions
\begin{equation*}
\forall h\in \Z, \forall (x,t)\in \R^2\times\R,\quad \bar\varphi(x,h,t) = (-1)^{i_h}\bar\varphi,\quad \bar\theta(x,h,t) = \bar\theta_h,
\end{equation*}
define a homogeneous macroscopic equilibria.
\end{proposition}

\noindent The case of opposite vector flows may be non-stable since a small deviation from the equilibria leads to the global alignment of the layers. In particular, numerical simulations (see section \ref{section:hom_sim}) capture only equilibria with the same velocity in each layers. The question whether other stable macroscopic equilibria exists remains open.

\section{Numerical experiments}
\label{section:num_experiment}

In this section, we compare numerical simulations of both the microscopic and macroscopic models. The numerical methods are variations of those presented in Ref.~\cite{2011_MotschNavoret}: the microscopic model is solved with an implicit scheme and the macroscopic model is solved using the splitting method. The two methods are detailed in \ref{annex:num}.

\subsection{Homogenous simulations}

\subsubsection{Convergence to equilibria.} We consider $3$ layers and each layer contains $600$ particles. The particle positions are uniformly randomly distributed on the square $[0,L_x]\times[0,L_y]$ with $L_x = L_y=1$, the particle velocity direction angles $\varphi$ are uniformly randomly distributed on $\R/[0,2\pi]$ and the particle inclinations $\theta$ are uniformly randomly distributed on $\R/[0,\pi]$. We first consider an homogeneous test-case: the interaction radii $R_1$, $R_2$ and $R_3$ all equal $L_x/2$ and each particle thus interacts with (almost) all the other ones.

We first consider non-interacting layers supposing $h > 2R$. In Figure \ref{Fig:Non-interact-layers}, we plot the distributions of $\varphi$ and $\theta$ of the three layers. We also plot the von Mises distributions:
\begin{equation*}
 M_{1,\nu/D,\bar\varphi^\ell}^\ell (\varphi),\quad M_{2,K/\delta,\bar\theta^\ell}^\ell (\theta),\quad \ell \in \left\{1,2,3\right\}. 
\end{equation*}
where $\bar\varphi^\ell$ and $\bar\theta^\ell$ are the mean velocity angle and mean inclination angle of the particle of the $\ell$-th layer:
\begin{equation*}
e^{ i \bar\varphi^\ell} = \frac{\sum_{j, \, h_j = \ell} V_j(t)}{|\sum_{j, \, h_j = \ell} V_j(t)|},\quad e^{2 i \bar \theta^k} = \frac{\sum_{j, \, h_j = \ell} e^{2 i \theta_j(t)}}{|\sum_{j, \, h_j = h_k} e^{2 i \theta_j(t)}|}.
\end{equation*}
Both velocity and inclination distribution are in good agreement with the von Mises distributions. In Figure \ref{Fig:Non-interact-layers-mean}, we present the time evolution of the mean angles $\bar\varphi^\ell$ and $\bar\theta^\ell$. Since there are no layer interactions, these mean angles are almost constant in time up to stochastic fluctuations.

\begin{figure}[!h]
\centering
\includegraphics[width=0.5\textwidth]{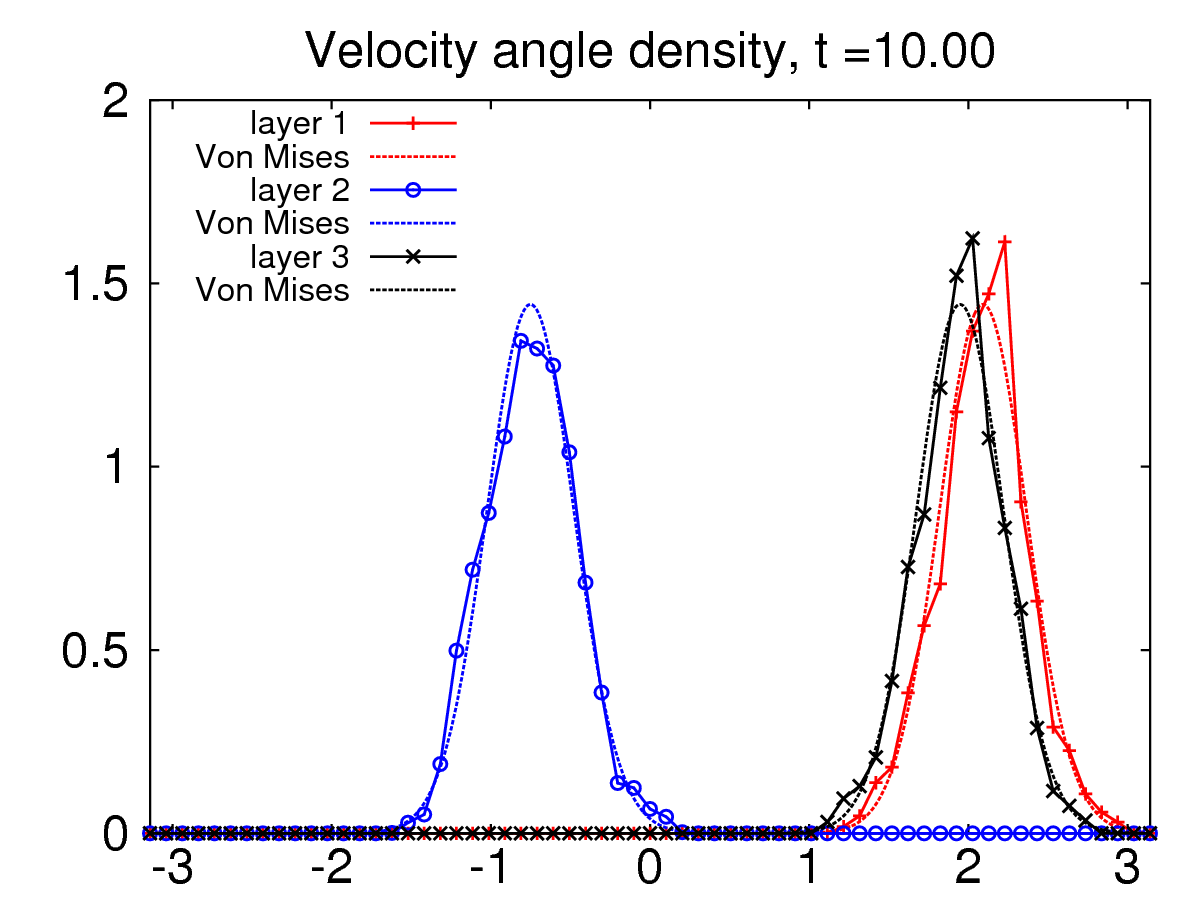}\includegraphics[width=0.5\textwidth]{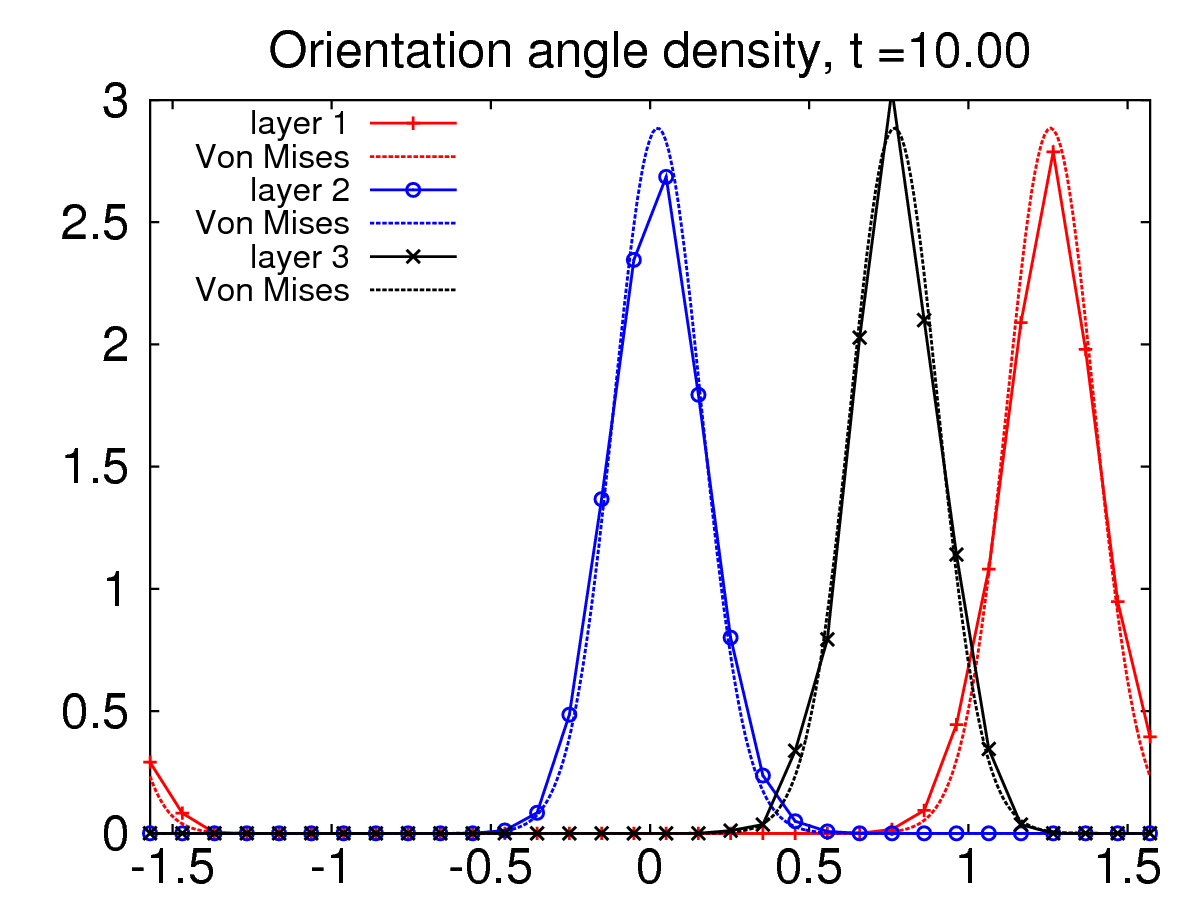}
\caption{(Homogeneous case, non-interacting layers) Left: Distributions of velocity angle $\varphi$ at time $t=10$. Right: Distributions of inclination angle $\theta$ at time $t=10$. Numerical parameters: $\nu = 4$, $D = 0.3$, $K=4$, $\delta = 0.3$, $\Delta t = 10^{-2}$. Number of particles: $600$ per layer.} 
\label{Fig:Non-interact-layers}
\end{figure}

\begin{figure}[!h]
\centering
\includegraphics[width=0.49\textwidth]{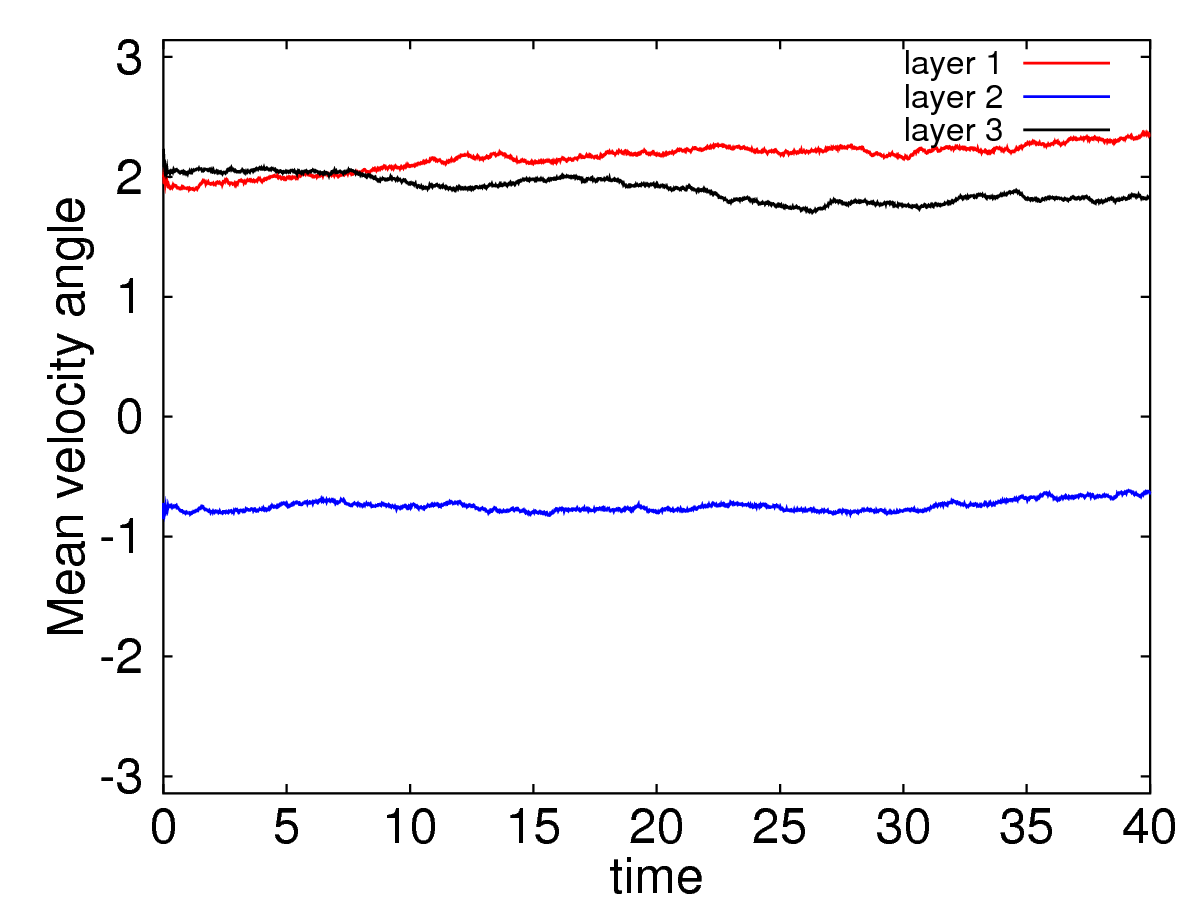}
\includegraphics[width=0.49\textwidth]{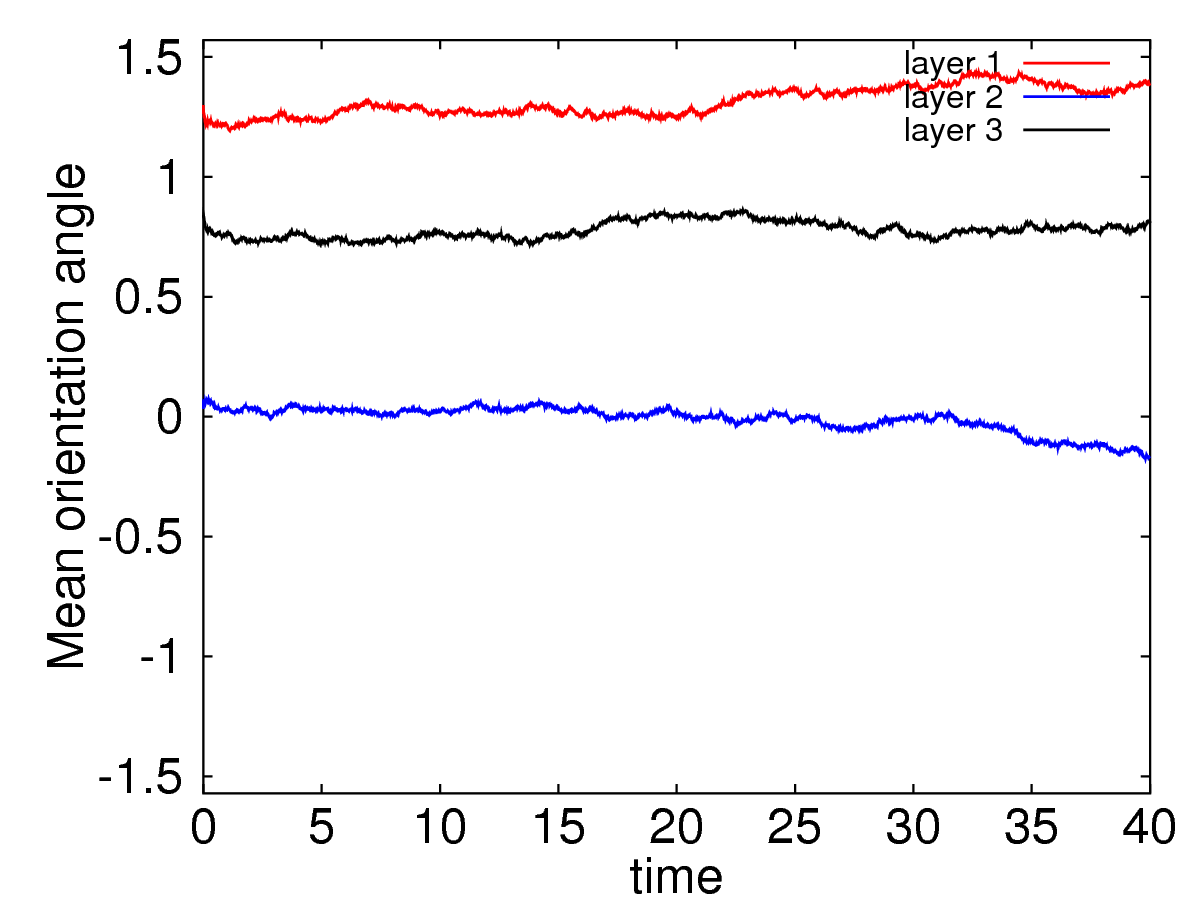}
\caption{(Homogeneous case, non-interacting layers) Left: Mean velocity angles $\bar\varphi^\ell$ as function of time. Right: Mean inclination angles $\bar\theta^\ell$ as function of time. Numerical parameters: $\nu = 4$, $D = 0.3$, $K=4$, $\delta = 0.3$, $\Delta t = 10^{-2}$. Number of particles: $600$ per layer.} 
\label{Fig:Non-interact-layers-mean}
\end{figure}

\subsubsection{Interactions between layers.}\label{section:hom_sim} We still consider $3$ layers but the number of particles per layer equals $25000$. The particle positions are uniformly randomly distributed on the square $[0,L_x]\times[0,L_y]$ with $L_x = L_y=1$ and the interaction radii $R_1$, $R_2$, $R_3$ equal $0.02$: there are in average $31$ neighboring particles. We include layer interactions: we suppose that the inter-layer distance $h$ equals the particle radius $R = 0.02$. The layer interaction coefficients are chosen as follows: $\mu = 3$ and $\nu = 0.5$. Particle velocity and inclination angles are randomly distributed according to their respective von Mises distribution. Initial mean velocity angle and mean inclination angle for the three layers are chosen as follows:
\begin{align}
&\bar\varphi^1 = -1,&&\bar\varphi^2 = 1,&&\bar\varphi^3 = -2,\label{init-cond-phi}\\
&\bar\theta^1 = 0,&&\bar\theta^2 = -0.9,&&\bar\theta^3 = -0.8. 
\label{init-cond-theta}
\end{align}
We also perform a time rescaling in the microscopic model. Let $\eps > 0$ and consider the following microscopic parameters:
\begin{equation}
\nu^\eps = \nu/\eps,\quad D^\eps = D/\eps,\quad K^\eps = K/\eps,\quad \delta^\eps = \delta/\eps,\quad \text{ and }\quad \beta^\eps = \eps\beta.\label{eq:rescaling}
\end{equation}
with $\eps = 0.1$. Consequently, we choose a macroscopic time scale. We compare particle simulations with macroscopic simulation. For the macroscopic model, we thus consider  a constant initial density in each layer given by:
\begin{equation*}
\rho^\ell = \frac{(\text{Number of particles})}{L_x L_y},\quad\text{ for }\ell \in \left\{1,2,3\right\}.
\end{equation*}
and the uniform initial values of $\bar\varphi^\ell$ and $\bar\theta^\ell$ given by  \eqref{init-cond-phi} - \eqref{init-cond-theta}. The macroscopic layer interaction coefficients are given by:
\begin{equation*}
\beta' = \beta,\quad \mu' = \mu R\,\pi R_3^2,
\end{equation*}
with no $\eps$, since the particle simulation already consider the macroscopic time scale. 

Fig. \ref{Fig:interact-layers-mean-macro-test3} (top) depicts the time evolution of the mean velocity and mean inclination for both the microscopic simulation (dashed line) and macroscopic simulation (continuous line). Microscopic simulations are averaged of $20$ particle simulations. Let us first describe the macroscopic dynamics. The dynamics can be split into two steps: during the first step, up to time $t\approx 2$, layer 2 mainly interact with layer 1 since the overlap function is more important between this two layers. This leads to a first relaxation dynamics that make the mean velocities of these two layers align. Then, during the second step, after time $t=2$, interactions between layers 2 and 3 becomes predominant and a second relaxation dynamics occur that leads to alignment of the three layers. We then note that microscopic and macroscopic simulations coincide during the first $1.5$ time unit (including the first relaxation mechanism). This confirms that the macroscopic model captures the right interaction time scale. However, we see that, due to the finite number of particles, stochastic fluctuations make the second relaxation occur earlier around time $t=3$ (micro) instead of $t=3.5$ (macro). Looking at Figures \ref{Fig:interact-layers-mean-macro-test3} (bottom), where $10$ particle simulations are plotted and compared to macroscopic simulation, we see that the time of the second relaxation depends on the simulation and occurs always before the macroscopic relaxation. As noted in section \ref{sec:macro_equilibria}, once the particle velocities of the three layers are aligned, homogeneous inclination angles per layer define equilibria. Therefore, the time of the second relaxation step strongly determines the final inclinations. This explains the large deviation between macro and micro simulations after time $t\geq 4$ as regards inclinations.

We now conserve the same parameters except that $h = 0.0205 > R = 0.02$ and, consequently, some inclination configuration prevent layers from interacting. In Fig. \ref{Fig:interact-layers-mean-macro-test8}, we plot the time evolution of the mean velocity angle and mean inclination angle. We observe that, in the macroscopic simulation, a slight increase of the inter-layer distance results in large time translation of the second relaxation step going from $t=3.5$ (Fig. \ref{Fig:interact-layers-mean-macro-test3}) to $t=7.5$ (Fig. \ref{Fig:interact-layers-mean-macro-test8}). This highlights the meta-stability of the system between the two relaxation steps. Concerning the particle simulations (in dashed lines), the slight increase of the inter-layer distance does not result in a so much increase of the second relaxation time (it goes from $t=3$ (Fig. \ref{Fig:interact-layers-mean-macro-test3}) to only $t=3.5$ (Fig. \ref{Fig:interact-layers-mean-macro-test8})).  The second relaxation thus occurs two times earlier than predicted by the macroscopic simulation. Indeed, due to stochastic fluctuations, some particles interact instead of remaining in non-interacting configuration. Therefore, this is the stochastic fluctuations that impacts the long term dynamics of the model.

\begin{figure}[h]
\centering
\includegraphics[width=0.5\textwidth]{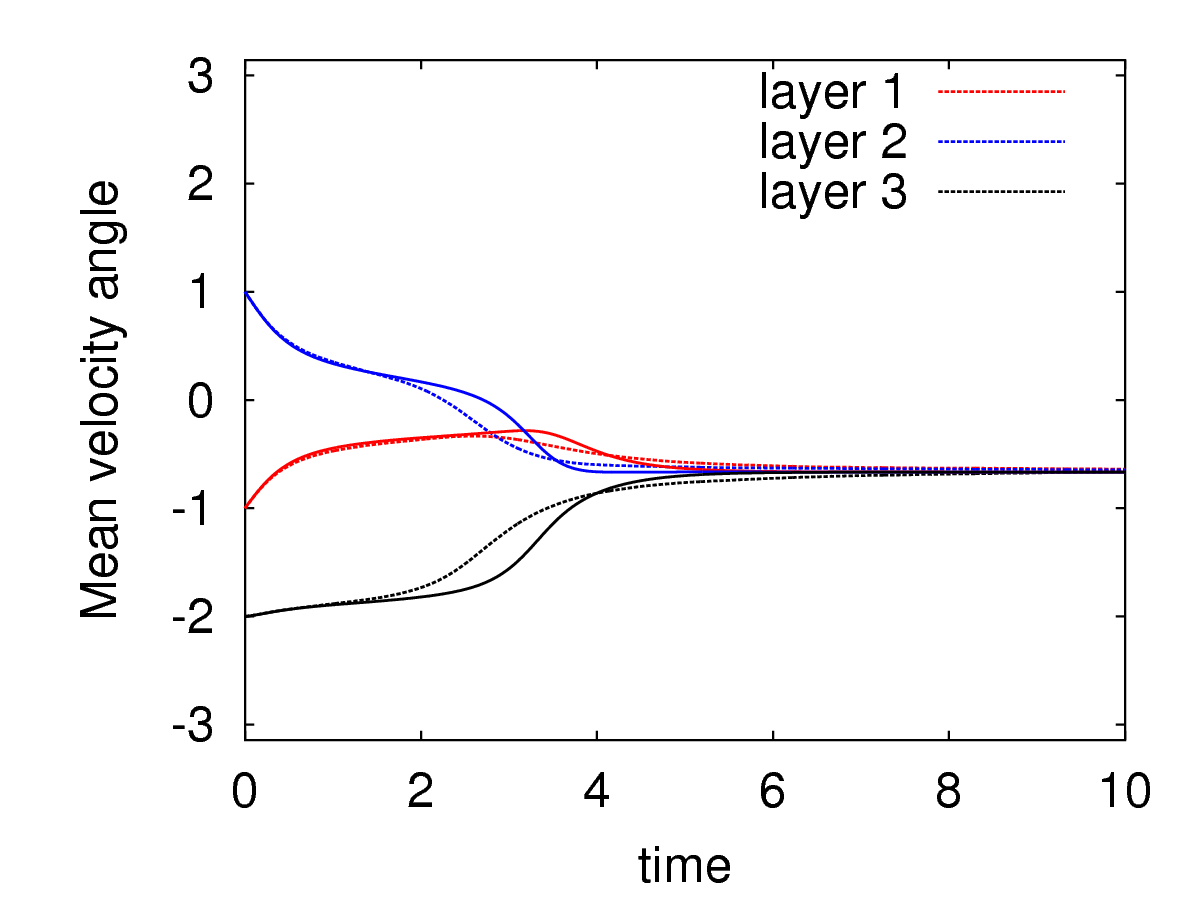}\includegraphics[width=0.5\textwidth]{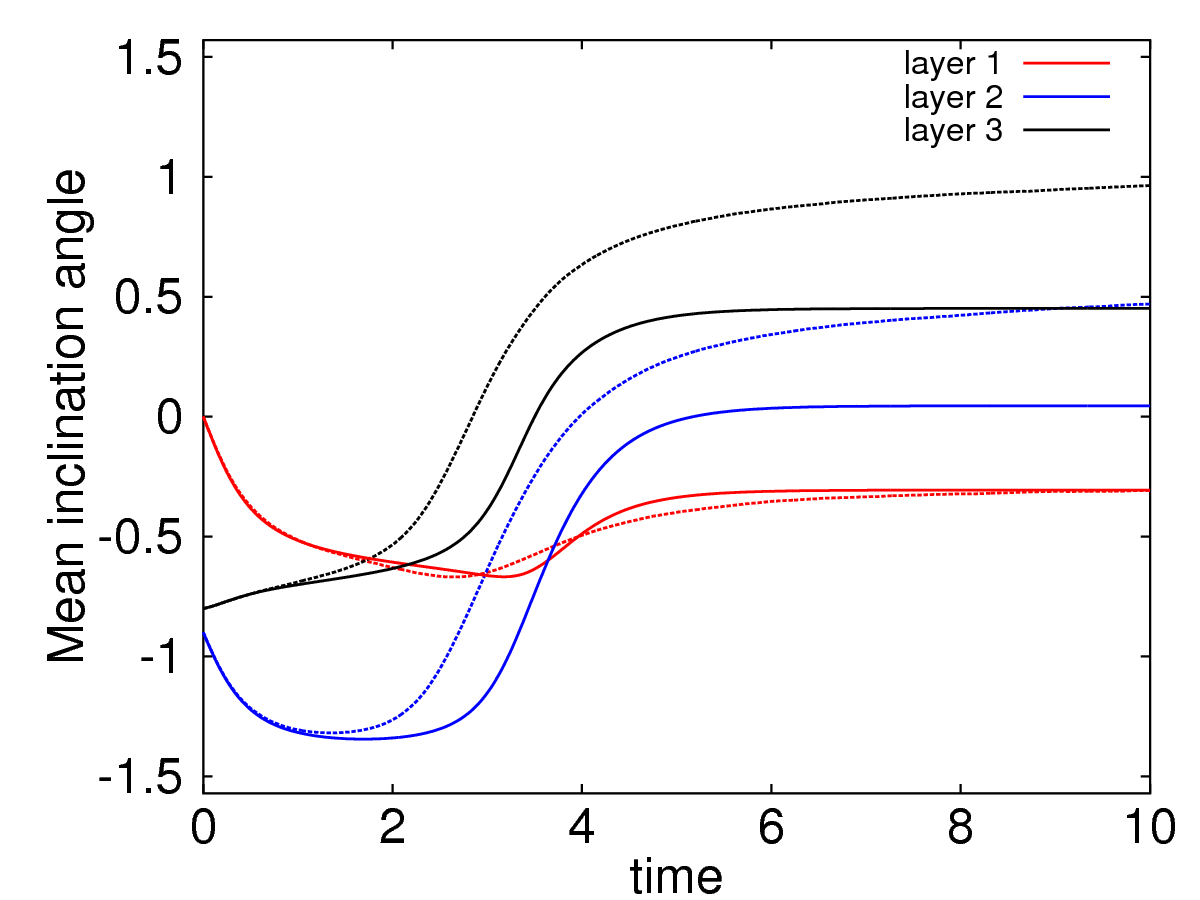}

\includegraphics[width=0.5\textwidth]{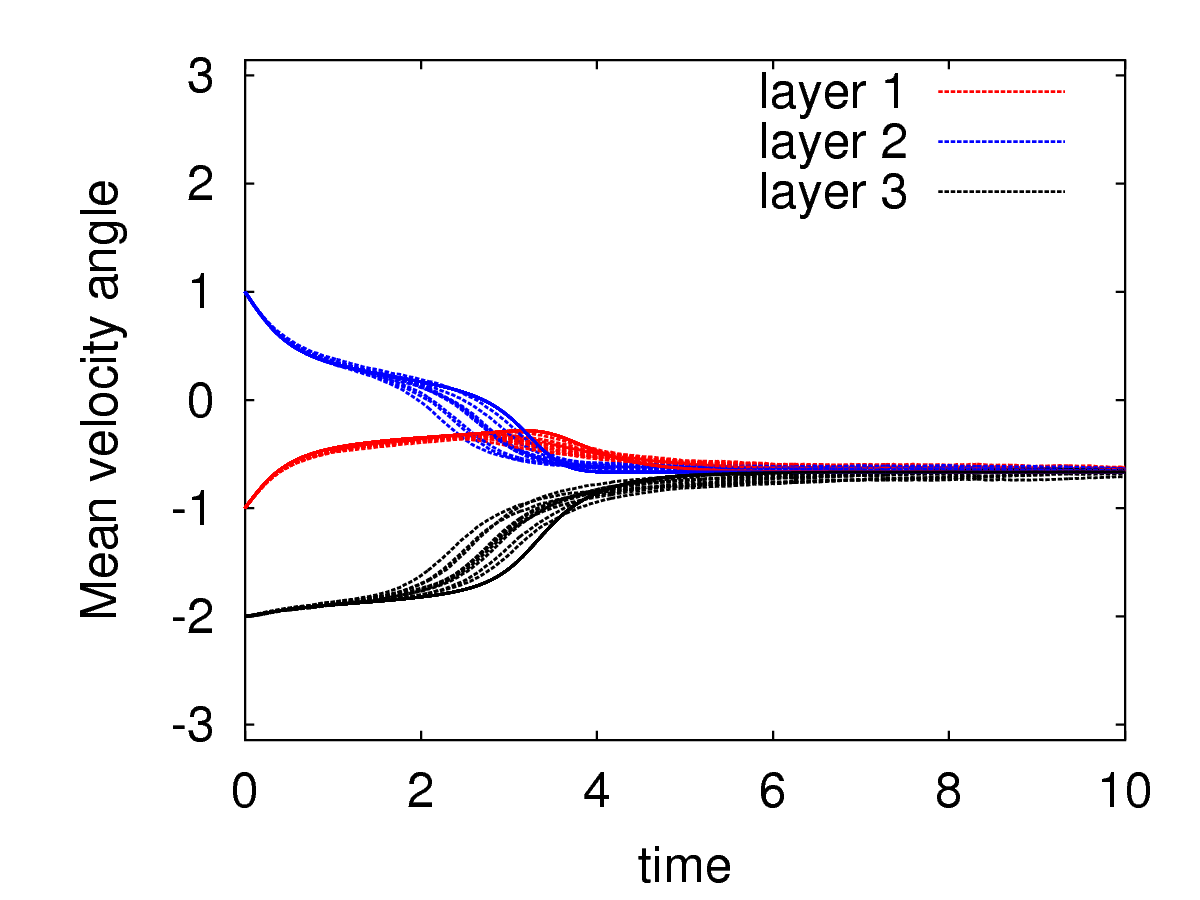}\includegraphics[width=0.5\textwidth]{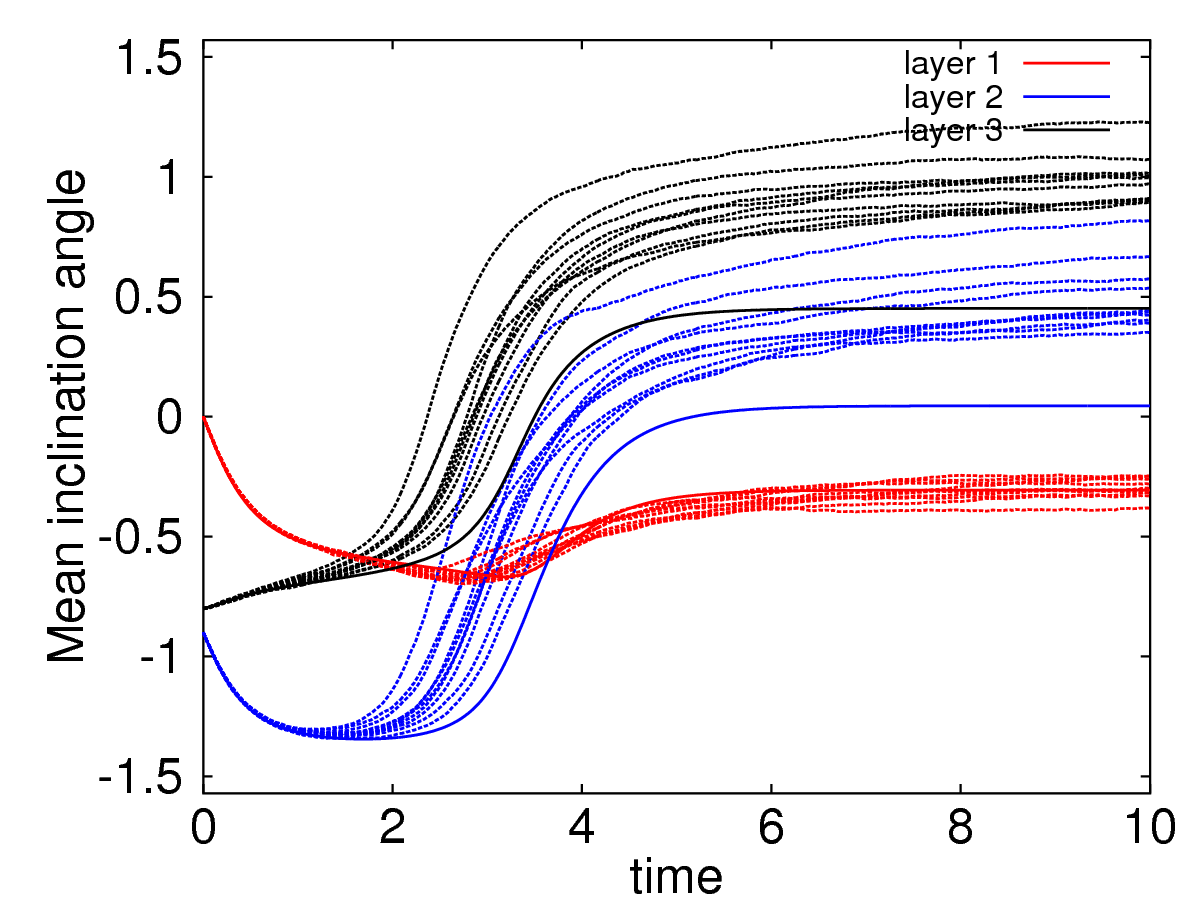}
\caption{(Homogeneous case, interacting layers, $h=R$) Comparison of macroscopic (continuous line) and microscopic (dashed line) simulations. Left: Mean velocity angles $\bar\varphi^\ell$ as function of time. Right: Mean inclination angles $\bar\theta^\ell$ as function of time. Up: comparison with the average of $20$ microscopic simulations. Down: comparison with $10$ microscopic simulations. Alignment interaction parameters: $\nu = 4$, $D = 0.1$, $K=4$, $\delta = 0.1$, $R_1 = R_2 = R$. Layer-interaction parameters: $h= 0.02$, $R=0.02$, $\beta= 0.5$, $R_3 = R$. Microscopic parameters: $25000$ particles per layer, $\mu = 3$, $\Delta t = 1\times 10^{-2}$. Macroscopic parameters: $\Delta x = \Delta y = 0.5$, $\Delta t = 1\times 10^{-2}$. } 
\label{Fig:interact-layers-mean-macro-test3}
\end{figure}

\begin{figure}[h]
\centering
\includegraphics[width=0.5\textwidth]{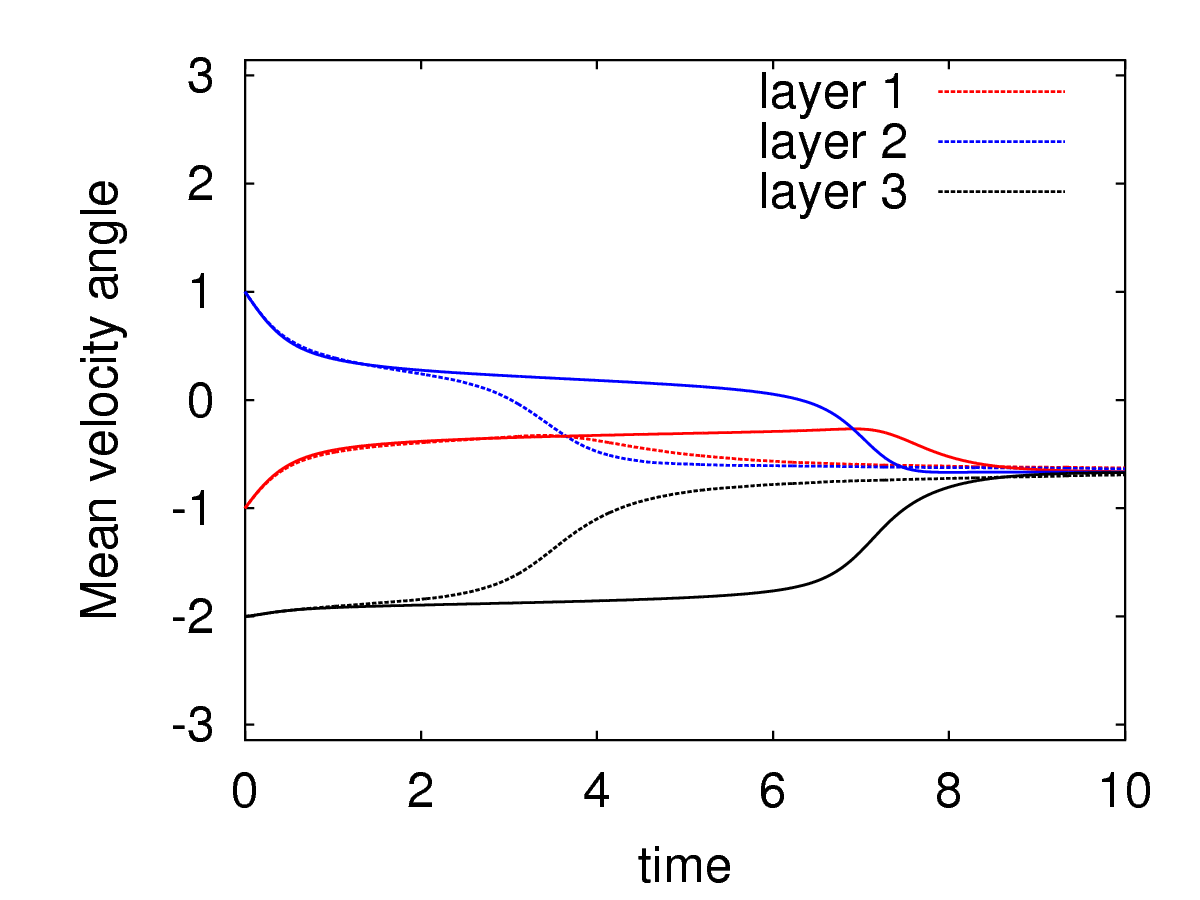}\includegraphics[width=0.5\textwidth]{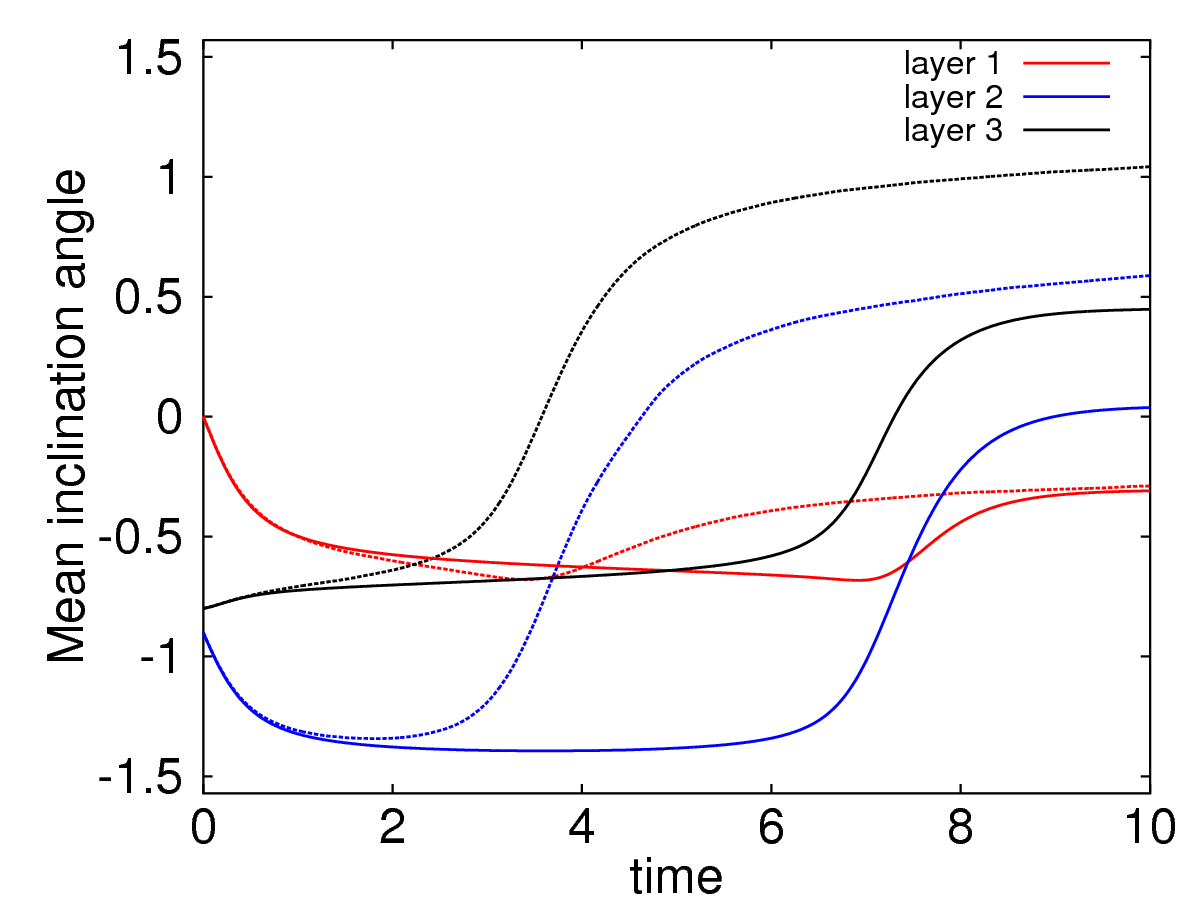}

\includegraphics[width=0.5\textwidth]{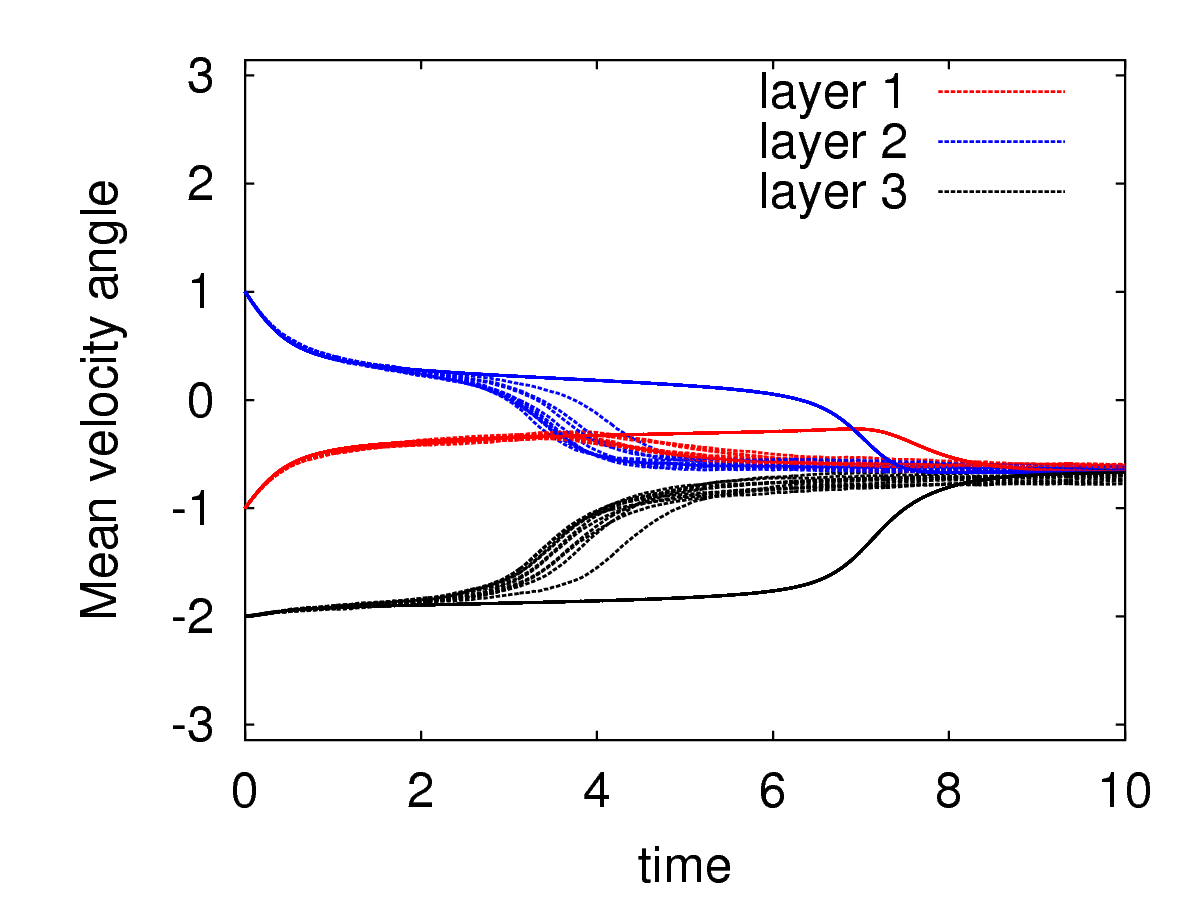}\includegraphics[width=0.5\textwidth]{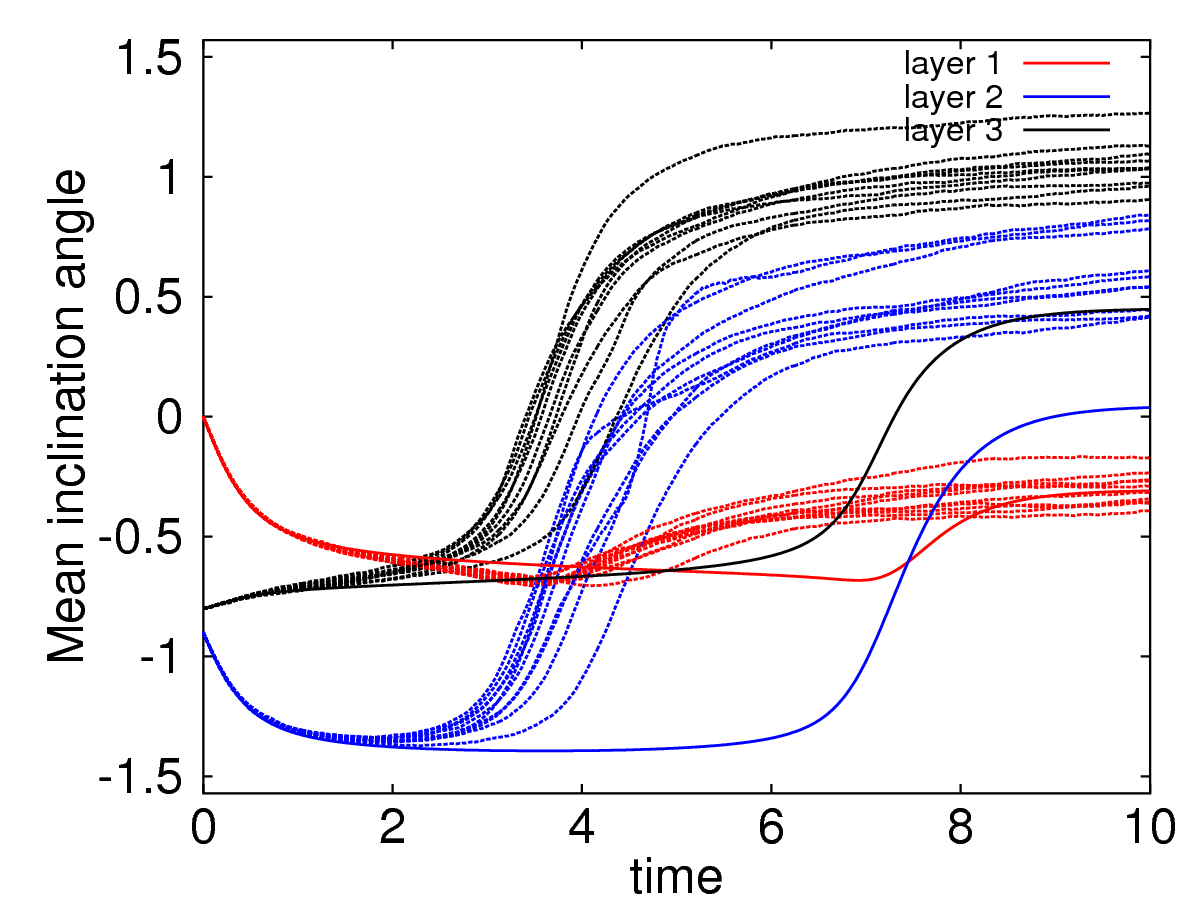}
\caption{(Homogeneous case, interacting layers, $h > R$) Comparison of macroscopic (continuous line) and microscopic (dashed line) simulations. Left: Mean velocity angles $\bar\varphi^\ell$ as function of time. Right: Mean inclination angles $\bar\theta^\ell$ as function of time. Up: comparison with the average of $20$ microscopic simulations. Down: comparison with $10$ microscopic simulations. Alignment interaction parameters: $\nu = 4$, $D = 0.1$, $K=4$, $\delta = 0.1$, $R_1 = R_2 = R$. Layer-interaction parameters: $h= 0.0205$, $R=0.02$, $\beta= 0.5$, $R_3 = R$. Microscopic parameters: $25000$ particles per layer, $\mu = 3$, $\Delta t = 1\times 10^{-2}$. Macroscopic parameters: $\Delta x = \Delta y = 0.5$, $\Delta t = 1\times 10^{-2}$. } 
\label{Fig:interact-layers-mean-macro-test8}
\end{figure}

\subsection{Inhomogenous simulations}

We now consider $3$ layers on the square domain $[0,L_x]\times[0,L_y]$ with $L_x = L_y=10$.
As in \cite{2015_Degond}, we are interested in Taylor-Green vortex initial condition. Initial densities are taken uniform equal to $10^3$. Velocity angles are given by :
\begin{equation*}
\varphi^\ell(x,y) = \widehat{\begin{pmatrix}
u(\bar x_\ell, \bar y_\ell)\\
v(\bar x_\ell, \bar y_\ell)
\end{pmatrix}},\quad\text{ for }\ell \in \left\{1,2,3\right\},
\end{equation*}
where $\widehat w$ denotes the angle between vectors $w \in \R^2$ and $(1,0)^T$, $(u,v)$ is the vector defined by:
\begin{align*}
&u(x,y)= \frac{1}{3} \sin\left(\frac{\pi}{5} x\right)\cos\left(\frac{\pi}{5} y\right)+  \frac{1}{3} \sin\left(\frac{3\pi}{10} x\right)\cos\left(\frac{3\pi}{10} y\right) +\frac{1}{3} \sin\left(\frac{\pi}{2} x\right)\cos\left(\frac{\pi}{2} y\right),\\
&v(x,y)= -\frac{1}{3} \cos\left(\frac{\pi}{5} x\right)\sin\left(\frac{\pi}{5} y\right) -  \frac{1}{3} \cos\left(\frac{3\pi}{10} x\right)\sin\left(\frac{3\pi}{10} y\right) -\frac{1}{3} \cos\left(\frac{\pi}{2} x\right)\sin\left(\frac{\pi}{2} y\right),
\end{align*}
and $(\bar x_\ell, \bar y_\ell)$ are translation of $(x,y)$ :
\begin{equation*}
\bar x_\ell = x-t^\ell_x \mod L_x,\quad \bar y_\ell  = y-t^\ell_y \mod L_y,
\end{equation*}
with translation vectors  $(t^\ell_x, t^\ell_y)$ given by:
\begin{equation*}
(t^1_x, t^1_y) = (0,0),\quad (t^2_x, t^2_y) = (2,2),\quad (t^3_x, t^3_y) = (5,2).
\end{equation*}
As regards the velocity initial condition, each layer is thus the translation of a normalized Taylor-Green vortex. Consequently, layers velocity fields are not initially the same and alignment dynamics should arise. Finally, inclination angles are taken uniform with the same value as the previous test-case \eqref{init-cond-theta}.

\subsubsection{Non-interacting layers} We first consider non-interacting layers : $h > 2R$. This test-case thus reduces to a simulation of the SOH model. In Figure \ref{Fig:noninteract-layers-TG-velocity}, we plot the space distribution of density, velocity and inclination for both the particle (left figures) and macroscopic simulations (right figures) for layer $2$ (first and third layer are identical up to translation). For the particle simulation, we consider that each layer contains $10^5$ particles. The interaction radii $R_1$, $R_2$ and $R_3$ equal $0.04$. Therefore, the particles have in averaged $\pi R_1^2 10^3 \approx 5$ neighboring particles in each layer at the beginning of the simulation.  We consider the same rescaling \eqref{eq:rescaling} with $\eps = 0.1$. Densities are computed on a grid with space steps equal to $\Delta x = \Delta y = 0.2$ and as regards the particle simulation, they are averaged over $20$ runs of the test-case. We use the same time step for both macroscopic and microscopic simulations.

In Fig. \ref{Fig:noninteract-layers-TG-velocity} (top), we observe clustering for both microscopic and macroscopic simulation in region of negative divergence flow. The density in layer $2$ has maximal value equal to $\|\rho_2\|_\infty =  6520$ for the microscopic simulations and $\|\rho_2\|_\infty =  7547$ for the macroscopic simulation. Consequently the number of neighboring particles is multiplied by a factor $10$ in these regions. We also observe that the vortex are better conserved with the particle simulations. Concerning the velocity field, the microscopic velocity is obtained by dividing the local momentum by the local density. The velocity vectors should be of size $c_1$ but this is not the case in low density regions due to the small number of particles. This partly explains the observed differences between the microscopic and macroscopic simulations.  However, there is a quite good overall agreement between the two simulations. Fig. \ref{Fig:noninteract-layers-TG-velocity} (bottom) is represented the cosine (in absolute value) of the inclination angle: as layer interactions do not occur, it remains uniform.

\begin{figure}[!h]
\centering
\includegraphics[width=0.45\textwidth]{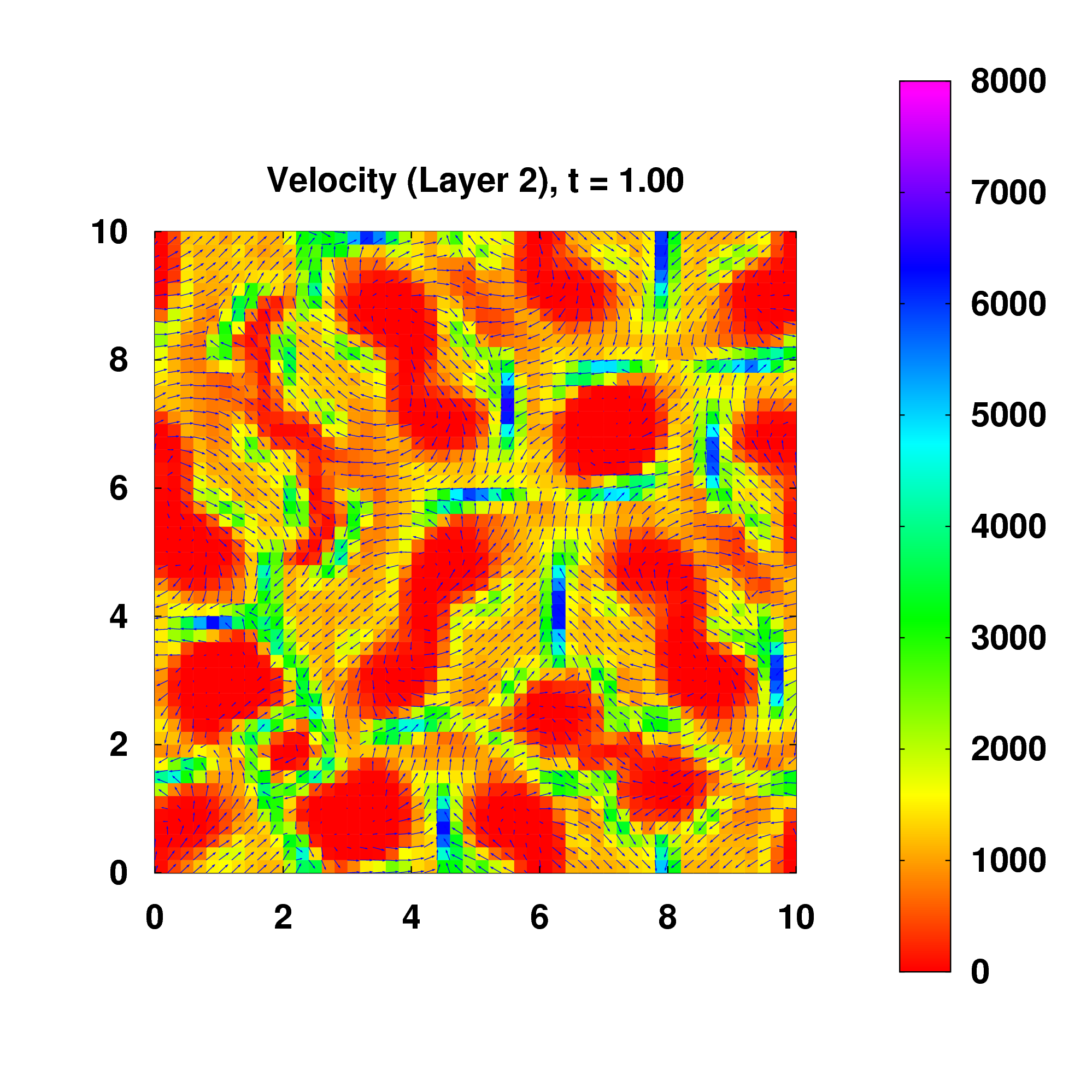}\includegraphics[width=0.45\textwidth]{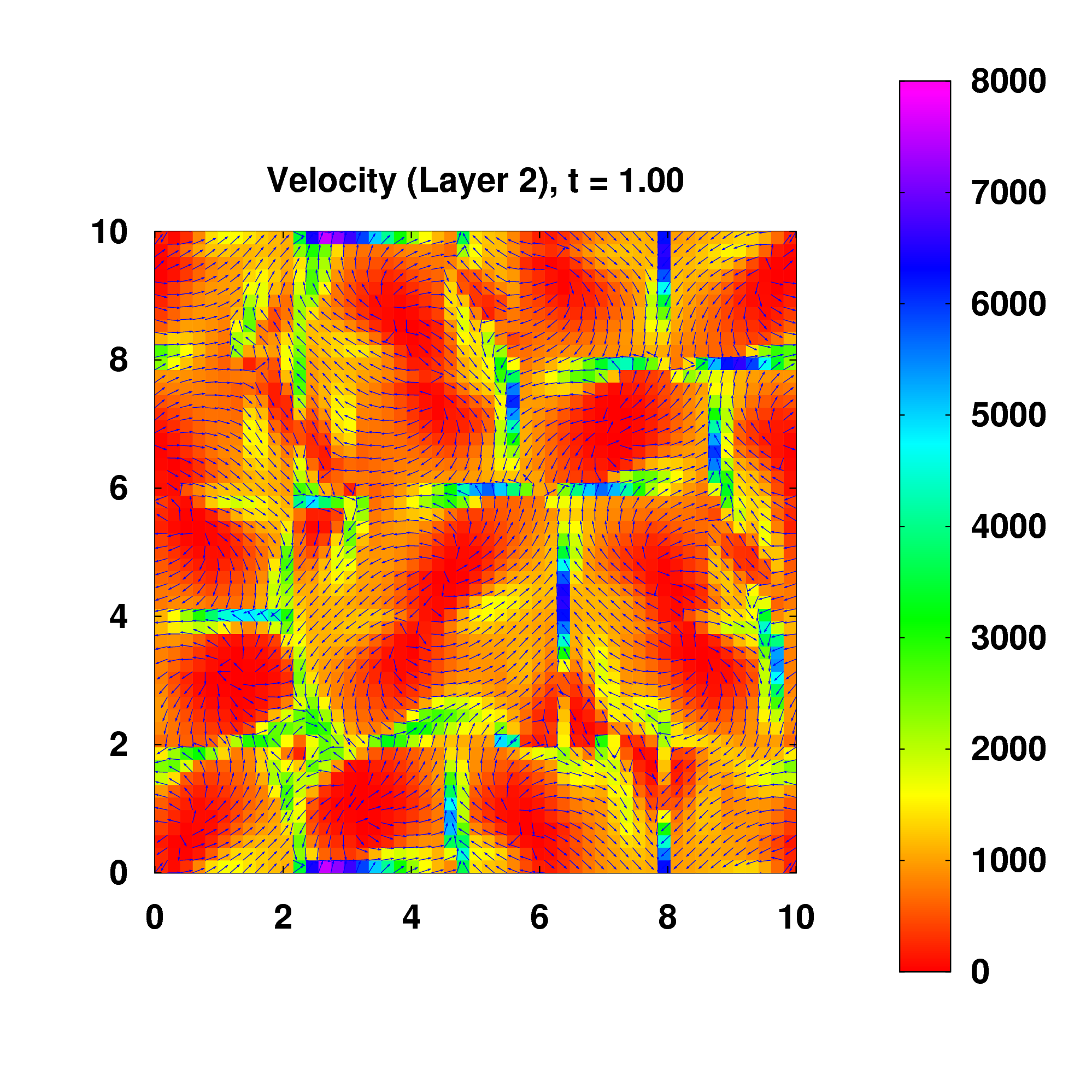}
\includegraphics[width=0.45\textwidth]{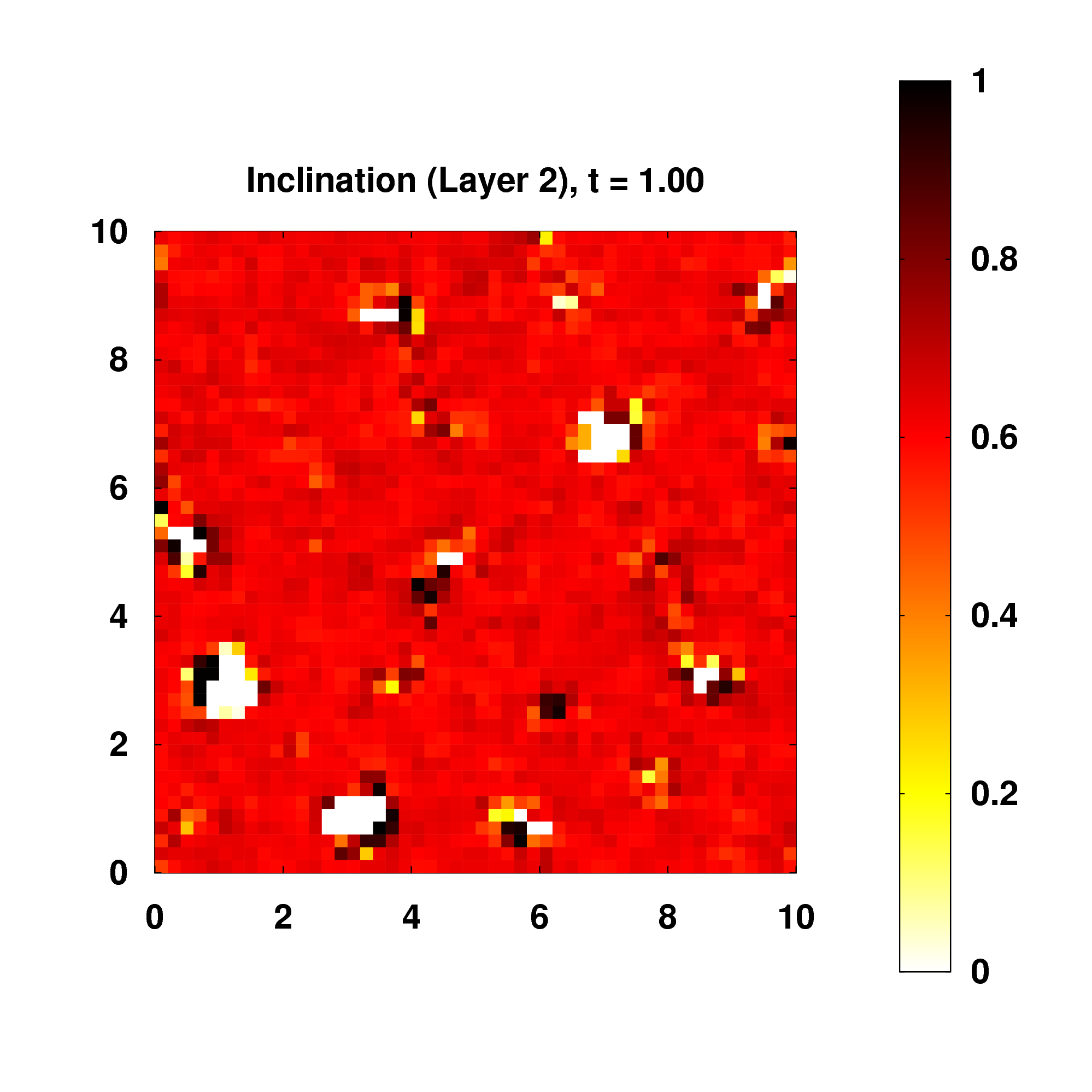}\includegraphics[width=0.45\textwidth]{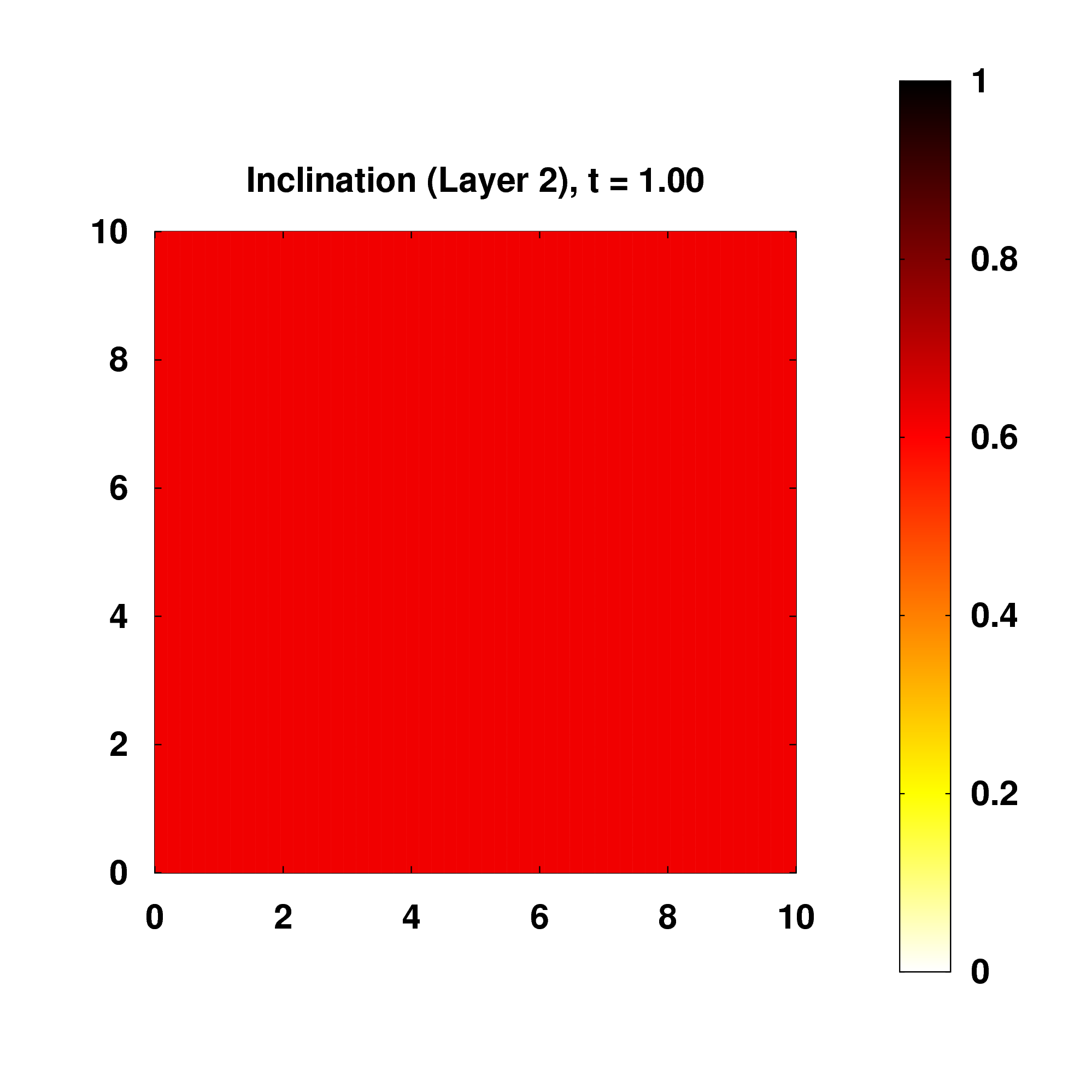}
\caption{(Inhomogeneous case, non-interacting layers, velocity) Top: Velocity vector field. Down: Cosine of the inclination angle (in absolute value). Left: Particle simulation with $\eps = 0.1$ (averaged on $20$ simulations). Number of particles: $10^5$ per layer. $\Delta t = 1\times 10^{-2}$. Right: Macroscopic simulation. $\Delta t = 1\times 10^{-2}$, $\Delta x = \Delta y = 0.2$. Alignment interaction parameters: $\nu = 4$, $D = 0.1$, $K=4$, $\delta = 0.1$. Layer-interaction parameters: $h= 1 >R=0.02$, $R_1 = R_2 = R_3 = R$.} 
\label{Fig:noninteract-layers-TG-velocity}
\end{figure}


\subsubsection{Interacting layers} We then consider interacting layers when setting $h=R$. The layer interaction parameters are taken equal to:  $\beta= 2$, $\mu=20$. The other parameters are the same as previously. The parameter $\mu$ is chosen in such a way that the interaction coefficient $\mu R\, \pi R_3^2 \rho_0 \approx 0.5$ is of the same order as in section \ref{section:hom_sim}. We compare the average of $20$ particle simulations with one macroscopic simulation on Fig. \ref{Fig:interact-layers-TG-velocity} and \ref{Fig:interact-layers-TG-inclination}.  

Fig.~\ref{Fig:interact-layers-TG-velocity} depicts the density and the velocity vector field for the three superposed layers. As in the previous test-case, we observe that the microscopic and macroscopic simulations are in good agreement. The velocity vector fields of the three layers are mostly aligned and consequently the density concentrations are localized almost in the same regions. This is particularly true for the macroscopic simulations (right). Concerning the microscopic simulations (left), we still observe some differences between the layers. 

In Fig. \ref{Fig:interact-layers-TG-inclination}, we represent the cosine of the inclination angle (in absolute value). The inclination for particle simulations is obtained by computing the local mean inclination angle with formula \eqref{eq:nematic_alignment}. Due to layer interactions, the inclination is no more uniform. Contrary to the velocity vector field, we observe large differences between the macroscopic and particle inclinations. This could be a consequence of the differences pointed out in section \ref{section:hom_sim}. Note that regions with aligned inclinations do not necessarily match regions of  uniform densities. Finally, the inter-layer interactions on inclinations affect in return the density and velocity vector field. This is particularly clear when comparing the density of Layer 2 with the one of the non-interacting test-case (Fig.~\ref{Fig:noninteract-layers-TG-velocity} (top)) for the macroscopic simulation. All the symmetries inherited from the initial distribution have been diluted by the inclination/velocity inter-layer interactions.

\begin{figure}[!h]
\centering
\includegraphics[width=0.45\textwidth]{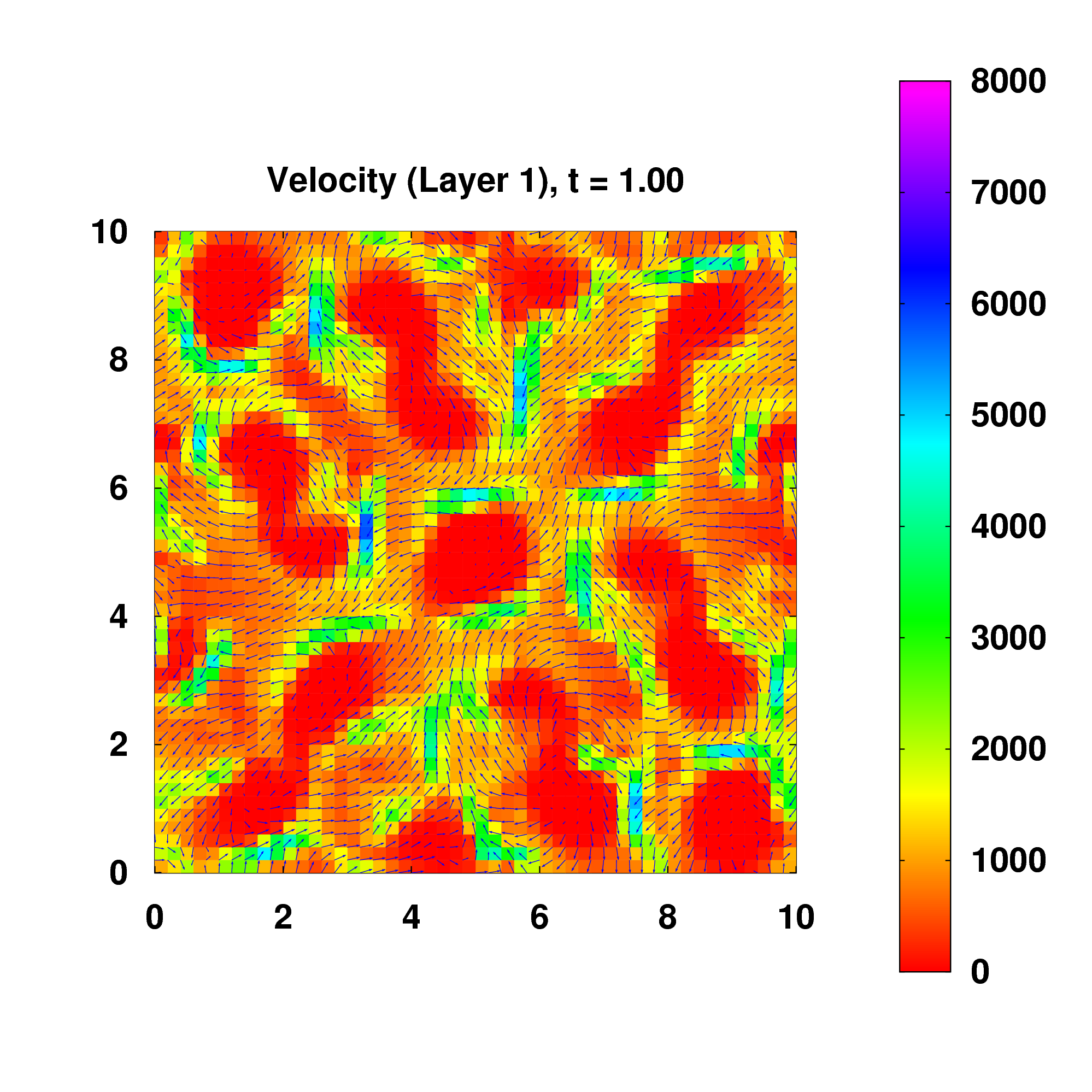}\includegraphics[width=0.45\textwidth]{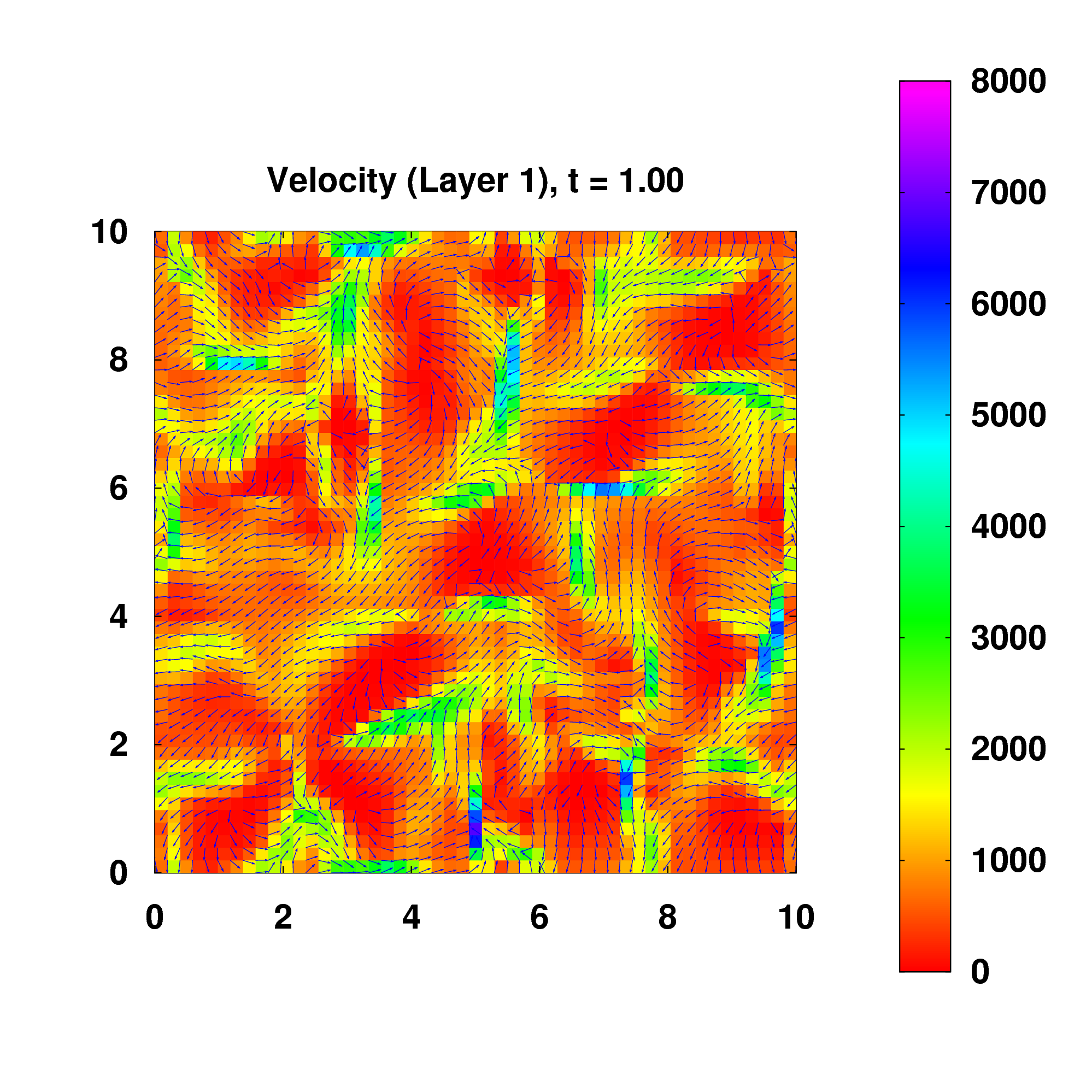}
\includegraphics[width=0.45\textwidth]{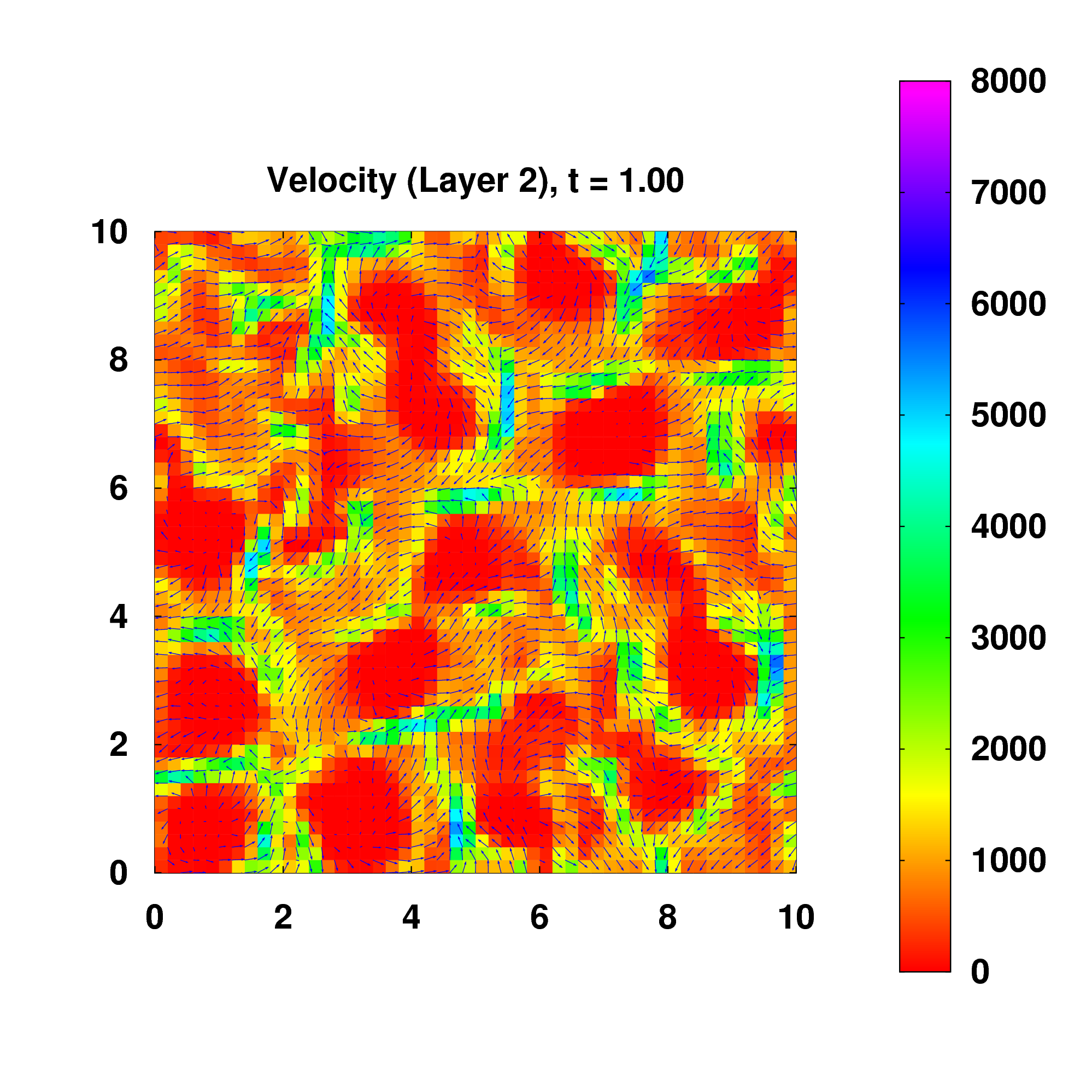}\includegraphics[width=0.45\textwidth]{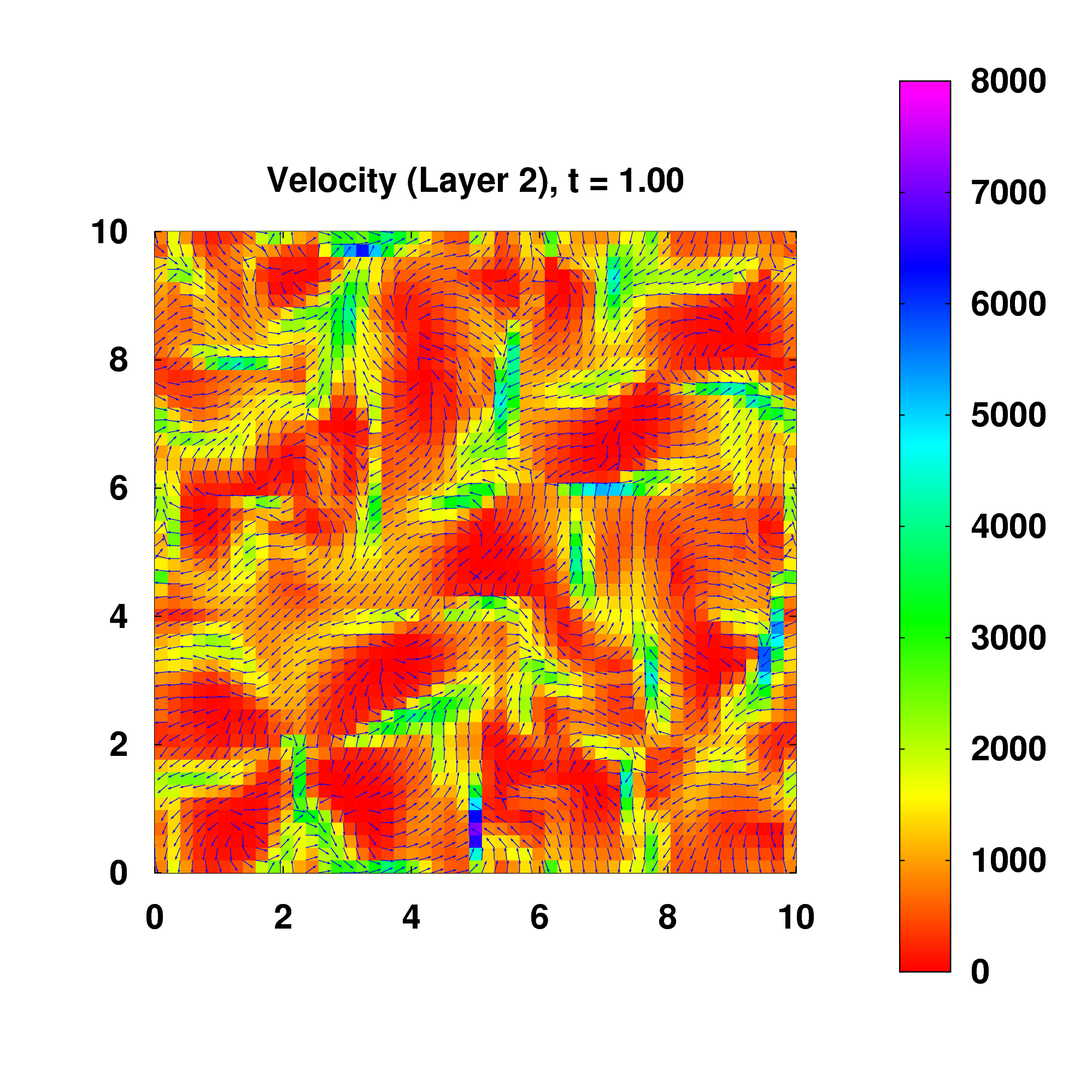}
\includegraphics[width=0.45\textwidth]{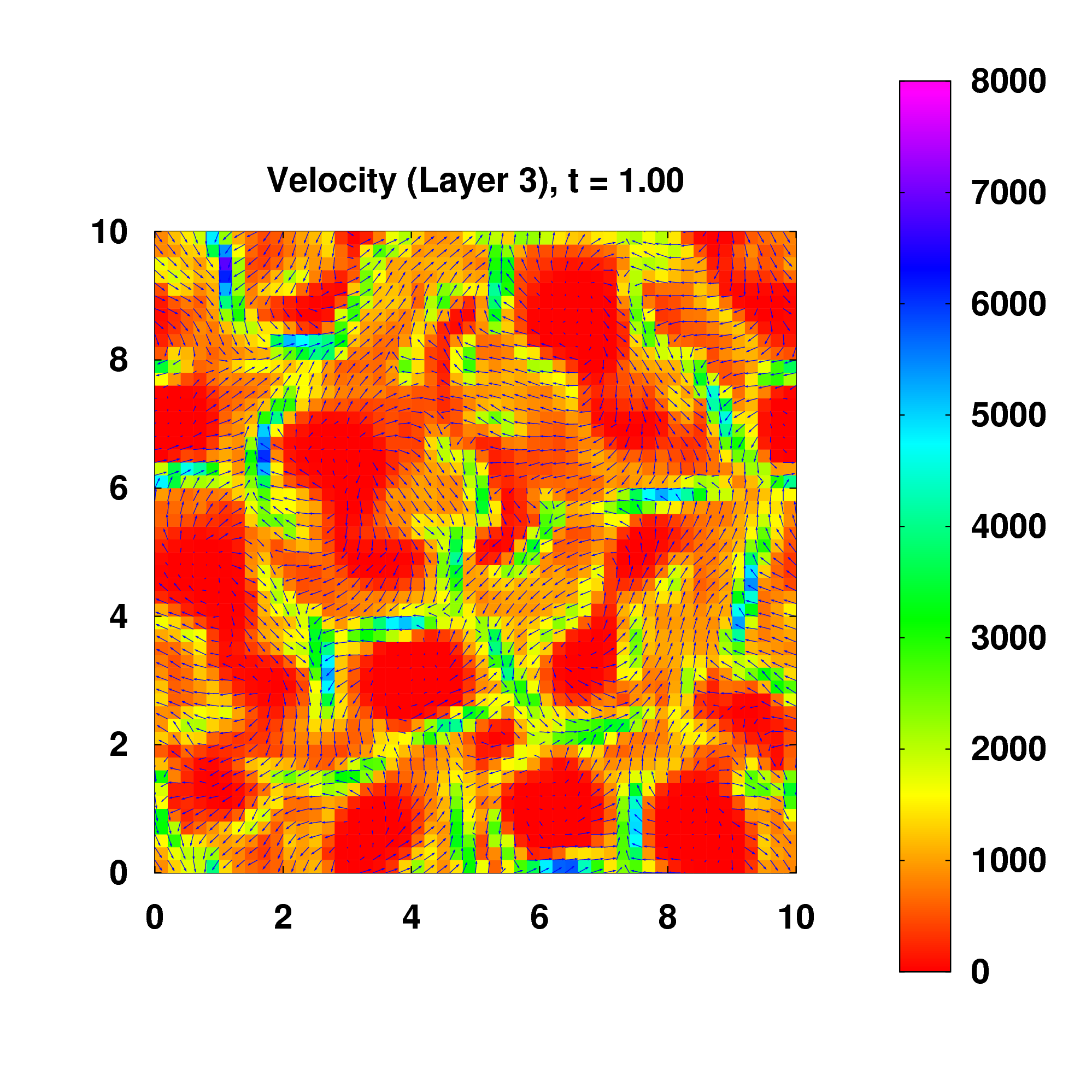}\includegraphics[width=0.45\textwidth]{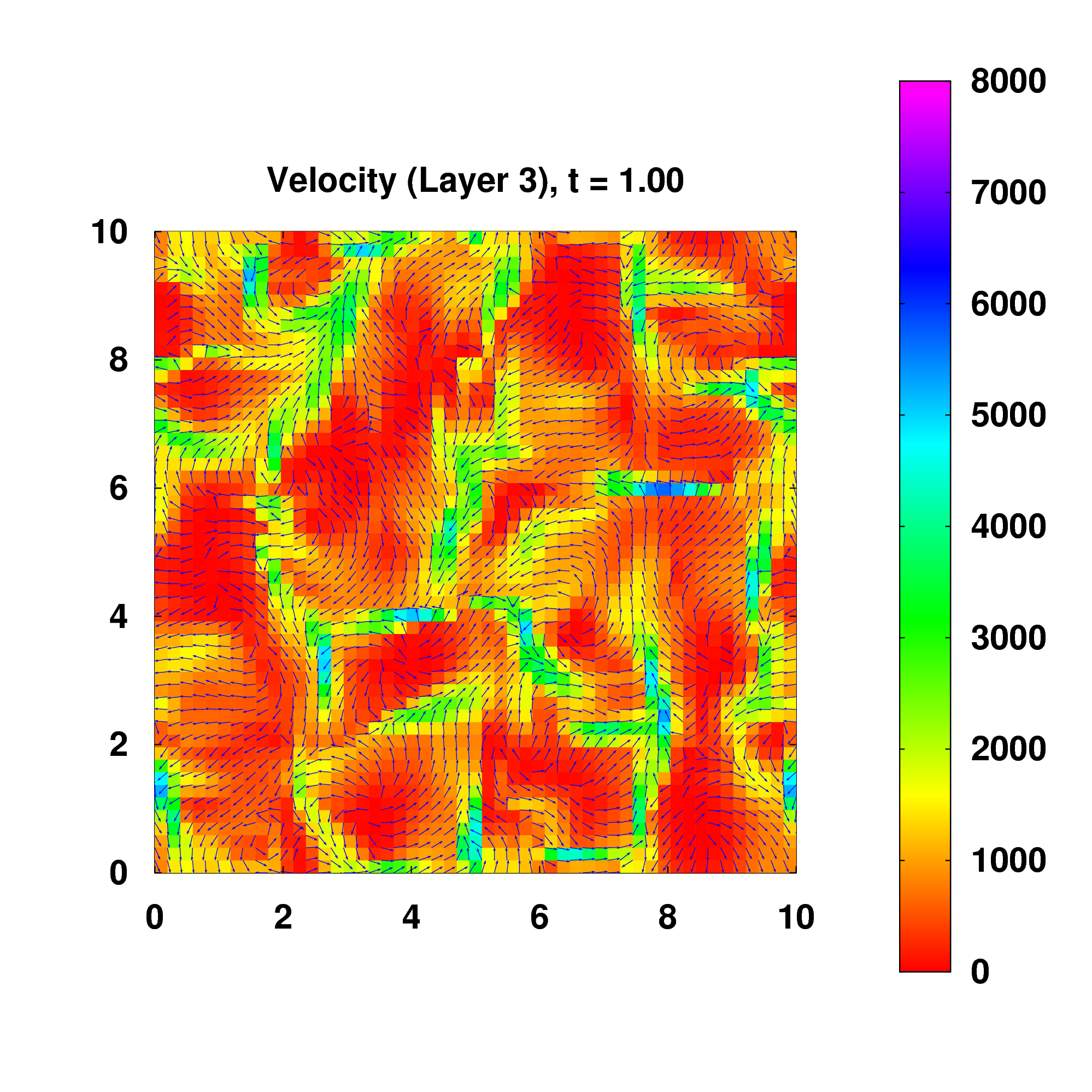}
\caption{(Inhomogeneous case, interacting layers, velocity) Left: Velocity vector field for the particle simulation with $\eps = 0.1$ (averaged on $20$ simulations). Number of particles: $10^5$ per layer. $\Delta t = 1\times 10^{-2}$. Right: Velocity vector field for the macroscopic simulation. $\Delta t = 1\times 10^{-2}$, $\Delta x = \Delta y = 0.2$. Alignment interaction parameters: $\nu = 4$, $D = 0.1$, $K=4$, $\delta = 0.1$. Layer-interaction parameters: $h = R = 0.02$, $\beta= 2$, $\mu=20$, $R_1 = R_2 = R_3 = R$.} 
\label{Fig:interact-layers-TG-velocity}
\end{figure}

\begin{figure}[!h]
\centering
\includegraphics[width=0.45\textwidth]{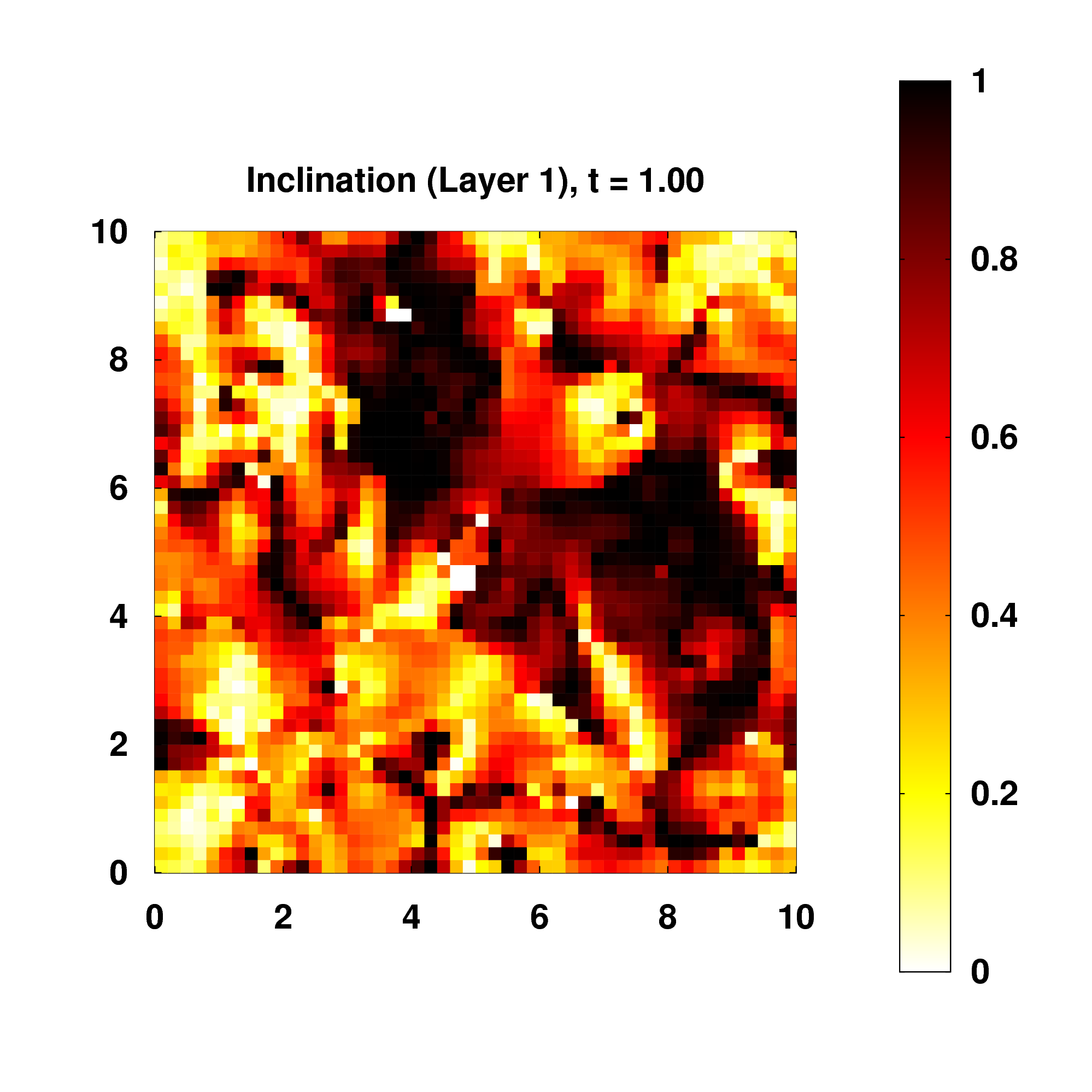}\includegraphics[width=0.45\textwidth]{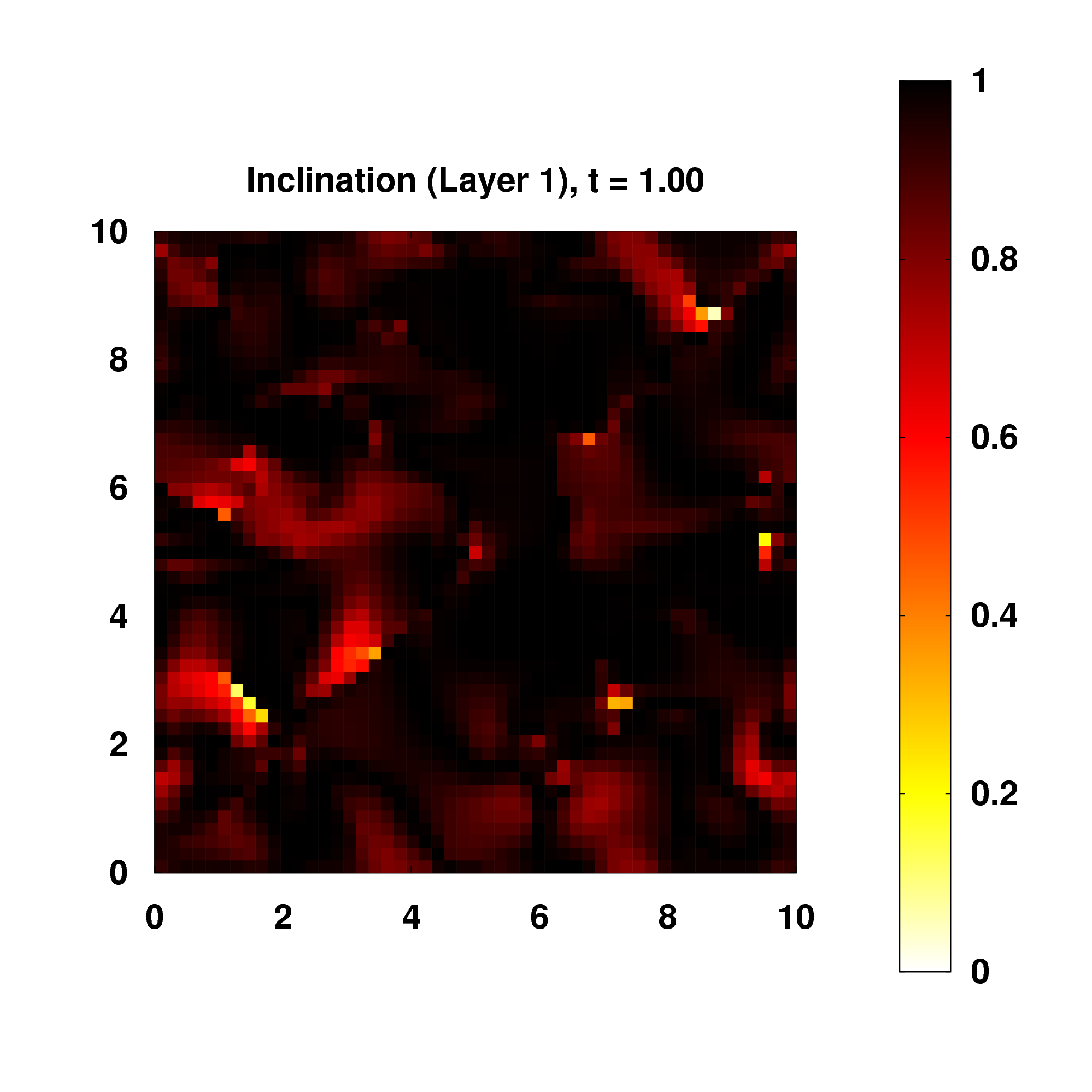}
\includegraphics[width=0.45\textwidth]{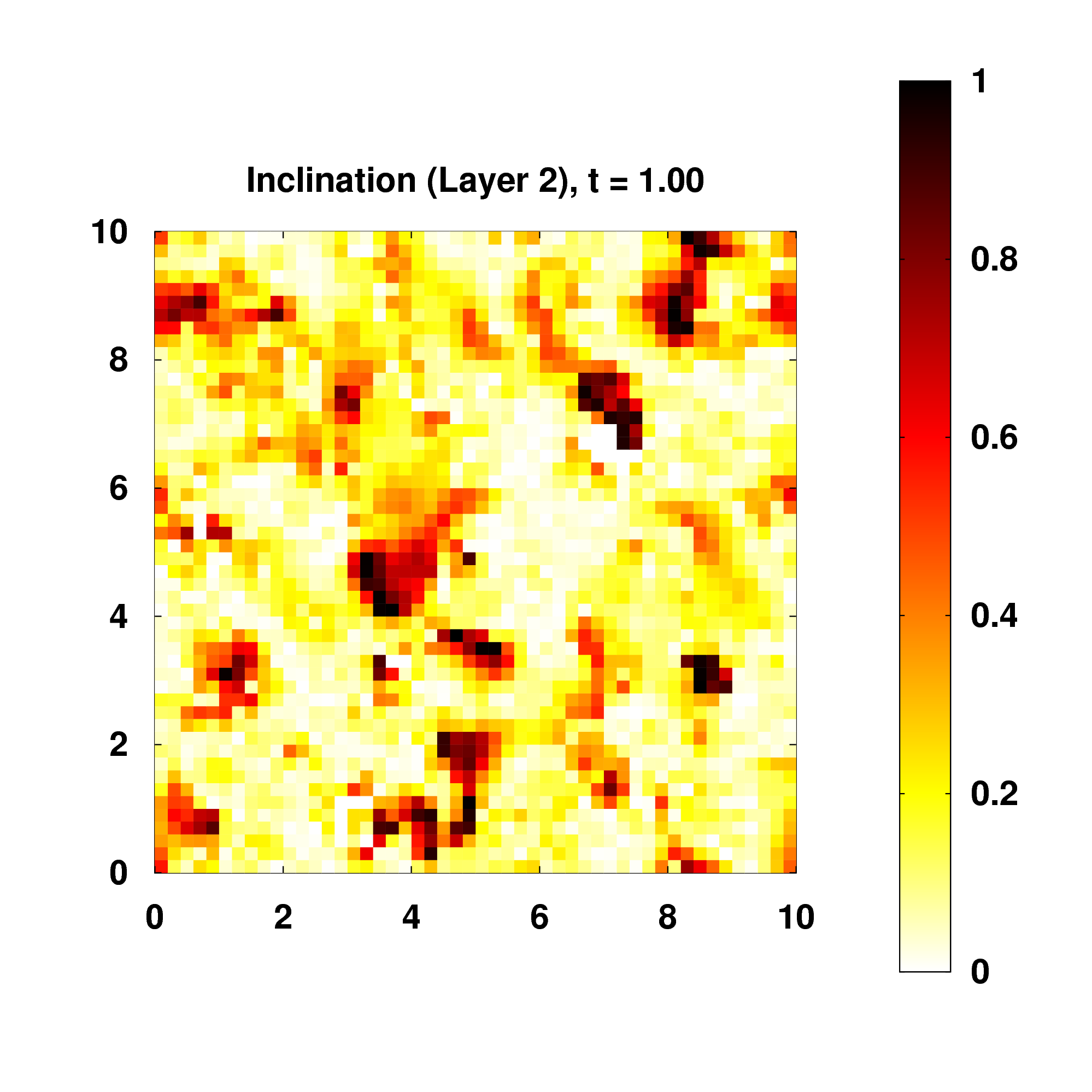}\includegraphics[width=0.45\textwidth]{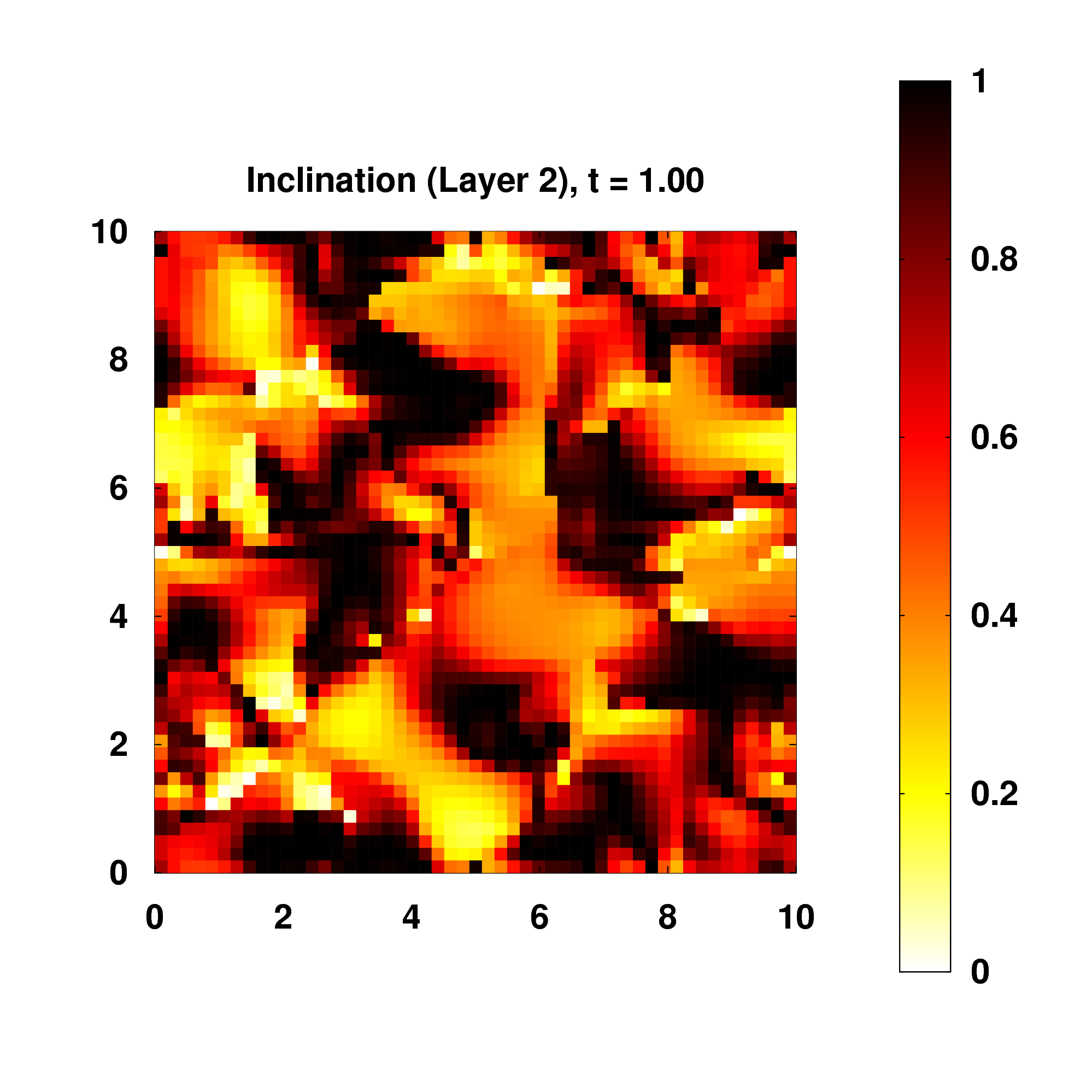}
\includegraphics[width=0.45\textwidth]{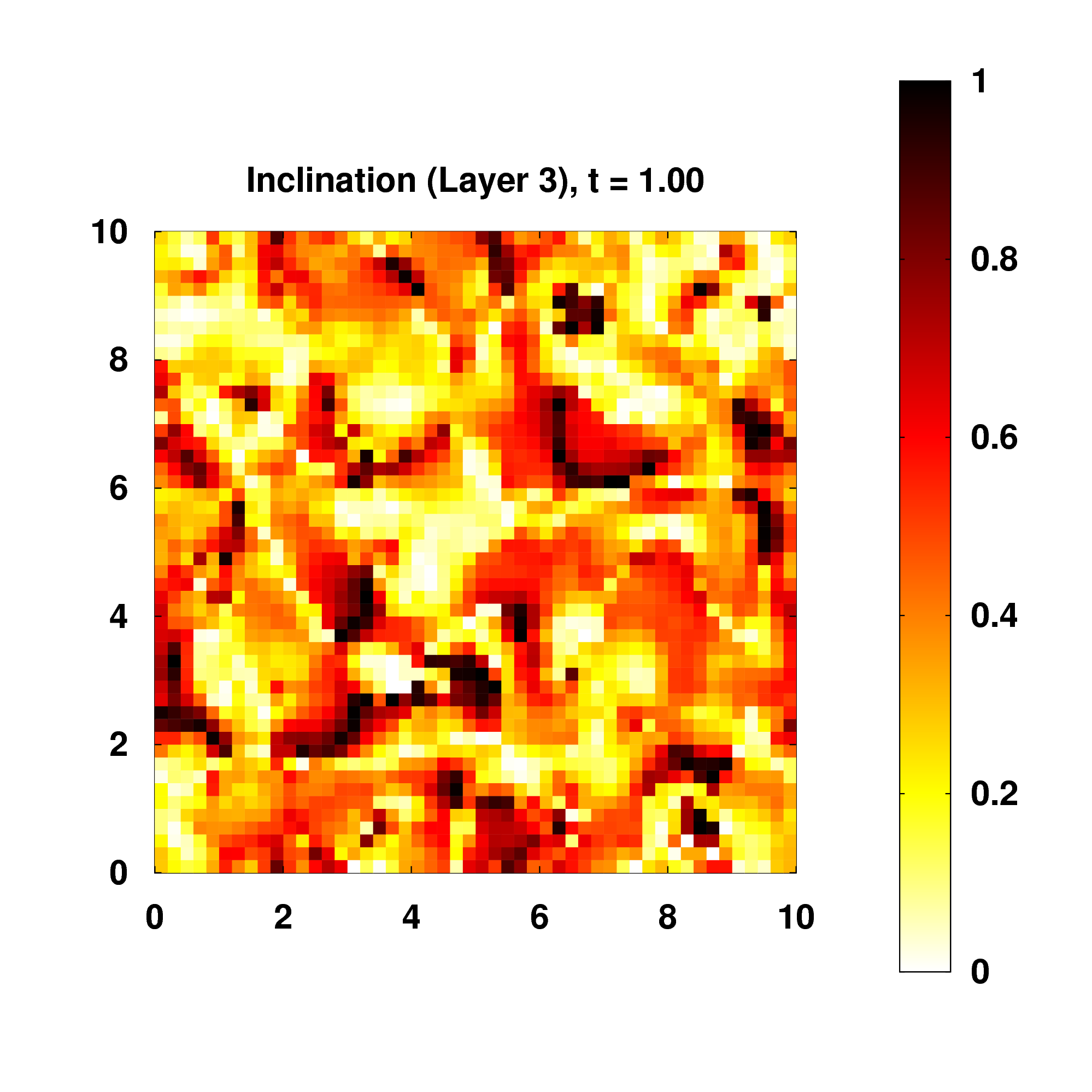}\includegraphics[width=0.45\textwidth]{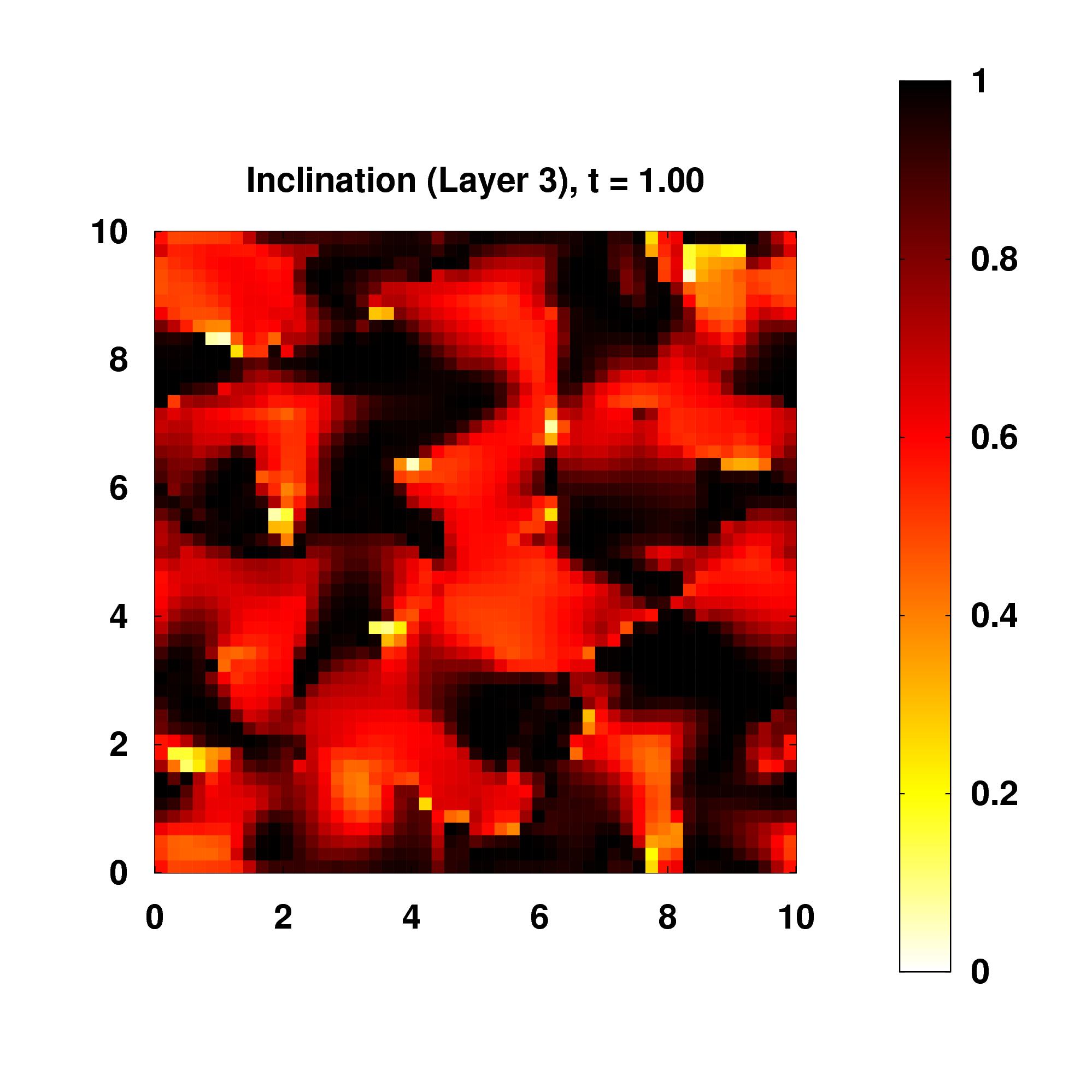}

\caption{(Inhomogeneous case, interacting layers, inclination) Left: Cosine of the inclination angle (in absolute value) for the particle simulation with $\eps = 0.1$ (averaged on $20$ simulations). Number of particles: $10^5$ per layer. $\Delta t = 1\times 10^{-2}$. Right: Cosine of the inclination angle (in absolute value) for the macroscopic simulation. $\Delta t = 1\times 10^{-2}$, $\Delta x = \Delta y = 0.2$. Alignment interaction parameters: $\nu = 4$, $D = 0.1$, $K=4$, $\delta = 0.1$. Layer-interaction parameters: $h = R = 0.02$, $\beta= 2$, $\mu=20$, $R_1 = R_2 = R_3 = R$.} 
\label{Fig:interact-layers-TG-inclination}
\end{figure}

\section{Conclusion and discussion}
\label{sec:conclu}

In this article, we have proposed an individual-based model of self-propelled disk-like particles interacting through alignment and volume exclusion. This model is intended to provide a framework for modeling collective sperm-cell dynamics. Particle motion is supposed confined in two-dimensional planar layers. Particle interactions between nearby layers contribute to modify the disk inclinations, which generates a coupling between inclinations and motion. We have then derived a continuum model from this individual-based model. It describes the evolution of the local density, mean velocity direction and mean inclination of the disks in the various layers. Numerical simulations have shown a good agreement between the continuum model and the individual-based one, but has also highlighted some differences. 

There are many possible directions to expand the current work and make it more realistic. At the individual-based level, the description of the agents and the interaction rules could be improved for a better account of actual sperm-cell motion. For instance, the assumption that all sperm-cells have the same constant velocity is obviously unrealistic. In any semen sample, there is always a certain proportion of dead sperm-cells and of less motile ones. This could be accounted for by allowing the particle velocities to span a certain range of values. The shape of the head could be improved from the current infinitely thin disk to finite thickness ellipsoids. The inclination interaction could involve a density dependency as it is more difficult to fit actual disks in one layer if they are inclined towards the plane than if they stand vertically. Finally, one could also imagine a process by which particles would change layers. At the level of the continuum model, one major improvement should be to add a random fluctuation term in order to account for finite system size effects, similar to Ref. \cite{Toner_etal_AnnPhys05}. We believe that adding such a term would help achieve a better match between the continuum model and the individual-based one. Other improvements would consist in adding a spatial diffusion to retain some of the nonlocality of the alignment interaction, similar to Ref. \cite{2013_DegondLiu} or adding a layer-changing term. Finally, for both the individual-based and continuum models, the model parameters should be calibrated by close comparisons with biological data.

\appendix
\section{Nematic alignment}
\label{appendix:nematic}

Nematic alignment consists in alignment with the mean direction. The mean direction can be defined as the eigenspace of the averaged projection matrix: 
\begin{equation}
\int_0^\pi \begin{pmatrix}
\cos \theta\\
\sin \theta
\end{pmatrix}\otimes\begin{pmatrix}
\cos \theta\\
\sin \theta
\end{pmatrix}  f(\theta)\, d\theta,
\end{equation}
corresponding to the largest eigenvalue. The eigenvector $(\cos\bar\theta,\sin\bar\theta)^T$ satisfies:
\begin{equation}
\left[\int_0^\pi \begin{pmatrix}
\cos^2 \theta&\cos\theta\sin\theta\\
\cos\theta\sin\theta&\sin^2\theta
\end{pmatrix}  f(\theta)\, d\theta\right] \begin{pmatrix}
\cos \bar\theta\\
\sin \bar\theta
\end{pmatrix} \times  \begin{pmatrix}
\cos \bar\theta\\
\sin \bar\theta
\end{pmatrix} = 0
\end{equation}
Easy computations lead to the following relation:
\begin{equation}
\int_0^\pi 
\sin(2(\bar\theta - \theta))\, f(\theta)\, d\theta   = 0,
\end{equation}
that can be written also as follows:
\begin{equation}
\int_0^\pi  \begin{pmatrix}
\cos 2\theta\\
\sin 2\theta
\end{pmatrix}   f d\theta   \times \begin{pmatrix}
\cos 2\bar\theta\\
\sin 2\bar\theta
\end{pmatrix} = 0.
\end{equation}
The mean direction then satisfies:
\begin{equation}
\begin{pmatrix}
\cos 2\bar\theta\\
\sin 2\bar\theta
\end{pmatrix} = \pm\frac{\displaystyle\int_0^\pi  \begin{pmatrix}
\cos 2\theta\\
\sin 2\theta
\end{pmatrix}   f d\theta}{\left|\displaystyle\int_0^\pi  \begin{pmatrix}
\cos 2\theta\\
\sin 2\theta
\end{pmatrix}   f d\theta\right|}
\end{equation}
Then $\bar\theta$ is defined modulo $\pi/2$. Therefore, the orthogonal vectors $(\cos\bar\theta,\sin\bar\theta)^T$ and $(\cos\bar(\theta+\pi/2),\sin(\bar\theta+\pi/2))^T$ are both eigenvectors. Using the relation
\begin{equation}
 \displaystyle\int_0^\pi  \begin{pmatrix}
\cos 2\theta\\
\sin 2\theta
\end{pmatrix}   f d\theta = \pm\, C \begin{pmatrix}
\cos 2\bar\theta\\
\sin 2\bar\theta
\end{pmatrix},\quad \text{ where }C = \left|\displaystyle\int_0^\pi  \begin{pmatrix}
\cos 2\theta\\
\sin 2\theta
\end{pmatrix}   f d\theta\right|,
\end{equation}
the averaged projection matrix writes:
\begin{equation}
 \frac{1}{2}\begin{pmatrix}
1\pm\, C \cos 2 \bar\theta&\pm\, C \sin2\bar\theta\\
\pm\, C \sin2\bar\theta&1\mp\, C \cos 2 \bar\theta
\end{pmatrix},
\end{equation}
whose eigenvalues are given by $(1\pm\, C)/2$. Consequently, the eigenvector $(\cos\bar\theta,\sin\bar\theta)^T$ corresponding to the largest eigenvalue involves the angle $\bar\theta$ satisfying:
\begin{equation}
\begin{pmatrix}
\cos 2\bar\theta\\
\sin 2\bar\theta
\end{pmatrix} = +\frac{\displaystyle\int_0^\pi  \begin{pmatrix}
\cos 2\theta\\
\sin 2\theta
\end{pmatrix}   f d\theta}{\left|\displaystyle\int_0^\pi  \begin{pmatrix}
\cos 2\theta\\
\sin 2\theta
\end{pmatrix}   f d\theta\right|}.
\end{equation}
$\bar\theta$ is now defined modulo $\pi$. We recover the definition of equation \ref{eq:nematic_alignment}.

\section{Coefficients of the macroscopic model}
\label{annex:coeff_macro}

From Ref.~\cite{2008_ContinuumLimit_DM}, coefficients $c_1$ and $c_2$ are positive and bounded. We have the following expansion (in the limit $\kappa_1 \rightarrow 0$, that corresponds to large diffusion compared to alignment):
\begin{equation*}
c_1 = \frac{\kappa_1}{2} + O(\kappa_1^2),\quad c_2 = \frac{3\kappa_1}{16} + O(\kappa_1^2),
\end{equation*}
The following proposition asserts similar results for coefficients $c_3$ and $c_4$.
\begin{proposition} Constants $c_3$ and $c_4$ can be written:
\begin{equation*}
c_3 = 1 - \frac{\pi^2}{\left(\int_0^\pi e^{\kappa_1 \cos u du}\right)\left(\int_0^\pi e^{-\kappa_1 \cos u du}\right)},\quad c_4 = 1 - \frac{\pi^2}{\left(\int_0^\pi e^{\kappa_2 \cos 2u du}\right)\left(\int_0^\pi e^{-\kappa_2 \cos 2u du}\right)}.
\end{equation*}
In particular, they are positive and lower than $1$. We have the following Taylor expansions, as $\kappa_1 \rightarrow 0$ and $\kappa_2 \rightarrow 0$:
\begin{equation*}
c_3 = \frac{\kappa_1^2}{2} + O(\kappa_1^2),\quad c_4 = \frac{\kappa_2^2}{2} + O(\kappa_2^2).
\end{equation*}
\end{proposition}
\begin{proof} 
By integration by part:
\begin{equation*}
c_3 = \kappa_1 \int_{\varphi \in [0,2\pi]} \partial_\varphi M_1 I_1\, d\varphi = - \kappa_1 \int_{\varphi \in [0,2\pi]} M_1 \partial_\varphi I_1\, d\varphi
\end{equation*}
Then, using expression \eqref{eq:I1}, we easily get the expected expression for $c_3$. The same manipulations lead to the expression for $c_4$. The positivity of $c_3$ and $c_4$ then results from the Cauchy-Schwarz inequality. 
\end{proof}
\noindent Moreover the weight functions $\langle gM_2M_2\rangle$ and $\langle gM_2M_2\partial_\theta I_2\rangle$  are also bounded.
\begin{proposition}We have:
\begin{equation*}
0 \leqslant \langle gM_2M_2\rangle(\bar\theta,\bar\theta_k) \leqslant 1,\quad \left|\langle gM_2M_2\partial_\theta I_2\rangle(\bar\theta,\bar\theta_k)\right| \leqslant 2.
\end{equation*}
We have the following Taylor expansions:
\begin{align*}
&\langle gM_2M_2\rangle(\bar\theta,\bar\theta_k) = \int_{\theta,\theta' \in [0,\pi]} g(\bar\theta+\theta,\bar\theta_k+\theta') \, \frac{d\theta}{\pi} \frac{d\theta'}{\pi} + O(\kappa_2),\\
&\langle gM_2M_2\partial_\theta I_2\rangle(\bar\theta,\bar\theta_k) = - \kappa_2 \int_{\theta,\theta' \in [0,\pi]} g(\bar\theta+\theta,\bar\theta_k+\theta') \cos 2\theta  \, \frac{d\theta}{\pi} \frac{d\theta'}{\pi} + O(\kappa_2^2).
\end{align*}
\end{proposition}

\begin{proof} The inequalities for $\langle gM_2M_2\rangle$ results from the bounds $0 \leqslant g \leqslant 1$. Using expression \eqref{eq:I2}, we have :
\begin{align*}
&\langle gM_2M_2\partial_\theta I_2\rangle(\bar\theta,\bar\theta_k) \\
&=   \int_{\theta,\theta' \in [0,\pi]} g(\bar\theta+\theta,\bar\theta_k+\theta')\Big(-M_2(\theta) + \\
&\hspace{4cm} \frac{\pi}{\left(\int_0^\pi e^{\kappa_2 \cos 2u du}\right)\left(\int_0^\pi e^{-\kappa_2 \cos 2u du}\right)}\Big) M_2(\theta')\, d\theta d\theta'\\
&= - \langle gM_2M_2\rangle(\bar\theta,\bar\theta_k)\\
&\qquad + \frac{\pi^2}{\left(\int_0^\pi e^{\kappa_2 \cos 2u du}\right)\left(\int_0^\pi e^{-\kappa_2 \cos 2u du}\right)}\int_{\theta,\theta' \in [0,\pi]} g(\bar\theta+\theta,\bar\theta_k+\theta') M_2(\theta')\, \frac{d\theta}{\pi} d\theta'.
\end{align*}
In the last expression, both terms have absolute value lower than $1$.
\end{proof}

\section{Numerical schemes}
\label{annex:num}

For the sake of completeness, we here recall the numerical scheme used for the numerical simulations.

\subsection{Microscopic equations.} The deterministic part of equations \eqref{eq:phi} and \eqref{eq:thetadyn} can be written:
\begin{align*}
&dV(\varphi_k)  = P_{V(\varphi_k)^\perp}\left[\nu (V(\bar \varphi_k^{\text{tot}}) - V(\varphi_k)) \right] dt, \\
&de^{2i\theta_k} = (ie^{2i\theta_k})\cdot\left(2K (e^{2i\bar\theta_k} - e^{2i\theta_k}) + 2T_k (ie^{2i\theta_k})\right) ie^{2i\theta_k} dt, 
\end{align*}
Following \cite{2011_MotschNavoret} (Annex C), we then use the scheme:
\begin{align*}
&\varphi_k^{n+1} = \varphi_k^n + 2 \left[\widehat{ V(\varphi_k^n) }- \widehat B\right] + \sqrt{2D\Delta t}\, \eps_k^n,  \quad \mbox{ modulo } 2\pi,\\
&\hspace{5cm} \text{with }B = V(\varphi_k^n) + \Delta t\, \frac{\nu (V(\bar \varphi_k^{\text{tot}\, n}) - V(\varphi_k^n))}{2},\\
&\theta_k^{n+1} = \theta_k^n + \left[2\theta_k^n - \arg(C)\right] + \sqrt{\delta\Delta t}\, \tilde\eps_k^n, \quad \mbox{ modulo } \pi,\\
&\hspace{5cm} \text{with }C = e^{2i\bar\theta_k^n} + \Delta t\, \frac{2K (e^{2i\bar\theta_k^n} - e^{2i\theta_k^n}) + 2T_k^n (ie^{2i\theta_k^n})}{2},
\end{align*}
where $\varepsilon_k^n$ and $\tilde\varepsilon_k^n$ are random variables with standard normal distribution, $\arg(z) \in (0,2\pi)$ denotes the argument of the complex number $z \in \C$ and we recall that $\hat v \in (0,2\pi)$ denotes the angle between the vector $v \in \R^2$ and the vector $(1,0)^T$.

\subsection{Macroscopic equations.} Multiplying equations \eqref{eq:phi_final} and \eqref{eq:theta_final} respectively by $V(\bar\varphi)^\perp$ and $2ie^{2i\bar\theta}$, we get:
\begin{align}
&\rho \big( \partial_t V(\bar \varphi) +  c c_2 (V(\bar \varphi) \cdot\nabla_x) V(\bar \varphi) \big) + P_{V(\bar\varphi)^\perp}\left[\frac{c}{\kappa_1} \nabla_x \rho\right]\nonumber\\
&\ =  P_{V(\bar\varphi)^\perp}\left[\frac{\nu\beta' }{c_3}\, \sum_{k,\, k-h = \pm 1}  \langle gM_2M_2\rangle(\bar\theta,\bar\theta_k)\, c_1 \rho_k V(\bar\varphi_k)\right],
\label{eq:phi_final2} \\
&\rho( \partial_t e^{2i\bar\theta} +  c c_1 (V(\bar \varphi) \cdot\nabla_x) e^{2i\bar\theta} ) =\nonumber\\
&\ = \left[\frac{1}{c_4}\, \, (c_1  \rho  V(\bar \varphi))^\perp\cdot \sum_{k,\, k-h = \pm 1} \sgn(k-h) \langle gM_2M_2\partial_\theta I_2\rangle(\bar\theta,\bar\theta_k)\,  c_1 \rho_k V(\bar \varphi_k)\right] 2ie^{2i\bar\theta}.
\label{eq:theta_final2}
\end{align}
To solve this system, we introduce a relaxation model
\begin{align}
&\partial_t \rho + c c_1 \, \nabla_x \cdot \left( \rho\, v\right)= 0, 
\label{eq:mass_final3} \\
&\partial_t \rho v +  c c_2 \nabla_x\cdot(\rho\, v\otimes v\big) + \frac{c}{\kappa_1} \nabla_x \rho\nonumber\\
&\qquad\qquad=  \frac{\nu\beta' }{c_3}\, \sum_{k,\, k-h = \pm 1}  \langle gM_2M_2\rangle(\theta(u),\theta(u_k))\, c_1 \rho_k v_k - \frac{1}{\eta} (1- |v|^2)v,
\label{eq:phi_final3} \\
& \partial_t \rho u +  c c_1 \nabla_x \cdot \left( \rho\, v\otimes u\right)  =  \frac{2}{c_4}\, \, (c_1  \rho\, v)^\perp\cdot N\ u^{\perp} - \frac{1}{\eta} (1- |u|^2)u,
\label{eq:theta_final3}
\intertext{with }
&N = \mu'\,   \sum_{k,\, k-h = \pm 1} \sgn(k-h) \langle gM_2M_2\partial_\theta I_2\rangle(\theta(u),\theta(u_k))\,  c_1 \rho_k v_k,
\end{align}
and where $\theta(u) \in [0,2\pi[$ denotes the angle of the vector $u \in \R^2$.  In the limit $\eta \rightarrow 0$, the solution $(\rho, v, u)$ to \eqref{eq:mass_final3}-\eqref{eq:phi_final3}-\eqref{eq:theta_final3} formally converges to $(\rho, V(\bar\varphi), (\cos 2\bar\theta,\sin 2\bar\theta)^T)$, solution of equations \eqref{eq:mass_final} and \eqref{eq:phi_final2}-\eqref{eq:theta_final2}.

Equation \eqref{eq:mass_final3}-\eqref{eq:phi_final3}-\eqref{eq:theta_final3} is numerically solved using a splitting method. We first solve the conservative part (with a Roe-like method\cite{1999_Degond_PolynomialUpwind}), we then add the source term and we finally solve the relaxation part. For the last step, we just perform a renormalization of the vectors. This kind of scheme has been validated. In particular, it captures the correct discontinuous solutions of the macroscopic model corresponding to the solutions of the microscopic simulations. For more details, we refer to Ref. \cite{2011_MotschNavoret}. 

The characteristic velocities of the conservative system in the $x$ direction are:
 \begin{align*}
&\gamma_1 = c \left( c_2 v_x + \sqrt{c_1/\kappa_1 + v_x^2 c_2(c_2 - c_1)}\right),\ \gamma_2= c c_2 v_x,\\
& \gamma_3 = c \left( c_2 v_x - \sqrt{c_1/\kappa_1 + v_x^2 c_2(c_2 - c_1)}\right),\ \gamma_4 = c c_1 v_x,\quad \gamma_5 = c c_1 v_x.
 \end{align*} 
 where $v_x$ denotes the first component of the vector $v$. As noticed in Ref.~\cite{2013_DegondHua}, the conservative part of the equation is hyperbolic under the condition: 
 \begin{equation*}
 c_1/\kappa_1 + v_x^2 c_2(c_2 - c_1) \geq 0,
 \end{equation*}
That is the case for the parameters chosen in section \ref{section:num_experiment}.
To ensure the stability of the conservative step, the time and space steps $\Delta t$ and $\Delta x$ have to satisfy the CFL condition: $$\underset{1\leq i \leq 5}{\max} |\gamma_i|\, \Delta t \leq \Delta x.$$
To ensure the stability of the source term step, we choose $\Delta t$ small enough to ensure:
\begin{align*}
&\Delta t \left\|\frac{\nu\beta }{c_3}\, \sum_{k,\, k-h = \pm 1}  \langle gM_2M_2\rangle(\theta(u),\theta(u_k))\, c_1 \rho_k v_k \right\| \leq 2,\\
&\Delta t \left\|  \frac{2}{c_4}\, \, (c_1  \rho\, v)^\perp\cdot N\ u^{\perp}\right\|\leq 2,
\end{align*}
at each time step.

\section*{Acknowledgment}
This work has been supported by the Agence Nationale pour la Recherche (ANR) in the framework of the contract MOTIMO (ANR-11-MONU-009-01) and by the National Science Foundation under grant RNMS11-07444 (KI-Net). PD is on leave from CNRS, Institut de Math\'ematiques de Toulouse, France. PD acknowledges support from the Royal Society and the Wolfson foundation through a Royal Society Wolfson Research Merit Award and from the Engineering and Physical Sciences Research Council (EPSRC) under grant ref: EP/M006883/1.


\begin{thebibliography}{00}

\bibitem{Aoki_BullJapSocSciFish92} I. Aoki, A simulation study on the schooling mechanism in fish, {\it Bulletin of the Japan Society of Scientific Fisheries} {\bf 48} (1982) 1081--1088.

\bibitem{Baskaran_Marchetti_PRL10}
A. Baskaran and M. C. Marchetti, Nonequilibrium statistical mechanics of self-propelled hard rods, {\it J. Stat. Mech. Theory Exp.} (2010) P04019.

\bibitem{Bertin_etal_JPA09}
E.~Bertin, M.~Droz and G.~Gr\'egoire, Hydrodynamic equations for self-propelled particles: microscopic derivation and stability analysis, {\it J. Phys. A: Math. Theor.} {\bf 42} (2009) 445001.

\bibitem{2012_Bolley} F. Bolley, J.A. Canizo and J.A. Carrillo, Mean-field limit for the stochastic Vicsek model, {\it Appl. Math. Lett.} {\bf 25} (2012) 339--343.

\bibitem{Chate_etal_PRE08}
H. Chat\'e, F. Ginelli, G. Gr\'egoire and F. Raynaud, Collective motion of self-propelled particles interacting without cohesion, {\it Phys. Rev. E} {\bf 77} (2008) 046113.

\bibitem{Chuang_etal_PhysicaD07}
Y-L. Chuang, M. R. D'Orsogna, D. Marthaler,  A. L. Bertozzi and L. S. Chayes, State transitions and the continuum limit for a 2D interacting, self-propelled particle system, {\it Physica D} {\bf 232} (2007) 33-47.

\bibitem{Couzin_etal_JTB02}
I. D. Couzin, J. Krause, R. James, G. D. Ruxton and N. R. Franks, Collective Memory and Spatial Sorting in Animal Groups, {\it J. theor. Biol.}, {\bf 218} (2002), 1-11.

\bibitem{Cucker_Smale_IEEETransAutCont07}
F. Cucker and S. Smale, Emergent behavior in flocks, {\it IEEE Transactions on Automatic Control} {\bf 52} (2007) 852-862.

\bibitem{2015_Degond} P. Degond, G. Dimarco, T.B.N. Mac, and N. Wang, Macroscopic models of collective motion with repulsion, {\it Comm. Math. Sci.} (2015),  to appear.

\bibitem{2013_DegondHua} P. Degond and J. Hua, Self-organized hydrodynamics with congestion and path formation in crowds, {\it J. Comp. Phys.} {\bf 237} (2013) 299--319.

\bibitem{2013_DegondLiu} P. Degond, J-G. Liu, S. Motsch and V. Panferov, Hydrodynamic models of self-organized dynamics: derivation and existence theory, {\it Methods Appl. Anal.} {\bf 20} (2013) 089-114.

\bibitem{2008_ContinuumLimit_DM} P. Degond and S. Motsch, Continuum limit of self-driven particles with orientation interaction, {\it Math. Models Methods Appl. Sci.} {\bf 18} (2008) 1193--1215.

\bibitem{1999_Degond_PolynomialUpwind} P. Degond, P.-F. Peyrard, G. Russo, and P. Villedieu, Polynomial upwind schemes for hyperbolic systems, {\it CR Acad. Sci. Paris} {\bf 328} (1999) 479--483.

\bibitem{2004_Doare_PlantInteract} O. Doar\'e, B. Moulia and E. De Langre, Effect of plant interaction on wind-induced crop motion, {\it J. Biomech. Eng.} {\bf 126} (2004) 126--146.

\bibitem{Fetecau_M3AS11} R. Fetecau, Collective behavior of biological aggregations in two dimensions: a nonlocal kinetic model, {\it Math. Models Methods Appl. Sci.} {\bf 21} (2011) 1539.

\bibitem{2012_Frouvelle} A. Frouvelle, A continuum model for alignment of self-propelled particles with anisotropy and density-dependent parameters, {\it Math. Models Methods Appl. Sci.} {\bf 22} (2012) 1250011.

\bibitem{Henkes_etal_PRE11}
S. Henkes, Y. Fily and M. C. Marchetti, Active jamming: Self-propelled soft particles at high density, {\it Phys. Rev. E} {\bf 84} (2011) 040301.


\bibitem{2011_Koch} D. L. Koch and G. Subramanian, Collective hydrodynamics of swimming microorganisms: Living fluids, {\it Annu. Rev. Fluid Mech.} {\bf 43} (2011).

\bibitem{Mogilner_etal_JMB03}
A. Mogilner, L. Edelstein-Keshet, L. Bent and A. Spiros, Mutual interactions, potentials, and individual distance in a social aggregation, {\it J. Math. Biol.} {\bf 47} (2003) 353-389.

\bibitem{2011_MotschNavoret} S. Motsch and L. Navoret, Numerical simulations of a nonconservative hyperbolic system with geometric constraints describing swarming behavior, {\it Multiscale Model. Simul.} {\bf 9} (2011) 1253-1275.

\bibitem{Motsch_Tadmor_JSP11}
S. Motsch and E. Tadmor, A new model for self-organized dynamics and its flocking behavior, {\it J. Stat. Phys.} {\bf 144} (2011) 923-947.

\bibitem{2013_Navoret} L. Navoret, A two-species hydrodynamic model of particles interacting through self-alignment, {\it Math. Models Methods Appl. Sci.} {\bf 23} (2013) 1067--1098.

\bibitem{2006_Peruani} F. Peruani, A. Deutsch and M. B\"ar, Nonequilibrium clustering of self-propelled rods, {\it Phys. Rev. E } {\bf 74} (2006) 030904(R).

\bibitem{Plouaboue_2015} F. Plourabou\'e, Personal communication. 

\bibitem{Ratushnaya_etal_PhysicaA07}
V. I. Ratushnaya, D. Bedeaux, V. L. Kulinskii and A. V. Zvelindovsky, Collective behavior of self propelling particles with kinematic constraints: the relations between the discrete and the continuous description, {\it Phys. A} {\bf 381} (2007) 39-46.

\bibitem{Szabo_PRL06} B.Szab\'o, G.J Sz\"oll\"osi, B. G\"onci, Zs. Jur\'anyi, D. Selmeczi and T. Vicsek,  Phase transition in the collective migration of tissue cells: Experiment and model, {\it Phys. Rev. Lett} {\bf 74} (2006) 061908.

\bibitem{Toner_etal_AnnPhys05} J. Toner, Y. Tu and S. Ramaswamy, Hydrodynamics and phases of flocks, {\it Annals of Physics} {\bf 318} (2005) 170-244.

\bibitem{1995_Vicsek} T. Vicsek, A. Czir\'ok, E. Ben-Jacob, I. Cohen and O. Shochet, Novel type of phase transition in a system of self-driven particles, {\it Phys. Rev. Lett.} {\bf 75} (1995) 1226--1229.

\bibitem{2012_Vicsek} T. Vicsek, A. Zafeiris, Collective motion, {\it Phys. Rep.} {\bf 517} (2012) 71-140.

\end{thebibliography}
\end{document}